\documentclass[12pt]{article}

\usepackage{natbib}
\usepackage[nottoc,notlof,notlot]{tocbibind} 

\bibpunct{(}{)}{;}{a}{}{;}
\usepackage{breakcites}
\usepackage[breaklinks=true]{hyperref}
\usepackage{graphicx}
\graphicspath{ {images/} }

\usepackage{amsmath, amsthm}
\usepackage{bbm}
\usepackage{amsfonts}%
\usepackage{amssymb}%

\usepackage{geometry} 
\usepackage{graphicx} 
\usepackage{framed}
\usepackage{caption}
\usepackage{subcaption}
\usepackage{pgfplots}
\usepackage{adjustbox}

\usepackage{tikz}
\usetikzlibrary{shapes.geometric} 

\usepackage{booktabs} 
\usepackage{multirow} 
\usepackage{array} 
\usepackage{paralist} 
\usepackage{verbatim} 

\usepackage{lipsum}
\usepackage{setspace}

\usepackage{lipsum}
\usepackage[singlelinecheck=false]{caption}
\usepackage{tabularx}

\usepackage{tikz}
\usepackage[justification=centering]{caption}
\newtheorem{theorem}{Theorem}[section]
\newtheorem{lemma}{Lemma}[section]

\newtheorem{corollary}{Corollary}[section]
\newtheorem{proposition}{Proposition}[section]

\newtheorem{assumption}{Assumption}[section]
\theoremstyle{definition}
\newtheorem{example}{Example}[section]
\newtheorem{remark}{Remark}[section]
\newtheorem{algorithm}{Algorithm}[section]
\title{\setstretch{1} The Local Approach to Causal Inference\\ under Network Interference}

\author{Eric Auerbach\footnote{Department of Economics, Northwestern University. E-mail: eric.auerbach@northwestern.edu.} \and Hongchang Guo\footnote{Department of Economics, Northwestern University. E-mail: hongchangguo2028@u.northwestern.edu.} \and Max Tabord-Meehan\footnote{Department of Economics, University of Chicago. E-mail: maxtm@uchicago.edu. \newline We thank Tim Armstrong, Stephane Bonhomme, Ivan Canay, Ben Golub, Joel Horowitz, Chuck Manski, Roger Moon, Seth Richards-Shubik, Azeem Shaikh, Chris Taber, Alex Torgovitsky and seminar participants at Cornell, JSM, Oxford, UCL, UC Boulder, U Rochester, U Glasgow, Laval, UW Madison, USC, Northwestern, U Chicago, TSE for helpful discussions. Research supported by NSF grants SES-2149408 and SES-2149422. This paper supersedes the earlier working paper \emph{A Nonparametric Network Regression}.}
}

\begin{document}
\maketitle
\begin{abstract} \setstretch{1}\noindent
We propose a new nonparametric modeling framework for causal inference when outcomes depend on how agents are linked in a social or economic network. Such network interference describes a large literature on treatment spillovers, social interactions, social learning, information diffusion, disease and financial contagion, social capital formation, and more. Our approach works by first characterizing how an agent is linked in the network using the configuration of other agents and connections nearby as measured by path distance. The impact of a policy or treatment assignment is then learned by pooling outcome data across similarly configured agents. We demonstrate the approach by deriving finite-sample bounds on the mean-squared error of a k-nearest-neighbor estimator for the average treatment response as well as proposing an asymptotically valid test for the hypothesis of policy irrelevance.
\end{abstract}

\section{Introduction}\label{sec:introduction}
Economists are often tasked with predicting outcomes under a counterfactual policy or treatment assignment. In many cases, the counterfactual depends on how the agents are linked in a social or economic network. A diverse literature on treatment spillovers, social interactions, social learning, information diffusion, disease and financial contagion, social capital formation, and more approaches this problem from a variety of specialized, often highly-parametric frameworks \citep[see broadly][]{graham2011econometric,blume2010identification,de2016econometrics,athey2017state,jackson2017economic,bramoulle2020peer}. In this paper, we propose a new framework for causal inference that accommodates many such examples of network interference. 

Our main innovation is a nonparametric modeling approach for sparse network data based on local configurations. Informally, a local configuration refers to the features of the network (the agents, their characteristics, treatment statuses, and how they are connected) nearby a focal agent as measured by path distance. The idea is that these local configurations index the different ways in which a policy or treatment assignment can impact the focal agent's outcome under network interference.  

This approach generalizes a developed literature on spillovers and social interactions in which the researcher specifies reference groups or an exposure map that details exactly how agents influence each other \citep[see for instance][]{manski1993identification,manski2013identification,hudgens2008toward,aronow2017estimating,vazquez2017identification,leung2019causal,arduini2020treatment,savje2021causal,qu2021efficient}. One limitation of this literature is that the effect of a policy or treatment assignment is generally sensitive to how the researcher models the dependence. For example, in the spillovers literature it is often assumed that agents respond to the average treatment of their peers, while in the diffusion literature agents may be informed or infected by any peer. When the researcher is uncertain as to exactly how agents influence each other, misspecification can lead to inaccurate estimates and invalid inferences about the impact of the policy or treatment assignment of interest.\footnote{The frameworks of \cite{leung2019causal,savje2021causal} are misspecification-robust in the sense that their inference procedures may still be valid even if the model of interference used to define the estimand is misspecified. Their procedures do not address the issue we highlight here that an estimand based on a misspecified model may fail to accurately describe the impact of the policy of interest. We discuss this issue in more detail in our note \cite{auerbach2024discussion}.}

Another limitation of this literature is that it does not generally consider policies that change the structure of the network. Network-altering policies are common in  economics. Examples include those that add or remove agents, or connections between agents, from the community \cite[see for instance][]{ballester2006s,azoulay2010superstar, donaldson2016railroads,lee2021key}. Such policies may be difficult to evaluate using standard frameworks, which typically focus on the reassignment of treatment to agents keeping the network structure fixed.\footnote{An example of a policy where an agents is removed from a social network is incarceration, where an agent is forcibly removed from society and put in jail.}

Our methodology addresses these limitations by using local configurations to model network interference. Intuitively, we use the space of local configurations as a ``network sieve'' that indexes the distribution of agent-specific outcomes associated with a given policy or treatment assignment. A contribution of our work is to formalize this local approach and apply it to causal inference under network interference. 

The use of local configurations in economics was proposed by \cite{de2018identifying,anderson2019collaborative} and related to that of ego-centered networks in sociology \citep[see generally][]{wasserman1994social}. In their work, local configurations index moment conditions that partially identify the parameters of a strategic network formation model. The researcher has the flexibility to choose the configurations used for this task and can restrict attention to a small number that occur frequently in the data. In our setting, local configurations correspond to fixed counterfactual policies. It is usually the case that no exact instances of a given policy appear in the data and so we substitute outcomes associated with similar but not exactly the same configurations. Formalizing this procedure requires new machinery, which we introduce building on work by \cite{benjamini2001recurrence}. 

We demonstrate our local approach with applications to two causal inference problems. In both problems the researcher starts with a status-quo policy as described by one local configuration and is tasked with evaluating the impact of an alternative policy as described by another local configuration. The researcher has access to data from multiple-networks corresponding to a partial or stratified interference setting. Such designs are common in education, industrial organization, labor, and development economics where the researcher may collect network data on multiple independent schools, markets, firms, or villages.

The first problem is to estimate average or distributional policy effects/treatment response. For instance, the status-quo policy may be given by a particular network structure and the new policy may be one in which a key agent is removed. The policy effect to be estimated is the expected change in outcome for one or more agents. We propose a $k$-nearest-neighbors estimator for the policy effect and provide non-asymptotic bounds on mean-squared error building on work by \cite{doring2017rate}. We also provide sufficient conditions for the estimator to be asymptotically normal, which can be used to construct confidence intervals for the policy effect in the usual way.

The second problem is to test policy irrelevance/no treatment effects. For instance, the status-quo policy may be given by a particular network structure where no agents are treated and the new policy may keep the same connections between agents but have every agent treated. The hypothesis to be tested is that both policies are associated with the same distribution of outcomes for one or more agents. We propose an asymptotically valid randomization test for this hypothesis building on work by \cite{canay2018approximate}. 


Section \ref{sec:local} defines a local configuration. Section \ref{sec:motivation} incorporates our definition of a local configuration into an econometric model with network interference. Section \ref{sec:policy} describes applications to estimating policy effects and testing policy irrelevance. Section \ref{sec:empirical_illustration} contains an empirical illustration evaluating the impact of social network structure on favor exchange in the setting of \cite{jackson2012social}. Section \ref{sec:conclusion} concludes. Proof of claims, simulation evidence, and other details can be found in an appendix. 



\section{Local configurations}\label{sec:local}
In this section, we first provide an informal description of a local configuration along the lines of \cite{de2018identifying}. We then give a formal definition. 

\subsection{Informal description}\label{sec:local_informal}
Intuitively, agent $i$'s local configuration refers to the agents within path distance $r$ of $i$, their characteristics, and how they are connected. Larger values of $r$ are associated with more complicated configurations, which give a more precise picture about how $i$ is connected in the network. This idea is illustrated in Figure 1.

\begin{figure}
    \centering
    \caption{Illustration of local configurations.}
    \begin{subfigure}[b]{1\textwidth}
    \centering
            \begin{tikzpicture}[->,>= stealth,shorten >=1pt,auto,node distance= 1.5cm,
        thick,main node/.style={circle,fill=blue!20,draw,minimum size=.5cm,inner sep=0pt]}]

    \node[main node] (2) {$2$};
    \node[main node] (6) [ left of =2] {$6$};
    \node[main node] (5) [above right of=2, fill = red, opacity = .5, shape = rectangle] {$5$};
    \node[main node] (8) [below right of =5] {$8$};
    \node[main node] (1) [ below of =2] {$1$};
    \node[main node] (3) [ below right of =1, fill = red, opacity = .5, shape = rectangle] {$3$};
    \node[main node] (4) [ below left of =1] {$4$};
    \node[main node] (7) [ above right of =3] {$7$};
     \node[main node] (9) [ above left of =4] {$9$};
     \node[main node] (10) [ right of =7, fill = red, opacity = .5, shape = rectangle] {$10$};
       \node[main node] (11) [above right of =10] {$11$};
        \node[main node] (12) [above left of =9, fill = red, opacity = .5, shape = rectangle] {$12$};

    \path[-]
    (6) edge node {} (2)
    	edge node {} (9)
    (2) edge node {} (5)
    	edge node {} (5)
         edge node {} (1)
    (5) edge node {} (8)     
    (1) edge node {} (3)
    (1) edge node {} (4)
    (3) edge node {} (7)
    (4) edge node {} (9)
    (7) edge node {} (10)
    (11) edge node {} (10)
        (9) edge node {} (12);
\end{tikzpicture}
\caption{A network connecting a dozen agents. }
    \end{subfigure}

    \begin{subfigure}[b]{1\textwidth}
    \centering
    \begin{tikzpicture}[->,>= stealth,shorten >=1pt,auto,node distance=1cm,
        thick,main node/.style={circle,fill=blue!20,draw,minimum size=.3cm,inner sep=0pt]}]

    \node[main node] (2) [below right of=6]  {};
    \node[draw, fill=blue!20, shape=diamond, aspect=0.7, minimum height=0.3cm, inner sep=0pt] (1) [ below of =2] {$1$};
    \node[main node] (3) [ below right of =1, fill = red, opacity = .5, shape = rectangle] {};
    \node[main node] (4) [ below left of =1] {};

    \path[-]
    (2) edge node {} (1)
    (1) edge node {} (3)
    (1) edge node {} (4);
\end{tikzpicture}
      \begin{tikzpicture}[->,>= stealth,shorten >=1pt,auto,node distance=1cm,
        thick,main node/.style={circle,fill=blue!20,draw,minimum size=.3cm,inner sep=0pt]}]

    \node[draw, fill=blue!20, shape=diamond, aspect=0.7, minimum height=0.3cm, inner sep=0pt] (2) [below right of=6]  {$2$};
    \node[main node] (1) [ below of =2] {};
    \node[main node] (5) [ above right of =2, fill = red, opacity = .5, shape = rectangle] {};
    \node[main node] (6) [ left of =2] {};

    \path[-]
    (2) edge node {} (1)
    (2) edge node {} (5)
    (2) edge node {} (6);
\end{tikzpicture}

\caption{The local configurations for agents $1$ and $2$ associated with radius $1$ are equivalent. }
\end{subfigure}

    \begin{subfigure}[b]{1\textwidth}
    \centering
        \begin{tikzpicture}[->,>= stealth,shorten >=1pt,auto,node distance=1cm,
        thick,main node/.style={circle,fill=blue!20,draw,minimum size=.3cm,inner sep=0pt]}]

      \node[main node] (2) {};
    \node[main node] (6) [ left of =2] {};
    \node[main node] (5) [above right of=2, fill = red, opacity = .5, shape = rectangle] {};
    \node[draw, fill=blue!20, shape=diamond, aspect=0.7, minimum height=0.3cm, inner sep=0pt] (1) [ below of =2] {$1$};
    \node[main node] (3) [ below right of =1, fill = red, opacity = .5, shape = rectangle] {};
    \node[main node] (4) [ below left of =1] {};
    \node[main node] (7) [ above right of =3] {};
     \node[main node] (9) [ above left of =4] {};

    \path[-]
    (6) edge node {} (2)
        	edge node {} (9)
    (2) edge node {} (5)
    	edge node {} (5)
         edge node {} (1)     
    (1) edge node {} (3)
    (1) edge node {} (4)
    (3) edge node {} (7)
    (4) edge node {} (9);
\end{tikzpicture}
\hspace{5mm}
     \begin{tikzpicture}[->,>= stealth,shorten >=1pt,auto,node distance=1cm,
        thick,main node/.style={circle,fill=blue!20,draw,minimum size=.3cm,inner sep=0pt]}]

      \node[draw, fill=blue!20, shape=diamond, aspect=0.7, minimum height=0.3cm, inner sep=0pt] (2) {$2$};
    \node[main node] (6) [ left of =2] {};
    \node[main node] (5) [above right of=2, fill = red, opacity = .5, shape = rectangle] {};
    \node[main node] (1) [ below of =2] {};
    \node[main node] (3) [ below right of =1, fill = red, opacity = .5, shape = rectangle] {};
    \node[main node] (4) [ below left of =1] {};
    \node[main node] (8) [ below right of =5] {};
     \node[main node] (9) [ above left of =4] {};

    \path[-]
    (6) edge node {} (2)
        	edge node {} (9)
    (2) edge node {} (5)
    	edge node {} (5)
         edge node {} (1)     
    (1) edge node {} (3)
    (1) edge node {} (4)
    (5) edge node {} (8)
    (4) edge node {} (9);
\end{tikzpicture}
\caption{The local configurations for agents $1$ and $2$ associated with radius 2 are equivalent.}
\end{subfigure}

    \begin{subfigure}[b]{1\textwidth}
    \centering
   \begin{tikzpicture}[->,>= stealth,shorten >=1pt,auto,node distance= 1cm,
        thick,main node/.style={circle,fill=blue!20,draw,minimum size=.3cm,inner sep=0pt]}]

    \node[main node] (2) {};
    \node[main node] (6) [ left of =2] {};
    \node[main node] (5) [above right of=2, fill = red, opacity = .5, shape = rectangle] {};
    \node[main node] (8) [below right of =5] {};
    \node[draw, fill=blue!20, shape=diamond, aspect=0.7, minimum height=0.3cm, inner sep=0pt] (1) [ below of =2] {$1$};
    \node[main node] (3) [ below right of =1, fill = red, opacity = .5, shape = rectangle] {};
    \node[main node] (4) [ below left of =1] {};
    \node[main node] (7) [ above right of =3] {};
     \node[main node] (9) [ above left of =4] {};
     \node[main node] (10) [ right of =7, fill = red, opacity = .5, shape = rectangle] {};
        \node[main node] (12) [above left of =9, fill = red, opacity = .5, shape = rectangle] {};

    \path[-]
    (6) edge node {} (2)
    	edge node {} (9)
    (2) edge node {} (5)
    	edge node {} (5)
         edge node {} (1)
    (5) edge node {} (8)     
    (1) edge node {} (3)
    (1) edge node {} (4)
    (3) edge node {} (7)
    (4) edge node {} (9)
    (7) edge node {} (10)
        (9) edge node {} (12);
\end{tikzpicture}
\hspace{5mm}
       \begin{tikzpicture}[->,>= stealth,shorten >=1pt,auto,node distance= 1cm,
        thick,main node/.style={circle,fill=blue!20,draw,minimum size=.3cm,inner sep=0pt]}]

    \node[draw, fill=blue!20, shape=diamond, aspect=0.7, minimum height=0.3cm, inner sep=0pt] (2) {$2$};
    \node[main node] (6) [ left of =2] {};
    \node[main node] (5) [above right of=2, fill = red, opacity = .5, shape = rectangle] {};
    \node[main node] (8) [below right of =5] {};
    \node[main node] (1) [ below of =2] {};
    \node[main node] (3) [ below right of =1, fill = red, opacity = .5, shape = rectangle] {};
    \node[main node] (4) [ below left of =1] {};
    \node[main node] (7) [ above right of =3] {};
     \node[main node] (9) [ above left of =4] {};
        \node[main node] (12) [above left of =9, fill = red, opacity = .5, shape = rectangle] {};

    \path[-]
    (6) edge node {} (2)
    	edge node {} (9)
    (2) edge node {} (5)
    	edge node {} (5)
         edge node {} (1)
    (5) edge node {} (8)     
    (1) edge node {} (3)
    (1) edge node {} (4)
    (3) edge node {} (7)
    (4) edge node {} (9)
        (9) edge node {} (12);
\end{tikzpicture}
\caption{The local configurations for agents $1$ and $2$ associated with radius $3$ are not equivalent.}
\end{subfigure}
\end{figure}

Panel (a) depicts twelve agents connected in an unweighted and undirected network with a binary individualized treatment. Agents are either assigned to the treatment (red square nodes) or the control (blue circle nodes). Panel (b) depicts the local configurations of radius 1 for agents 1 and 2. They are both equivalent to a wheel with the focal untreated agent in the center and three other agents on the periphery, one of which is treated. Panel (c) depicts the local configurations of radius 2 for agents 1 and 2. They are  both equivalent to a ring between five untreated agents (one of which is the focal agent) where the focal agent is also connected to a treated agent who is connected to an untreated agent and another agent in the ring adjacent to the focal agent is connected to a treated agent. Panel (d) depicts the local configurations of radius 3 for agents 1 and 2. They are not equivalent because the local configuration for agent 1 contains four treated agents while the local configuration for agent 2 contains only three treated agents. 

In this way, one can describe the local configurations for any choice of agent and radius. Since the diameter of the network (the maximum path distance between any two agents) is 7, any local configuration of radius greater than 7 will be equal to the local configuration of radius 7. However, for networks defined on large connected populations, increasing the radius of the local configuration typically reveals a more complicated network structure. 

 Following \cite{benjamini2001recurrence}, we call the infinite-radius local configuration a rooted network. We now provide a formal definition. 
 
\subsection{Formal definition}\label{sec:local_formal}
We formally define the local configuration from Section \ref{sec:local_informal}. To do this, we first review some terminology and notation from the network theory literature. We then define the space of rooted networks building on \cite{benjamini2001recurrence,aldous2004objective}. 

\subsubsection{Terminology and notation}\label{sec:local_terminology}
A countable population of agents indexed by $\mathcal{I} \subseteq \mathbb{N}$ is linked in a weighted and directed network. The weight of a link from agent $i$ to $j$ is given by $D_{ij} \in \mathbb{Z}_{+}\cup \{\infty\}$. Each agent is also assigned a treatment $T_i \in \mathbb{R}$. The population vector of treatment assignments is denoted $\bold{T} = \{T_i\}_{i \in \mathcal{I}}$ and $\mathcal{T}$ is the set of all possible assignments. 

We explicitly allow for weighted networks because they are used in many of our motivating examples. They might measure, for example, the amount of physical distance between students, number of social connections surveyed between villagers, or degree of collaboration between researchers. The matrix $D$ indexed by $\mathcal{I}\times\mathcal{I}$ with $D_{ij}$ in the $ij$th entry is called the adjacency matrix. We take the convention that larger values of $D_{ij}$ correspond to weaker relationships between $i$ and $j$. We suppose that $D_{ij} = 0$ if and only if $i = j$. When the network is unweighted (agent pairs are either linked or not) $D_{ij} = 1$ denotes a link and $D_{ij} = \infty$ denotes no link from $i$ to $j$. 
 
A path from $i$ to $j$ is a finite ordered set $\{t_{1},...,t_{L}\}$ with values in $\mathbb{N}$, $t_{1} = i$, $t_{L} = j$, and $L \in \mathbb{N}$. The length of the path $\{t_{1},...,t_{L}\}$ is given by $\sum_{s = 1}^{L-1}D_{t_{s}t_{s+1}}$. The path distance from $i$ to $j$ denoted  $\rho(i,j)$ is the length of the shortest path from $i$ to $j$. That is,
\begin{align*}
\rho(i,j) := \inf_{\substack{\{t_{1},...,t_{L}\} \\ \text{ s.t. } t_{1} = i, t_{L} = j  }}\sum_{s = 1}^{L-1}D_{t_{s}t_{s+1}}.
\end{align*}  
If the path distance from $i$ to $j$ is finite we say that $i$ is path connected to $j$. For any $i \in \mathcal{I}$ and $r \in \mathbb{Z}_{+}$, agent $i$'s $r$-neighborhood $\mathcal{I}_{i}(r) := \{j \in \mathcal{I}: \rho(i,j) \leq r\}$ is the collection of agents within path distance $r$ of $i$. $N_{i}(r) := |\mathcal{I}_{i}(r)|$ is the size of agent $i$'s $r$-neighborhood (i.e. the number of agents in $\mathcal{I}_{i}(r)$). For any agent-specific variable (such as an outcome or treatment assignment)  $\mathbf{W} := \{W_{i}\}_{i \in \mathcal{I}}$, $W_{i}(r) := \sum_{j \in \mathcal{I}}W_{j}\mathbbm{1}\{\rho(i,j) \leq r\}$ is the $r$-neighborhood count of $\bold{W}$ for agent $i$. It describes the partial sum of $\bold{W}$ for the agents in $\mathcal{I}_{i}(r)$. $\mathcal{I}_{i}(0) = \{i\}$, $N_{i}(0) = 1$, and $W_{i}(0) = W_{i}$ since $D_{ij} = 0$ if and only if $i = j$. 

We assume that the network (adjacency matrix) is \emph{locally finite}. That is, for every $i \in \mathcal{I}$ and $r \in \mathbb{Z}_{+}$, $N_{i}(r) < \infty$. In words, the assumption is that every $r$-neighborhood of every agent contains only a finite number of agents. The assumption is implicit in much of the literature on network interference (including the examples in Section 3.2 below) where the researcher observes all of the relevant dependencies between agents in finite data. We do not assume a uniform bound on the size of the $r$-neighborhoods. The set of all locally finite adjacency matrices is denoted $\mathcal{D}$.

\subsubsection{Rooted networks}
For an adjacency matrix $D \in \mathcal{D}$ and treatment assignment vector $\mathbf{T} \in \mathcal{T}$, a network is the triple $(\mathcal{I}, D, \mathbf{T})$ where $\mathcal{I}$ is the vertex set and $D$ is the weighted edge set. A rooted network $G_{i} = G_i((\mathcal{I},D,\mathbf{T}), i)$ is the triple $(\mathcal{I},D,\mathbf{T})$ with a focal agent $i \in \mathcal{I}$ called the root. Informally, $G_i$ is the network $(\mathcal{I}, D, \mathbf{T})$  ``from the point of view'' of $i$. In this paper, we will often use the abbreviated notation $G_{i} = G_i(D,\mathbf{T})$ when the population $\mathcal{I}$ is clear. We will also use notation like $\mathcal{I}(g)$, $D(g)$ and $\mathbf{T}(g)$ to refer to the vertex set, weighted edge set, and set of treatment assignments associated with a given rooted network $g$. 

For any $r \in \mathbb{Z}_{+}$, $G_{i}^{r}$ is the subnetwork of $G_i((\mathcal{I},D,\mathbf{T}),i)$ induced by the agents within path distance $r$ of $i$ as measured by path distance $\rho$. Formally, $G_{i}^{r} := ((\mathcal{I}_{i}(r),D_{i}(r),T_{i}(r)), i)$, where $\mathcal{I}_{i}(r) := \mathcal{N}_{i}(r) := \{j \in \mathcal{I}:\rho(i,j) \leq r\}$, $D_{i}(r) := \{D_{jk} : j,k \in \mathcal{I}_i(r)\}$, and $T_{i}(r) := \{T_{j} \in \mathbf{T} : j \in \mathcal{I}_{i}(r)\}$. We say $G_{i}^{r}$ is the rooted network $G_{i}$ truncated at radius $r$. The rooted network formalizes the idea of a local configuration, which we described previously in Section \ref{sec:local_informal} . 

For any $\varepsilon \geq 0$, two rooted networks $G_{i_{1}}$ and $G_{i_{2}}$ are $\varepsilon$-isomorphic (denoted $G_{i_{1}} \simeq_{\varepsilon} G_{i_{2}}$) if all of their $r$-neighborhoods are equivalent up to a relabeling of the non-rooted agents, but where treatment assignments are allowed to disagree up to a tolerance of $\varepsilon$. Formally, $G_{i_1} \simeq_{\varepsilon} G_{i_{2}}$ if for any $r  \in \mathbb{Z}_{+}$ there exists a bijection $f: \mathcal{I}_{i_{1}}(r) \leftrightarrow \mathcal{I}_{i_{2}}(r)$ such that $f(i_{1}) = i_{2}$, $D_{jk} = D_{f(j)f(k)}$, for any $j, k \in \mathcal{I}_{i_{1}}(r)$, and $|T_{j} - T_{f(j)}| \leq \varepsilon$ for any $j \in \mathcal{I}_{i_{1}}(r)$. 



Two rooted networks that are not $0$-isomorphic are assigned a strictly positive distance inversely related to the largest $r$ and smallest $\epsilon$ such that they have $\varepsilon$-isomorphic $r$-neighborhoods. Specifically, we define the following distance on the set of rooted networks:
\begin{align}\label{metric}
d(G_{i_{1}},G_{i_{2}}) := \min\left\{\inf_{(r,\varepsilon) \in \mathbb{Z}_+\times\mathbb{R}_{++}} \{\zeta(r) + \varepsilon: G_{i_{1}}^{r} \simeq_{\varepsilon} G_{i_{2}}^{r}\}, \hspace{2mm} 2\right\}.
\end{align}
where $\zeta(x) = (1+x)^{-1}$. We demonstrate that $d(\cdot, \cdot)$ is a pseudo-metric in Appendix \ref{sec:G_props}. The outer minimization is not essential, but taking $d$ to be bounded simplifies later arguments.  Further note that our specific choice of $\zeta(\cdot)$ is arbitrary, in that any function that is supported on $\mathbb{R}_+$ which is monotonically decreasing to zero would be applicable, but we fix our choice of $\zeta(\cdot)$ for concreteness. 




The notions of an $\varepsilon$-isomorphism and and the rooted network distance $d$ are illustrated in Figure 1. Panels (b) and (c) both depict a pair of rooted networks that are $0$-isomorphic and so have a distance of $0$. Panel (d) depicts a pair of rooted networks that are not $\varepsilon$-isomorphic for any $\varepsilon > 0$ because the two networks do not have the same number of vertices (and so there cannot exist a bijection between them), but these networks truncated up to a radius of $2$ are in fact $0$-isomorphic. Other illustrations can be found in Tables 8-10 in Appendix Section D. 

We use $\mathcal{G}$ to denote the set of equivalence classes of all possible locally finite rooted networks under $d$. We demonstrate that $(\mathcal{G}, d)$ is a separable and complete metric space in Appendix \ref{sec:G_props}.\footnote{In principle our definition of a network and the corresponding definition of $d(\cdot, \cdot)$ could be modified to allow for edges which are continuously weighted. However, our proof of completeness would no longer be valid in this case.} Following \cite{aldous2004objective}, we call the topology on $\mathcal{G}$ induced by $d$ the local topology, and more broadly call modeling on $\mathcal{G}$ the local approach.

To be sure, the $\varepsilon$-isomorphism used to define the network distance $d$ can be modified to fit the researcher's specific setting. For instance, economic theory may suggest that two agents are qualitatively similar if their rooted networks are similar under edit distance (i.e. they can be made $\varepsilon$-isomorphic by adding or deleting a small number of agents or links). In this case the researcher may wish to relax the $\varepsilon$-isomorphism to allow for such discrepancies. The results developed in Section \ref{sec:policy} continue to hold mutatis mutandis with alternative distance metrics. We use the specific metric (\ref{metric}) because it was the simplest version we could think of such that any two agents with a distance of $0$ are observationally equivalent. That is, there is no remaining information in the treatment assignments and network connections that can be used to distinguish them. 
\section{Econometric model}\label{sec:motivation}
\subsection{Main specification}\label{sec:motivation_model}
Each agent $i \in \mathcal{I}$ has an outcome $Y_{i} \in \mathbb{R}$. The population vector of outcomes is denoted $\mathbf{Y} = \{Y_{i}\}_{i \in \mathcal{I}}$. Our main specification is 
\begin{align}\label{eq:outcome_model}
Y_i = h(G_i(D,\mathbf{T}), U_i).
\end{align}
 In this specification, $Y_i$ is determined by three factors: the population treatment assignments $\mathbf{T}$ (adding covariates to the model is straightforward, see Appendix A.1), the population network connections $D$, and agent-specific policy-invariant unobserved heterogeneity term $U_i$ (in some settings, for instance the empirical illustration we consider in Section \ref{sec:empirical_illustration}, the outcome model will depend solely on the network configuration $D$, and thus the treatment assignment vector $\mathbf{T}$ would not appear in model specification \eqref{eq:outcome_model}). 
 
 The function $h : \mathcal{G}\times\mathcal{U} \to \mathbb{R}$ maps $G_i$ and $U_i$ into the outcome $Y_i$.\footnote{$U_i$ is assumed to take values in a separable metric space $\mathcal{U}$. We endow $\mathcal{G} \times \mathcal{U}$ with the usual product topology, define probability measures on the corresponding Borel sigma-algebra, and associate a stochastic rooted network and error pair $(G_i, U_i)$ with a probability measure $\mu$. For now we take $\mu$ as arbitrary and fixed by the researcher. We motivate a specific choice of $\mu$ in the context of a multiple-networks setting in Section 4.2 below.} We assume that the researcher has some knowledge of $\mathbf{Y}$, $\mathbf{T}$, and $D$. $U_i$ and $h$ are unknown.  
 The key assumption of the model is that the effect of the first two factors on the outcome is intermediated by the rooted network $G_i(D,\mathbf{T}) :  \mathcal{D} 
 \times \mathcal{T} \to \mathcal{G}$. That is, for a fixed level of unobserved heterogeneity $U_i$, any pair of treatment assignments and network connections $(d,\mathbf{t}) \in \mathcal{D}\times\mathcal{T}$ that result in the same rooted network $G_{i}(d,\mathbf{t})$ lead to the same outcome $h(G_i(d,\mathbf{t}),U_i)$. In the network interference literature, a function of the treatment assignments and network connections that has this property is often called an \emph{exposure map} and any element of the range of the exposure map is called an \emph{effective treatment}. 
 
What distinguishes our specification from the established network interference literature is that, in the literature, the exposure map is typically taken to be low-dimensional, known up to a finite-dimensional vector of parameters. For instance, the exposure map may be a linear function of the number of treated agents with a certain distance to the focal agent. A drawback of this strategy is that the researcher may misspecify the exposure map, i.e. they choose a low-dimensional function that does not fully intermediate the effect of the treatment assignments and network connections on the outcome. When the exposure map is misspecified, the researcher's inferences about the impact of policy counterfactuals will not generally be valid. For a further discussion of this point, see our note \cite{auerbach2024discussion}. 

Our model, by contrast, takes the exposure map to be rooted-network valued. A benefit of this approach is that the rooted network contains everything that is knowable about an agent based on the nearby treatment assignments and network connections. Because of this, it nests a large class of potential exposure maps and, as a result, is less prone to misspecification. A cost of this approach is that, in many settings, the space of rooted networks may be relatively complicated, potentially leading to slow rates of convergence. We discuss this issue in more detail in Section 4 below.

\subsection{Examples}\label{sec:motivation_examples}
We illustrate the prevalence of our main specification (\ref{eq:outcome_model}) with three concrete examples  from the network economics literature. In each example, it is possible to specify a valid low-dimensional exposure map that summarizes the information in the treatment assignments and network connections that is relevant for each agent's outcome. However, the natural exposure map in one example is not valid in the setting of another. Our rooted-network-valued exposure map, by contrast, is valid in all three settings. 

\begin{example}\label{ex:spillovers}
(Neighborhood spillovers): Agents are assigned to either treatment or control status with $T_{i} = 1$ if $i$ is treated and $T_{i} = 0$ if $i$ is not. Agent $i$'s outcome depends on their treatment status and the number of treated agents nearby 
\begin{align*}
Y_{i} = Y\left(T_{i}, T_{i}(r), U_i\right) 
\end{align*}
where $T_{i}(r) = \sum_{j \in \mathbb{N}}T_{j}\mathbbm{1}\{\rho(i,j) \leq r\}$ is the neighborhood count of treatment assigned within some fixed radius $r$ of $i$. See for instance \cite{cai2015social,leung2016treatment,he2018measuring,viviano2019policy}. We recast this model as a special case of (\ref{eq:outcome_model}) by noting that both $T_i$ and $T_i(r)$ can be written as functions of $G_i$, agent $i$'s rooted network. Specifically, $T_{i}$ is  associated with agent $i$ and so is determined by $G_{i}^{0}$. $T_{i}(r)$ counts the number of treated agents within distance $r$ of $i$ and so is determined by $G_{i}^{r}$. It follows that $Y_{i} = h(G_{i}, U_i)$ for some $h$.
\end{example}

\begin{example}\label{ex:capital}
(Social capital formation): Agents leverage their network connections to garner favors, loans, advice, etc. \cite{jackson2012social} specify a model in which one agent performs a favor for another when there is a third agent connected to both that can monitor the exchange. We consider a stochastic version of this model where agent $i$'s stock of social capital depends on their number of monitored connections 
\begin{align*}
Y_{i}= \left(\sum_{j \in \mathcal{I}} \mathbbm{1}\{\mathcal{I}_{i}(1)\cap \mathcal{I}_{j}(1) \neq \emptyset\}\right)\cdot U_{i}~.
\end{align*}
\cite{karlan2009trust,cruz2017politician} spectify related models of social capital. An example policy of interest is the effect of adding additional connections to the network. We recast this model as a special case of (\ref{eq:outcome_model}) by noting that the number of agents of path distance $2$ from $i$, $\sum_{j \in \mathcal{I}} \mathbbm{1}\{\mathcal{I}_{i}(1)\cap \mathcal{I}_{j}(1) \neq \emptyset\}$, is determined by $G_i^2$. It follows that $Y_i = h(G_i ,U_i)$ for some $h$.
\end{example}

\begin{example}\label{ex:social_interactions}
(Social interactions): Agent $i$'s equilibrium outcome depends on the number and length of paths between them and the other agents, the treatment statuses of those other agents, and the average treatment statuses of the agents nearby those other agents
\begin{align*}
Y_{i} = \lim_{S \to \infty}\sum_{s=0}^{S}\left[\delta^{s}A^*(D)^{s}\left(\mathbf{T}\beta + \mathbf{T}^*(1)\gamma \right)\right]_i + U_{i}~,
\end{align*}
where $T^*_{i}(1) = T_{i}(1)/N_{i}(1)$, $[\cdot]_i$ is the $i$th entry of a vector, and the $ij$th entry of $A^{*}(D)$ is $A_{ij}^*(D) = \mathbbm{1}\{0 < D_{ij} \leq 1\}/N_{i}(1)$. See for instance \cite{bramoulle2009identification, blume2010identification,de2010identification, lee2010specification, goldsmith2013social}. An example policy effect of interest is the average effect of removing a particular agent from the community. See for instance \cite{ballester2006s,calvo2009peer,lee2021key}. We recast this model as a special case of (\ref{eq:outcome_model}) by noting that the $i$th entry of $\delta^{s}A^*(D)^{s}\left(\mathbf{T}\beta + \mathbf{T}^*(1)\gamma \right)$ only depends on the treatment statuses and connections of agents within path-distance $s+1$ of $i$ and so is determined by $G_{i}^{s+1}$. It follows that 
$Y_{i} = \lim_{S\to\infty}\sum_{s=1}^{S}h_{s}(G_{i}^{s},U_{i}) = h(G_{i}, U_{i})$ for some functions $h_{s}$ and $h$. 
\end{example}

There are two key differences between Example \ref{ex:social_interactions} and Examples \ref{ex:spillovers} and \ref{ex:capital}. First, in Example (\ref{ex:social_interactions}), the effective treatment for $i$ depends on all of the treatment statuses and network connections of the agents path connected to $i$. \cite{leung2019causal} argues that this is a common feature of many economic models of network interference. Second, in Example (\ref{ex:social_interactions}), the  description of the effective treatment depends on model parameters which are unknown to the policy maker. Our local approach is, to our knowledge, the first to accommodate these two features when modeling causal effects under network interference. 
 
To be sure, in each of the three examples, it is possible to take the exposure map to be a relatively simple function of the treatment assignments and network connections. For instance, in Example (\ref{ex:spillovers}), a valid exposure map is $(T_{i}, T_{i}(r))$ and in Example (\ref{ex:social_interactions}) it is $(\{\left[A^*(D)^{s}\left(\mathbf{T}\beta + \mathbf{T}^*(1)\right)\right]_i \}_{s \in \mathbb{N}})$. However, the exact choice of exposure map requires the researcher to commit to one parametric specification. Our approach, by contrast, results in a valid exposure map for all three examples. 
 
 \subsection{Parameters of interest}\label{sec:parameters}
Our main objects of interest are the average structural function (ASF) and distributional structural function (DSF) that describe the outcome for $i$ associated with a policy that sets the rooted network to some $g \in \mathcal{G}$. These are, respectively,
\begin{align}\label{poi}
h(g) = E\left[h(g,U_i)\right]~\text{ and }
h_y(g) = E\left[\mathbbm{1}\{h(g,U_i)\le y\}\right]
\end{align}
where the expectation refers to the marginal distribution of $U_i$ under $\mu$ (see footnote 4 in Section \ref{sec:motivation_model} above). See for instance \cite{blundell2003endogeneity}. These functions can be used to estimate and conduct inferences about many causal effects of interest. For example, the average treatment effect (ATE) associated with a policy that alters agent $i$'s rooted network from $g$ to $g'$ is described by  $h(g') - h(g)$. Under appropriate continuity assumptions, $h(g)$ and $h_y(g)$ can be approximated by averaging the outcomes corresponding to rooted networks that are close to $g$ under $d$. We demonstrate this idea with applications to estimating policy effects/treatment response and  testing policy irrelevance/no treatment effects in Sections \ref{sec:est_asf} and \ref{sec:dist_test} below.  

This ATE is related to the average exposure effect of \cite{leung2019causal,savje2021average}. These authors consider inference in the setting where the exposure map chosen by the researcher is potentially invalid in the sense that the CTR assumption of Section \ref{sec:local_formal} is false. A limitation of their approach is that an exposure effect based on a misspecified exposure map may not characterize any policy of interest. We limit the scope for  misspecification by using a general class of exposure maps indexed by the space of rooted networks. 

We could alternatively consider conditional treatment effect parameters. For instance, the researcher may be interested in a policy that alters treatment assignment conditional on the structure of the network. Or the policy may alter the network connections conditional on some observed agent attributes. In such cases the researcher may wish to consider a conditional average treatment effect where the researcher conditions on the part of the rooted network not altered by the policy. Such parameters may in some cases be identified under alternative assumptions than those required to identify the above unconditional ATE. This is discussed following Assumption \ref{ass:sampling} below.

\section{Applications to causal inference}\label{sec:policy}
We apply the local approach framework of Section \ref{sec:motivation} to two causal inference problems. In both problems, the researcher begins with a rooted network associated with a status-quo policy and is tasked with evaluating the impact of an alternative. The researcher has access to data on outcomes and policies from multiple independent communities or clusters such as schools, villages, firms, or markets. This data structure, described in Section \ref{sec:exp_setup}, is not crucial to our methodology, but simplifies the analysis. It is also very common in the network economics literature. We leave applications to alternative settings, for example with data from one large network or endogenous policies, to future work. 

Our first application, described in Section \ref{sec:est_asf}, is to the estimation and inference of policy effects. We construct a $k$-nearest-neighbors estimator for the average structural function in (\ref{poi}) and provide non-asymptotic bounds on its estimation error. We then establish an asymptotic normality result for the estimator that can be used to construct confidence intervals for the policy effect in the usual way. Our second application, described in Section \ref{sec:dist_test}, is to testing policy irrelevance. That is, we test the hypothesis that the rooted networks associated with two policies generate the same distribution of outcomes. These applications contrast a literature studying the magnitude of any potential spillover effects or testing the hypothesis of no spillovers \cite[see for instance][]{aronow2012general,athey2018exact,hu2021average,savje2021average}.  

\subsection{Multiple-networks setting}\label{sec:exp_setup}
A random sample of communities (clusters) are indexed by $c \in [C] := 1, \ldots C$. Each $c \in [C]$ is associated with a finite collection of $m_{c}$ observations $\{W_{ic}\}_{i \in [m_{c}]}$ where $m_c$ is a random positive integer, $W_{ic} := (Y_{ic},G_{ic})$, and $Y_{ic} := h(G_{ic}, U_{ic})$ for some unobserved error $U_{ic}$. Intuitively, each community $c$ is represented by an initial network connecting $m_c$ agents and the $m_{c}$ rooted networks refer to this initial network rooted at each agent in the community. Let $W_{c} := \{W_{ic}\}_{i \in [m_c]}$. We impose the following assumptions on $\{(G_{ic}, U_{ic})\}_{i \in [m_{c}], c \in [C]}$. 
\begin{assumption}\label{ass:sampling}
\begin{enumerate}[(i)]
\item[]
\item $\{W_{c}\}_{c \in [C]}$ are independent and identically distributed  (across communities). 
\item $\{U_{ic}\}_{i \in [m_{c}], c \in [C]}$ are identically distributed (within and across communities). 
\item For any measurable $f$, $i \in [m_{c}]$, and $c \in [C]$, \[E\left[f(G_{ic},U_{ic}) | G_{1c},...,G_{m_{c}c}\right] = E\left[f(G_{ic},U)|G_{ic}\right]~,\] 
where $U$ is an independent copy of $U_{ic}$ (i.e. $U$ has the same marginal distribution as $U_{ic}$ but is independent of $\{W_{c}\}_{c \in [C]}$). 
\end{enumerate}
\end{assumption}
Assumption \ref{ass:sampling} (i) is what makes our analysis ``multiple-networks.'' It states that the networks and errors are independent and identically distributed across communities. We use this independence across communities to characterize the statistical properties of our test procedure and estimator below. We do not make any restrictions on the dependence structure between observations within a community. This allows for arbitrary dependence within a community; see for instance Example \ref{ex:social_interactions}. Weakening the independence assumption (for instance, considering dependent data from one large community) would require additional assumptions about the intra-community dependence structure which we leave to future work. 

Assumption \ref{ass:sampling} (ii) fixes the marginal distribution of the errors. It is used to define the policy effect of interest. The relevant structural function is defined as the expected outcome over the homogeneous marginal distribution of $U_{ic}$ for a fixed rooted network. This assumption can be dropped, for example, by defining the expectation to be with respect to the (mixture) distribution of $U_{\iota_{c} c}$ generated by drawing $\iota_{c}$ uniformly at random from $[m_{c}]$.

Assumption \ref{ass:sampling} (iii) states that the rooted networks are exogenous (i.e. the errors are policy-irrelevant). We require that the conditional distribution of $(G_{ic}, U_{ic})$ given $G_{1c},...,G_{m_{c}c}$ is equal to the conditional distribution of $(G_{ic}, U)$ given $G_{ic}$, where $U$ is an independent copy of $U_{ic}$. Exogeneity is a strong assumption, but allows us to approximate the unknown policy functions using sample averages. 

The literature cited in Section \ref{sec:motivation_examples} provides two separate motivations for assuming that the network is exogenous. One motivation comes from the random assignment of network connections in an experiment. An example of this is \cite{azoulay2010superstar}. Another motivation comes from correct model specification, where the researcher accounts for all of the relevant determinants of the outcome on the right-hand side of their model. An example of this is \cite{calvo2009peer}. Our Assumption  \ref{ass:sampling} is consistent with both motivations. 

The study of endogenous rooted networks where the policy $(\mathbf{T},D)$ is potentially related to the errors $\textbf{U}$ is left to future work. If the policy maker is interested in identifying a conditional average treatment effect as described in the discussion following the definition of $h(g)$ in Section \ref{sec:parameters}, then Assumption \ref{ass:sampling} (iii) could be modified to a conditional exogeneity assumption where the structural errors are independent of the effective treatments conditional on the relevant part of the rooted network.

\subsection{Estimating Policy Effects}\label{sec:est_asf}
The policy maker begins with a status-quo policy described by a treatment and network pair $(\mathbf{t}, d)$, and proposes an alternative $(\mathbf{t}', d')$. The researcher is tasked with estimating the expected effect of the policy change for an agent whose effective treatment under the status-quo is described by the rooted network $g \in \mathcal{G}$ and whose effective treatment under the alternative is described by the rooted network $g' \in \mathcal{G}$. Following Section 3.2, the potential outcomes under policies $g$ and $g'$ are described by $h(g,U)$ and $ h(g',U)$ respectively for some error $U$. The object of interest is the policy effect given by 
\[h(g') - h(g)\]
where $h(g) = E\left[h(g,U)\right]$. 

We provide two examples illustrating how this policy effect could be used in practice.
\begin{example}
(Cluster-randomized experiment): An individualized binary treatment is randomly assigned at the cluster-level (so that every member of the cluster is assigned the same treatment). In this case, we take $g'$ to be the rooted network of a given individual in a community when their cluster is treated, and $g$ to be the counterfactual rooted network with the same agents and connections but with none of the agents treated. The policy effect is then the difference in expected outcomes for the given individual under the two rooted network structures. 
\end{example}

\begin{example}
(Removal of an agent): A central agent is removed from their community. In this case, we take $g'$ to be the rooted network of a remaining individual in the community with the central agent removed, and $g$ to be the counterfactual rooted network that existed before the central agent was removed. The policy effect is then the difference in expected outcomes for the remaining individual under the two rooted network structures. 
\end{example}

The proposed estimator for the policy effect is described in Section \ref{sec:estimator}. In words, it compares the average outcome of the $k$ agents whose rooted networks are most similar to $g$ to the analogous average for $g'$, using the network distance $d$. 



We impose smoothness conditions on the model parameters and bound the variance of the outcome. We do not believe these assumptions to be restrictive in practice. 

\begin{assumption}\label{ass:no_ties}
For $g_{0} \in \{g,g'\}$, $\psi_{g_0}(\ell) := P\left(\min_{i \in [m_c]} d(G_{ic},g_{0}) \le \ell\right)$ is continuous at every $\ell \ge 0$.
\end{assumption}

The function $\psi_{g}(\ell)$ measures the probability that there exists an agent from a randomly drawn community whose rooted network is within $\ell$ of $g$. Assumption \ref{ass:no_ties} states that $\psi_{g}(\ell)$ and $\psi_{g'}(\ell)$ are continuous in $\ell$. It justifies a probability integral transform used to characterize the bias of the estimator. The assumption can be guaranteed by adding a randomizing component to the metric \citep[see the discussion following equation (19) in][]{gyorfi2020universal}, which allows for $G_{ic}$ to be discrete.

\begin{assumption}\label{ass:m_smooth}
For $g_{0} \in \{g,g'\}$ there exists an increasing function $\phi_{g_{0}}: \mathbb{R}_{+} \to \mathbb{R}_{+}$ such that  $\phi_{g_{0}}(x) \rightarrow \phi_{g_{0}}(0) = 0$ as $x \rightarrow 0$ and for every $\tilde{g} \in \mathcal{G}$ 
\[\left|h(g_{0}) - h(\tilde{g})\right| \le \phi_{g_{0}}\left(d(g_{0},\tilde{g})\right)~.\]
\end{assumption}

Assumption \ref{ass:m_smooth} states that $h$ has a modulus of continuity $\phi_{g_0}$ at $g_0$. Such a smoothness condition is standard in the nonparametric estimation literature. It is not generally testable. A model of network interference implies a specific choice of $\phi_{g}$. The three examples of Section \ref{sec:motivation_examples} all satisfy the assumption with $\phi_{g_0}(x) = C^{(x-1)/x}$ for some $C > 1$ potentially depending on parameters of the model but not $g_0$. See Appendix Section \ref{sec:bias_disc} for details. 

\begin{assumption}\label{ass:var}
For every $\tilde{g} \in \mathcal{G}$,  $E\left[h(\tilde{g},U_{ic})^2\right] \leq \bar{\sigma}^2$ for some $\bar{\sigma}^2 > 0$.
\end{assumption}
Assumption \ref{ass:var} bounds the variance of the outcome variable. It is also standard in the nonparametric estimation literature. 
 
\subsubsection{Estimator}\label{sec:estimator}
Let $D_c(g) := \min_{i \in [m_c]}d(G_{ic},g)$ be the distance of the closest rooted network to $g$ in network $c$, and let $\mathcal{I}_c(g) := \arg\min_{i \in [m_c]} d(G_{ic},g)$ be the corresponding set of roots of the rooted networks that achieve this minimum. For each $c \in [C]$, let $\bar{Y}_c(g) := \frac{1}{|\mathcal{I}_c(g)|}\sum_{i \in \mathcal{I}_c(g)}Y_{ic}$ denote the average outcome in $\mathcal{I}_c(g)$. Order $\{\bar{Y}_c(g)\}_{c \in [C]}$ to be increasing in $D_c(g)$ (where ties are broken uniformly at random: see Assumption \ref{ass:no_ties} above and the corresponding discussion). Let the reordered data sequence be given by $(D^*_1(g),\bar{Y}^*_1(g)),(D^*_2(g),\bar{Y}^*_2(g)), \ldots, (D^*_C(g),\bar{Y}^*_C(g))$. The proposed estimator for $h(g)$ is then the average of $\{\bar{Y}^*_k(g)\}_{1 \le j \le k}$ for some $1 \le k \le C$:
\[\hat{h}(g) := \frac{1}{k}\sum_{j = 1}^k \bar{Y}^{*}_{j}(g)\]
and the analogous estimator for $h(g') - h(g)$ is
\[\hat{h}(g') - \hat{h}(g) := \frac{1}{k}\sum_{j = 1}^k \left(\bar{Y}^{*}_{j}(g')-\bar{Y}^{*}_{j}(g)\right).\]

\subsubsection{Bound on estimation error}\label{sec:mse_bound}
We derive a finite-sample bound on the mean-squared error of $\hat{h}(g)$ building on work by \cite{biau2015lectures,doring2017rate,gyorfi2020universal}. 
 
\begin{theorem}\label{thm:mse_bound}
Under Assumptions \ref{ass:sampling}, \ref{ass:no_ties}, \ref{ass:m_smooth}, and \ref{ass:var},
\[E\left[\left(\hat{h}(g) - h(g)\right)^2\right] \le \frac{\bar{\sigma}^2}{k} + E\left[\varphi_g(U_{(k,C)})^2\right]~,\]
where $\varphi_g(x) = \phi_g \circ \psi_g^{\dagger}(x)$, $\psi_g^{\dagger}:[0,1] \rightarrow \mathbb{R}_{+}$ refers to the upper generalized inverse
\[\psi_g^{\dagger}(x) = \sup\{\ell \in \mathbb{R}_{+}: \psi_g(\ell) \le x\}~,\] and $U_{(k,C)}$ is distributed $Beta(k,C-k+1)$. 
\end{theorem}

The bound in Theorem \ref{thm:mse_bound} features a familiar bias-variance decomposition. It also establishes the consistency of $\hat{h}(g)$ in well-behaved settings. For example if (for $g$ fixed) $\varphi_g(\cdot)$ is bounded, continuous at zero, $\varphi_g(0) = 0$, and $k \rightarrow \infty$, $k/C \rightarrow 0$ as $C \rightarrow \infty$ then 
\[\bar{\sigma}^2/k \rightarrow 0 \text{ and } E[\varphi_g(U_{(k,C)})^2] = E\left[\left(\phi_g \circ \psi_g^{\dagger}(U_{(k,C)})\right)^2\right] \rightarrow 0~.\]
The variance component $\bar{\sigma}^2/k$ is standard and decreases as $k$ grows large. In contrast, the bias component $E[\varphi_g(U_{(k,C)})^2]$ and its relationship with $k$ and $C$ are difficult to characterize without further information about $\phi_{g}$ and $\psi_{g}$. Intuitively, the first controls the smoothness of the regression function $h(g)$ and the second controls the proximity of the nearest-neighbors that make up $\hat{h}(g)$ in terms of proximity to $g$. We provide further discussion of how the bound depends on features of these parameters in Appendix \ref{sec:bias_disc}. We note that the bound remains valid even if we do not observe rooted networks that are arbitrarily close to $g$ as $C$ grows; in this case, we would not expect $\hat{h}(g)$ to be consistent unless $\phi_g(\cdot)$ is exactly zero outside of a neighborhood of $g$. Theorem \ref{thm:mse_bound} has the immediate corollary
\begin{corollary}\label{cor:mse_bound}
Suppose the hypothesis of Theorem \ref{thm:mse_bound}. Then
\[E\left[\left(\hat{h}(g')-\hat{h}(g) - (h(g')-h(g))\right)^2\right] \le \frac{4\bar{\sigma}^2}{k} + 4E\left[\varphi_{g\vee g'}(U_{(k,C)})^2\right]\]
where $\varphi_{g\vee g'} := \varphi_{g}\vee \varphi_{g'}$. 
\end{corollary} 



\subsubsection{Limiting Distribution of $\hat{h}(g) - \hat{h}(g')$}
We derive the limiting distribution of $\hat{h}(g)$ (and $\hat{h}(g) - \hat{h}(g')$) under specific assumptions on the rate of growth of the nearest-neighbors parameter $k$ relative to $C$.
We impose the following additional assumptions on the data generating process:
\begin{assumption}\label{ass:pos_mass}
For any $\ell > 0$ and $g_{0} \in \{g,g'\}$, $\psi_{g_{0}}(\ell) := P\left(\min_{i \in [m_c]} d(G_{ic}, g_{0}) \le \ell\right) > 0$. 
\end{assumption}

Recall that $\psi_g(\ell)$ measures the probability that there exists an agent from a randomly drawn community whose rooted network is within distance $\ell$ of $g$ under $d$. Assumption \ref{ass:pos_mass} states that for a randomly drawn community there exists a rooted network within distance $\ell$ of $g$ or $g'$ with positive probability. It implies that as the number of communities $C$ grows, the researcher will eventually observe rooted networks that are arbitrarily close to $g$ or $g'$. 

\begin{assumption}\label{ass:normal_regular} For $g_0$ in $\{g, g'\}$, the conditional means
\[E[\bar{Y}_c(g_0)|D_c(g_0)=d] \text{ and } E[\bar{Y}_c(g_0)^2|D_c(g_0) = d]\]
are are continuous in a neighborhood of $d=0$. The conditional fourth moment of $\bar{Y}_c(g_0)$ is uniformly bounded, i.e. 
\[E[\bar{Y}_c(g_0)^4|D_c(g_0) = d] \le M~,\]
for some $M < \infty$, for all $d$.
\end{assumption}
Assumption \ref{ass:normal_regular} imposes some additional sufficient continuity to guarantee that the conditional variance $\text{var}(\bar{Y}_c(g_0)|D_c(g)=d)$ is well-behaved as a function of $d$. The bound on the conditional fourth moments of $\bar{Y}_c(g_0)$ could be reasonably weakened at the cost of introducing more complicated arguments. We present the limiting distribution of $\hat{h}(g)$ as the number of nearest neighbors $k$ grows large at an appropriate rate.

\begin{theorem}\label{thm:normal}
Maintain Assumptions \ref{ass:sampling} -- \ref{ass:normal_regular}. Further suppose that $k \rightarrow \infty$ as $k/C \rightarrow 0$ and that $kE[\varphi_g(U_{(k,C)})^2] \rightarrow 0$. Then
\[\sqrt{k}(\hat{h}(g) - h(g)) \xrightarrow{d} N(0, \sigma^2_g(0))~,\]
where $\sigma^2_g(d) := \text{var}(\bar{Y}_c(g)|D_c(g) = d)$.
\end{theorem}

We also obtain the following immediate corollary for the limiting distribution of $\hat{h}(g) - \hat{h}(g')$ when the estimators $\hat{h}(g)$ and $\hat{h}(g')$ are computed from disjoint subsamples of $\{W_c\}_{c \in [C]}$: 

\begin{corollary}\label{cor:normal}
Consider a partition of the data $\{W_c\}_{c \in [C]}$ into two disjoint sets labelled $\mathcal{W}_1$ and $\mathcal{W}_2$ of size $C_1$ and $C_2$, respectively. Let $\hat{h}(g)$ be the $k$-nearest neighbors estimator of $h(g)$ computed on $\mathcal{W}_1$ and let $\hat{h}(g')$ be the $k$-nearest neighbors estimator of $h(g')$ computed on $\mathcal{W}_2$. Maintain Assumptions \ref{ass:sampling} -- \ref{ass:normal_regular}. Further suppose that $k \rightarrow \infty$ as $k/\min\{C_1,C_2\} \rightarrow 0$ and that $kE[\varphi_{g \vee g'}(U_{k,C})^2] \rightarrow 0$. Then
\[\sqrt{k}(\hat{h}(g) - \hat{h}(g') - (h(g) - h(g'))) \xrightarrow{d} N(0, \sigma^2_g(0) + \sigma^2_{g'}(0))~.\]
\end{corollary}
To perform inference on $h(g)$ (or $h(g) - h(g')$) using these results, an estimator of the asymptotic variance $\sigma^2_{g_0}(0)$ for $g_0 \in \{g,g'\}$ can be constructed as
\[\hat{\sigma}^2_{g_0} := \frac{1}{k}\sum_{j  = 1}^k(\bar{Y}_j^*(g_0) - \hat{h}(g_0))^2~.\]
Theorem \ref{thm:variance} establishes conditions under which this estimator is consistent. 
\begin{theorem}\label{thm:variance}
Maintain Assumptions \ref{ass:sampling} -- \ref{ass:normal_regular}. Further suppose that $k \rightarrow \infty$ as $k/C \rightarrow 0$, then for $g_0 \in \{g,g'\}$
\[\hat{\sigma}^2_{g_0} \xrightarrow{p} \sigma^2_{g_0}(0)~.\]
\end{theorem}

Combining Theorem \ref{thm:normal} (or Corollary \ref{cor:normal})  with Theorem \ref{thm:variance} establishes the asymptotic validity of inferences based on these estimators. 

\begin{remark}\label{rem:undersmooth}
Our results are derived under the common assumption that $k$ grows at such a rate that we achieve ``under-smoothing."  Intuitively, this means that $k$ must grow sufficiently slowly to guarantee that the bias induced by choosing nearest neighbors ``far away" from $g$ is not too large. Although this condition may be straightforward to satisfy in certain special cases (for instance, if the smoothness function $\phi_g(\cdot)$ is suspected to decay very quickly to zero), in general it introduces a tension in the choice of $k$ when using Theorem \ref{thm:normal} to conduct inference on $h(g)$; we want $k$ to be large enough to guarantee that the normal approximation is appropriate, while being small enough to ensure that we do not introduce an asymptotic bias into the limiting distribution. We could consider procedures which explicitly account for this asymptotic bias \citep[see for instance][]{armstrong2020simple}. However, we leave the study of appropriate modifications of these approaches to our setting to future work. Instead, in Section \ref{sec:dist_test} we propose a test for a closely related but distinct hypothesis of policy irrelevance, which states that the rooted networks associated with two policies generate the same distribution of outcomes. As we will show, the gain from testing this stronger null hypothesis is that the resulting test will be asymptotically valid even when the number of nearest neighbors is held fixed in the limit.
\end{remark}

\begin{remark}
Our results concern the policy effect $h(g') - h(g)$ where $g'$ and $g$ are both specific elements of $\mathcal{G}$. In many settings, however, the policies might refer to a collection of networks $\mathcal{G}_1$ and $\mathcal{G}_2$ where $\mathcal{G}_1, \mathcal{G}_2 \subset \mathcal{G}$, i.e. the parameter of interest is $E[h(G_i,U_i)|G_i \in \mathcal{G}_1] - E[h(G_i,U_i)|G_i \in \mathcal{G}_2]$. For example, $\mathcal{G}_1$ may be the set of rooted networks that are observed in particular community, or the set of rooted networks where exactly half of the root's neighbors are treated. It should be possible to extend our methodology to these cases by first estimating the conditional expectation of the outcome with respect to each rooted network in each collection, and then taking a weighted average of the results to estimate $E[h(G_i,U_i)|G_i \in \mathcal{G}_1]$ and $E[h(G_i,U_i)|G_i \in \mathcal{G}_2]$. Following the literature on two-step semiparametric estimation \citep[see, broadly, ][]{powell1994estimation}, it is likely that, under certain conditions, this double averaging would lead to faster rates of convergence than those described in Theorems 4.1 and 4.2 above.  
\end{remark}

\subsection{Testing policy irrelevance}\label{sec:dist_test}
The policy maker begins with a status-quo community policy described by a treatment and network pair $(\mathbf{t}, d)$, and proposes an alternative $(\mathbf{t}', d')$. The researcher is tasked with testing whether the two policies are associated with the same distribution of outcomes for an agent whose effective treatment under the status-quo is described by the rooted network $g \in \mathcal{G}$ and whose effective treatment under the alternative is described by the rooted network $g' \in \mathcal{G}$. One can extend our procedure to test policy irrelevance for multiple agents via a multiple testing procedure.

Following Section 3.2, the potential outcomes under $g$ and $g'$ are given by $ h(g,U)$ and $ h(g',U)$ respectively for some error $U$. The hypothesis of policy irrelevance is 
\begin{equation}\label{eq:null}
H_0: h_{y}(g) = h_{y}(g') \text{ for every } y \in \mathbb{R}
\end{equation}
where $h_{y}(g) := E\left[\mathbbm{1}\{h(g,U)\}\leq y\right]$. Under Assumption 4.1 (iii), $h_{y}(g)$ describes the conditional distribution of $Y_{i}$ given $G_{i} \simeq g$. 

We revisit the two examples from Section 4.2 to illustrate how this hypothesis test may be used in practice
\begin{example}
(Cluster-randomized experiment, continued): Recall that in this example an individualized binary treatment is randomly assigned at the cluster-level, $g'$ is the rooted network of a given individual in a community when their cluster is treated, and $g$ is the counterfactual rooted network when their cluster is not treated. The hypothesis of policy irrelevance is that the distribution of outcomes associated with both rooted networks is the same. A test of this hypothesis may be used to establish that the effect of the treatment is statistically significant: the observed outcomes cannot be explained solely by variation in the unobserved heterogeneity terms.  
\end{example}

\begin{example}
(Removal of an agent, continued): Recall that in this example a central agent is removed from their community, $g'$ is the rooted network of a remaining individual in the community with the central agent is removed, and $g$ is the counterfactual rooted network that existed before the central agent was removed. The hypothesis of policy irrelevance is that the distribution of outcomes associated with both rooted networks is the same. A test of this hypothesis may be used to establish that the effect of removing the central agent is statistically significant: the observed outcomes cannot be explained solely by variation in the unobserved heterogeneity terms.  
\end{example}

Our test procedure is described in Section \ref{sec:algorithm}. Intuitively, it compares the empirical distribution of outcomes for agents in the data whose rooted networks are most similar to $g$ to the analogous distribution for $g'$. Asymptotic validity follows after imposing an additional continuity assumption on $h_y(g)$:

\begin{assumption}\label{ass:dist_cont}
For every $\tilde{g} \in \mathcal{G}$, the distribution of $h(\tilde{g},U)$ is either continuous or discrete with finite support. For every $y \in \mathbb{R}$, $h_{y}(\tilde{g})$ is continuous in $\tilde{g}$. 
\end{assumption}
Assumption \ref{ass:dist_cont} states that the distribution of outcomes associated with agents whose rooted networks are close to $\tilde{g}$ approximate the distribution of outcomes at $\tilde{g}$. These continuity assumptions are satisfied by the three examples of Section \ref{sec:motivation_examples}. 



\subsubsection{Test procedure}\label{sec:algorithm}

We propose an approximate permutation test of $H_{0}$ building on \cite{canay2017randomization, canay2018approximate}. The test procedure is described in Algorithm 4.1. We assume that $C$ is even to simplify notation. When determining the closest agent in Step 1 or reordering the vectors in Step 2, ties are broken uniformly at random. 

\begin{algorithm}\label{algo:test}{Input: data $\{W_{c}\}_{c \in [C]}$ partitioned into two disjoint sets of of size $C/2$ labelled $\mathcal{W}_1$ and $\mathcal{W}_2$, and parameters $q \le C/2$, $\alpha \in [0,1]$. Output: a rejection decision.}
\begin{itemize}


\item[1.] For every $c \in \mathcal{W}_1$, let $i_{c}(g) := \text{argmin}_{i \in [m_c]}\{d(G_{ic},g)\}$ be the agent in $c$ whose rooted network $G_{ic}$ is closest to $g$ and $W_{c}(g) := (Y_{c}(g),G_{c}(g)) := (Y_{i_{c}(g)c},G_{i_{c}(g)c})$ be $i_{c}(g)$'s outcome and rooted network. Similarly define  $W_{c}(g') := (Y_{c}(g'),G_{c}(g'))$ for every $c \in \mathcal{W}_2$. 

\item[2.] Reorder $\{W_{c}(g)\}_{c \in \mathcal{W}_1}$ and $\{W_{c}(g')\}_{c \in \mathcal{W}_2}$ so that the entries are increasing in $d(G_c(g),g)$ and $d(G_c(g'),g')$ respectively. Denote the first $q$ elements of the reordered $\{W_c(g)\}_{c \in \mathcal{W}_1}$ 
\[W^*(g) := (W^*_1(g), W^*_2(g), \ldots, W^*_q(g))~,\]
where $W^*_c(g) = (Y^*_c(g), G^*_c(g))$. Similarly define $W^*(g')$. Collect the $2q$ outcomes of $W^*(g)$ and $W^*(g')$ into the vector
\[S_C := (S_{C,1}, ..., S_{C,2q}) := \left(Y^{*}_{1}(g), \ldots, Y^{*}_{q}(g), Y^{*}_{1}(g'), \ldots, Y^{*}_{q}(g')\right)~.\] 
\item[3.] Define the Cramer von Mises test statistic 
\[R(S_C) = \frac{1}{2q}\sum_{j=1}^{2q}\left(\hat{F}_{1}(S_{C,j}; S_C) - \hat{F}_{2}(S_{C,j}; S_C)\right)^2~,\]
where \[\hat{F}_{1}(y; S_C) = \frac{1}{q}\sum_{j = 1}^q \mathbbm{1}\{S_{C,j} \le y\} \hspace{2mm} \text{and} \hspace{2mm}  \hat{F}_{2}(y; S_C) = \frac{1}{q}\sum_{j = q + 1}^{2q} \mathbbm{1}\{S_{C,j} \le y\}~.\]
\item[4.] Let $\mathbf{H}$ be the set of all permutations $\pi = (\pi(1), \ldots, \pi(2q))$ of $\{1, ..., 2q\}$ and 
\[S_C^{\pi} = (S_{C,\pi(1)}, ..., S_{C,\pi(2q)})~.\]
Reject $H_{0}$ if the $p$-value  $p \le \alpha$ where $p := \frac{1}{|\mathbf{H}|}\sum_{\pi \in \mathbf{H}} \mathbbm{1}\{R(S_C^{\pi}) \ge R(S_C)\}$.
\end{itemize}
\end{algorithm}


The $p$-value in Step 4 may be difficult to compute when $\mathbf{H}$ is large. Theorem \ref{thm:perm_cest} below continues to hold if $\mathbf{H}$ is replaced by $\mathbf{\hat{H}}$ where $\mathbf{\hat{H}} = \{\pi_1, ..., \pi_B\}$, $\pi_1$ is the identity permutation and $\pi_2, ..., \pi_B$ are drawn independently and uniformly at random from $\mathbf{H}$. Such sampling is standard in the literature, see also \cite{canay2018approximate}, Remark 3.2.

In practice we partition $\{W_c\}_{c \in [C]}$ into two disjoint sets iteratively in the following way. We first select the community which contains the observation closest to $g$, add it to $\mathcal{W}_1$, and remove it from the pool of candidate communities. We then select the community which contains the observation closest to $g'$, add it to $\mathcal{W}_2$, and remove it from the pool of candidate communities. We continue this process, alternating between $g$ and $g'$, until the pool of candidate communities is exhausted.

Since ties are broken uniformly at random when determining the closest agent in Step 1 or reordering the vectors in Step 2, it is likely that our procedure introduces uncertainty in the $p$-value. For this reason, we recommend that practitioners assess the sensitivity of these choices on the resulting $p$-value. Practitioners could also consider formally aggregating several $p$-values obtained from different choices in Steps 1 and 2 using the methods outlined in Section 2.1 of \cite{diciccio2020exact}.

The test presented in Algorithm \ref{algo:test} is non-randomized in the sense that, once the collection $S_C$ is selected, the decision to reject the null hypothesis is a deterministic function of the data. This leads to a test which is potentially conservative. One could alternatively consider a non-conservative version of this test which is randomized. See \cite{lehmann2006testing}, Section 15.2.

\subsubsection{Asymptotic validity}
If the entries of $Y^*(g) = \left(Y^{*}_{1}(g), \ldots, Y^{*}_{q}(g)\right)$ and $Y^*(g') = \left(Y^{*}_{1}(g'), \ldots, Y^{*}_{q}(g')\right)$ were identically distributed to $h(g,U)$ and $h(g',U)$ respectively, then the test described in Algorithm 4.1 would control size in finite samples following standard arguments \citep[e.g.][Theorem 15.2.1]{lehmann2006testing}. However, since the rooted networks corresponding to $Y^*_j(g)$ and $Y^*_j(g')$ are not exactly $g$ and $g'$, such an argument can not be directly applied. 

Our assumptions instead imply that in an asymptotic regime where $q$ is fixed and $C\to\infty$, the entries of $Y^{*}(g)$ and $Y^{*}(g')$ are approximately equal in distribution to $h(g,U)$ and $h(g',U)$. A fixed $q$ rule is appropriate in our setting because the quality of the nearest neighbors to $g$ or $g'$ may degrade rapidly with $q$. Our simulation evidence in Appendix \ref{sec:simulations} suggests that setting $q = 2\lfloor \log (C) \rfloor$ works well in practice. 

We demonstrate that the test procedure in Algorithm 4.1 is asymptotically valid building on work by \cite{canay2017randomization,canay2018approximate}. Finite-sample behavior is examined via simulation in Appendix \ref{sec:simulations}. 



\begin{theorem}\label{thm:perm_cest}
Under Assumptions \ref{ass:sampling}, \ref{ass:pos_mass} and \ref{ass:dist_cont}, the test described in Algorithm  \ref{algo:test} is asymptotically ($C\to\infty$, $q$ fixed) level $\alpha$. 
\end{theorem}

\section{Empirical illustration}\label{sec:empirical_illustration}
We illustrate our local approach with a study of the influence of network structure on favor exchange in the  setting of \cite{jackson2012social}. In their work, \cite{jackson2012social} specify an equilibrium model of favor exchange where one agent is willing to perform a favor for another if there is a third agent who has a social connection to both agents and can monitor the exchange. Intuitively, the social norm is for agents to be altruistic and provide favors for each other and the third agent exerts social pressure on the other two agents to conform to this norm. Without this social pressure, agents have an incentive to ignore the social norm and not provide favors for each other. When such a norm enforcing third agent exists they say that the favor exchange relationship between the first two agents is supported. That is, formally, the relationship between $i$ and $j$ is supported if and only if there exists a $k \in \mathcal{I}$ such that $D_{ik} = D_{jk} = 1$.\footnote{ Strictly speaking, \cite{jackson2012social} only define support for pairs of agents that exchange favors. That is, they say that the relationship between two agents is supported if they exchange favors and their relationship is monitored. We instead say that a relationship is supported if it is monitored, regardless of whether or not the relevant agents exchange favors.  } The notion of support is extended to agents by counting the number of supported links adjacent to that agent and we refer to this statistic as agent support. That is, formally, agent $i$'s support is given by $\sum_{j\neq i}\mathbbm{1}\{\sum_{k}D_{ik}D_{jk} > 0\}$. 

Support contrasts alternative measures of social cohesion such as clustering. An agent's clustering coefficient measures the fraction of agent pairs connected to the agent that are also connected to each other, i.e. $\left(\sum_{i \neq j \neq k}D_{ij}D_{ik}D_{jk}\right)/\left(\sum_{i \neq j \neq k}D_{ij}D_{ik} \vee 1\right)$. One might think that high levels of agent clustering also drives favor exchange between agents. However, in the \cite{jackson2012social} model, the two are unrelated conditional on support. Intuitively, the decision for agents to exchange favors is made on the extensive margin in that it only depends on whether or not their relationship is monitored. Adding additional monitors may increase the level of clustering, but under the theory does not impact favor exchange all else equal. The authors take the existence of low-clustering high-support favor exchange networks in real-world network data as corroborating evidence for their model. 

We apply our local approach to evaluate the influence of support and clustering on the number of favors exchanged directly. There is no treatment assigned in this illustration, so we take $\mathcal{T}$ to be the empty set. Instead, we compare the average outcomes for collections of agents whose rooted networks are isomorphic to certain fixed configurations with various amounts of support and clustering described in Figure 2. The advantage of our approach relative to the empirical analysis of  \cite{jackson2012social} is that we do not assume that the social network only impacts the amount of favors exchanged by an agent through the level of support and clustering. Instead, the outcomes can depend on any feature of the social connections, as determined by the agent's rooted network. Following \cite{jackson2012social}, we consider data from 75 rural villages in Karnakata, India and construct our favor exchange variable (the outcome) and monitoring networks as described in their footnote 39.\footnote{Specifically, to construct the monitoring network we use their hedonic network, which documents friends and visitors. To define the number of favors exchanged we use the number of households in the village to which they report borrowing or lending money or kerorice. The data was originally collected by \cite{banerjee2013diffusion} to study the diffusion of information about a microfinance program.} Summary statistics on the number of households in each village are provided in Table \ref{tab:summary_stats_app}.

To be sure, the social network links are not randomly assigned in this setting and our causal interpretation hinges on the validity of our Assumption 4.1. In particular, we assume that any other driver of favor exchange in the village are unrelated to the structure of the monitoring network. To the extent that this exogeneity assumption is violated, the associations inferred from our methodology may not be causal. 

Below we estimate a policy effect and test the hypothesis of policy irrelevance for two pairs of rooted networks describing three different ways that agents can monitor each other in the hedonic network. The networks we chose are given by the ``cutlery ensemble'' of rooted networks depicted in Figure \ref{fig:rooted_hedonic}. In the knife network, the root agent $\alpha$ has one supported relationship and a clustering coefficient of $0$. In the fork network, the root agent $\beta$ has two supported relationships and a clustering coefficient of $0$. In the spoon network, the root agent $\gamma$ has three supported relationships and a clustering coefficient of $1$. We chose these networks for our illustration because they are simple to describe and have varying degrees of support and clustering.


\begin{table}[htbp]
    \centering
    \caption{Summary statistics for the number of households per village}
    \begin{tabular}{cccccccc}
        \toprule
        Total Num. & Mean  & Min & 25\% Quantile & 50\% Quantile & 75\% Quantile & Max \\
        \midrule
        13405& 179 & 70 & 140 & 167 & 218 & 326\\
        \bottomrule
    \end{tabular}
    \label{tab:summary_stats_app}
\end{table}

\begin{figure}
    \centering
            \begin{tikzpicture}[->,>= stealth,shorten >=1pt,auto,node distance=1.3cm,
        thick,main node/.style={circle,fill=blue!20,draw,minimum size=.4cm,inner sep=0pt]}]
\raisebox{8.5mm}{
      \node[main node] (2) {};
    \node[main node] (6) [ left of =2] {};
    \node[draw, fill=blue!20, shape=diamond, aspect=0.7, minimum height=0.4cm, inner sep=0pt] (5) [left of=6] {$\alpha$};

    \path[-]
    (6) edge node {} (2)   
    (5) edge node {} (6);
    }
\end{tikzpicture}
\hspace{20mm}
        \begin{tikzpicture}[->,>= stealth,shorten >=1pt,auto,node distance=1.3cm,
        thick,main node/.style={circle,fill=blue!20,draw,minimum size=.4cm,inner sep=0pt]}]

      \node[main node] (2) {};
    \node[main node] (6) [ left of =2] {};
    \node[draw, fill=blue!20, shape=diamond, aspect=0.7, minimum height=0.4cm, inner sep=0pt] (5) [above left of=6] {$\beta$};
    \node[main node] (1) [below left of =6] {};

    \path[-]
    (6) edge node {} (2)   
    (1) edge node {} (6)
    (5) edge node {} (6);
\end{tikzpicture}
\hspace{20mm}
     \begin{tikzpicture}[->,>= stealth,shorten >=1pt,auto,node distance=1.3cm,
        thick,main node/.style={circle,fill=blue!20,draw,minimum size=.4cm,inner sep=0pt]}]

      \node[main node] (2) {};
    \node[draw, fill=blue!20, shape=diamond, aspect=0.7, minimum height=0.4cm, inner sep=0pt] (6) [ above left of =2] {$\gamma$};
    \node[main node] (5) [right of=2] {};
    \node[main node] (1) [ below left of =2] {};

    \path[-]
    (6) edge node {} (2)
        	edge node {} (1)
    (2) edge node {} (5)
    	edge node {} (1);
\end{tikzpicture}

\caption{ These three networks constitute the ``cutlery ensemble.'' We use them to describe three different ways that agents can monitor each other in a social network. The first network on the left is the ``knife'' network rooted at agent $\alpha$. The second network in the middle is the ``fork'' network rooted at agent $\beta$. The third network on the right is the ``spoon'' network rooted at the agent $\gamma$. }\label{fig:rooted_hedonic}
\end{figure}

Our first policy comparison is the level of favor exchange in the fork network to that of the knife network. Specifically, we compare the distributions of $Y_{\alpha}$ and $Y_{\beta}$ where the two random variables refer to the number of favors performed by $\alpha$ and $\beta$ respectively. Intuitively, this comparison measures the effect of taking $\alpha$'s community as given by the knife network, adding an additional fourth agent, and connecting them to the center agent of the knife. This policy change increases the support of the root agent from $1$ to $2$ but does not impact the clustering coefficient of the root. Our second policy comparison is the level of favor exchange in the spoon network to that of the fork network (i.e. we compare the distributions of $Y_{\beta}$ and $Y_{\gamma}$). Intuitively, this comparison measures the effect of taking $\beta$'s community as given by the fork network and adding a social connection between agent $\beta$ and the other ``prong'' of the fork directly below $\beta$. This policy change increases the support of the root agent from $2$ to $3$ and the clustering coefficient of the root agent from $0$ to $1$.  Table \ref{tab:est_ci} reports the point estimates and confidence intervals associated with the average differences in favors for these policy counterfactuals using our $k$-nearest neighbors estimator. We find that overall moving from fork to spoon leads to a positive and statistically significant increase in the number of favors, whereas this is not the case when moving from knife to fork.


\begin{table}[htbp]
    \centering
    \caption{Estimates and confidence intervals of treatment effects}
    \begin{tabular}{cccc}
        \toprule
        & & $ E[Y_\beta] -  E[Y_\alpha]$ & $ E[Y_\gamma] - E[Y_\beta]$\\
        \midrule
        \multirow{3}{*}{$k=10$} & Est.&  0.14 & 0.46\\
        & 95\% CI&[-0.46,0.75] & [-0.21,1.13] \\
        & 90\% CI&[-0.37,0.65] &[-0.10,1.02] \\
        \midrule
        \multirow{3}{*}{$k=20$} &Est. & 0.08& 0.58\\
        & 95\% CI &[-0.24,0.40] & [0.22, 0.95]\\
        & 90\% CI &[-0.19,0.35] &[0.28,0.89] \\
        \midrule
        \multirow{3}{*}{$k=30$} & Est.&0.04 &0.76 \\
        & 95\% CI &[-0.19,0.26] &[0.50,1.02] \\
        & 90\% CI &[-0.15,0.22] &[0.54,0.98] \\
        \bottomrule
    \end{tabular}
    \label{tab:est_ci}
\end{table}

We also test the hypotheses that $Y_{\alpha}$ and $Y_{\beta}$ are equal in distribution, and $Y_{\beta}$ and $Y_\gamma$ are equal in distribution, using the approximate randomization test described in Section \ref{sec:dist_test}. For the test of $Y_{\alpha} =_d Y_{\beta}$, we obtain p-values of $1.000, 0.680, 0.582$ for $q = 5, 10, 20$ respectively, which we interpret as providing no evidence against the hypothesis.  For the test of  $Y_{\beta} =_d Y_{\gamma}$, we obtain p-values of $0.065, 0.014, 0.049$ for $q = 5, 10, 20$ respectively, so that we have statistical significance at the $10\%$ level.\footnote{We note that, particularly in the case of the spoon network, there exist multiple rooted networks which are equidistant to the target rooted network both within and across villages; Figures \ref{fig:closest_knife_knife_vs_fork_jackson_testing}--\ref{fig:closest_spoon_fork_vs_spoon_jackson_testing} in Appendix \ref{sec:robust_app} provide depictions of these networks. As a consequence, the (arbitrary) selection of such rooted networks introduces additional randomness into the computation of the $p$-value. Appendix \ref{sec:robust_app} discusses the robustness of our results to these random selections.}

Ultimately, we view these empirical results as challenging the view that clustering does not play a role in favor exchange. This is because the first policy comparison (between knife and fork) constitutes an increase in support holding clustering fixed, but we found no evidence of a change in favor exchange. The second policy comparison (between spoon and fork) constitutes both an increase in support and clustering  and we did find evidence of an increase in favor exchange. Our conclusion is that there may still be some role for clustering to explain altruism in real-world networks.\footnote{We thank Ben Golub for helping us interpret these results in the context of \cite{jackson2012social}'s model.}

\section{Conclusion}\label{sec:conclusion}
This paper proposes a new nonparametric modeling framework for causal inference under network interference. Rooted networks serve the role of effective treatments in that they index the ways in which the treatment assignments and network structure can influence the agent outcomes. We demonstrate our approach with an estimation strategy for average or distributional policy effects, a test for the hypothesis of policy irrelevance, and an empirical illustration studying the effect of network structure on social capital formation. 

Much work remains to be done, as indicated in the discussions of the assumptions and results in the sections above. Other potential directions for future work include considering the problem of policy learning under network interference \citep[see for example][]{ananth2020optimal, viviano2019policy, kitagawa2020should}, applying the framework to identify and estimate the parameters of a strategic network formation model \cite[as in for instance][]{de2018identifying}, or semiparametrically estimating average treatment effects for binary treatment \cite[as in, for instance][]{abadie2006large}. Ultimately we see our work as one step in a direction allowing for the more flexible econometric modeling of sparse network structures. 

\bibliographystyle{aer}
\bibliography{literature}

\appendix

\section{Proof of claims}
We document the following lemma which is a special case of Proposition 8.1 in \cite{biau2015lectures}. It is used to demonstrate Theorems \ref{thm:mse_bound} and \ref{thm:normal}. 
\begin{lemma}\label{lem:order_properties}
Let $(X_1, Y_1), (X_2, Y_2), \ldots, (X_n, Y_n)$ be an i.i.d sample taking values in $\mathbb{R}^2$.  Denote by $(X^*_1, Y^*_1), \ldots, (X^*_n, Y^*_n)$ the re-ordering of the data according to increasing values of $X_i$.  Let $\xi(X,Y)$ be some function of $(X,Y)$ such that $E[|\xi(X,Y)|] < \infty$ and define $r(X) = E[\xi(X,Y)|X]$. Then, conditional on $X_1, \ldots, X_n$, the random variables
\[(X^*_1, Y^*_1), (X^*_2, Y^*_2)\ldots, (X^*_n, Y^*_n)~,\]
are independent. Moreover, for each $1 \le i \le n$, 
\[E[\xi(X^*_i,Y^*_i)|X_1, \ldots, X_n] = r(X^*_i)~.\]
\end{lemma}
\subsection{Properties of $(\mathcal{G}, d)$}\label{sec:G_props}
The results of this section consider a more general notion of rooted network space than given in Section \ref{sec:local_formal}. Specifically, each agent $i$ is associated with $P$ covariates $C_i \in \mathbb{R}^{P}$ where the $p$th entry of $C_{i}$ is given by $c_{ip}$ and $\bold{C} = \{C_{i}\}_{i \in \mathcal{I}}$. This is as opposed to Section \ref{sec:local_formal} where each agent $i$ is associated with just one covariate  $T_{i} \in \mathbb{R}$ called the treatment assignment. The definition of the $\varepsilon$-isomorphism in Section \ref{sec:local_formal} used to define the rooted network distance in this more general case is modified so that $|c_{jp} - c_{f(j)p}| < \epsilon$ for all $p \in [P]$. Otherwise the definition of the rooted network space $(\mathcal{G},d)$ is the same.

\begin{proposition} The rooted network distance $d$ is a pseudometric on the set of locally-finite rooted networks.
\end{proposition}
\begin{proof} To be a pseudometric on the set of rooted networks, $d$ must satisfy (i) the identity condition $d(g,g) = 0$, (ii) the symmetry condition $d(g,g') = d(g',g)$, and (iii) the triangle inequality $d(g,g'') \leq d(g,g') + d(g',g'')$ for any three rooted networks $g,g',g''$. Conditions (i) and (ii) follow immediately from the definition of $d$. We demonstrate (iii) below. 

Fix an arbitrary $g,g',g''$ and denote the vertex sets, edge sets, and agent-specific covariates associated with $g$ by $(\mathcal{I}(g),D(g),\bold{C}(g))$. Suppose $d(g,g') + d(g',g'') = \eta$. Recall $\zeta(x) = (1+x)^{-1}$. Then by the definition of $d$, for every $\nu > 0$ there exist $\varepsilon,\varepsilon' \in \mathbb{R}_{+}$ and $r,r' \in \mathbb{Z}_{+}$  with  $\zeta(r) + \varepsilon + \zeta(r') + \varepsilon' < \eta + \nu$ such that for every $x, x' \in \mathbb{Z}_{+}$ there exists root-preserving bijections $f_x:\mathcal{I}(g^{r})^x\leftrightarrow \mathcal{I}({g'}^{r})^x$ and  $f'_{x'}:\mathcal{I}({g'}^{r'})^{x'}\leftrightarrow \mathcal{I}({g''}^{r'})^{x'}$ such that $D_{jk}(g^{r}) = D_{f_x(j)f_x(k)}({g'}^{r})$, $|c_{jp}(g^{r}) - c_{f_x(j)p}({g'}^{r})| \leq \varepsilon$, for every $j,k \in \mathcal{I}(g^{r})^{x}$, $p \in [P]$, and  $D_{jk}({g'}^{r'}) = D_{f'_{x'}(j)f'_{x'}(k)}({g''}^{r'})$,  $|c_{jp}({g'}^{r'}) - c_{f'_{x'}(j)p}({g''}^{r'})| \leq \varepsilon$, for every $j,k \in \mathcal{I}({g'}^{r'})^{x'}$, $p \in [P]$. 

Let $r'' = \min(r,r')$. We construct a root-preserving bijection $f''_{x''}:\mathcal{I}(g^{r''})^{x''} \leftrightarrow \mathcal{I}(g''^{r''})^{x''}$ for every $x'' \in \mathbb{Z}_+$ such that $D_{jk}(g^{r''}) = D_{f_{x''}''(j)f_{x''}''(k)}({g''}^{r''})$ and  $|c_{jp}(g^{r''}) - c_{f_{x''}''(j)p}({g''}^{r''})| \leq \varepsilon + \varepsilon'$ for every $j, k \in \mathcal{I}({g}^{r''})^{x''}$, $p \in [P]$. It follows from the definition of $d$ and $\zeta$ that 
\begin{align*}
d(g,g'') \leq \zeta(r'') + \varepsilon + \varepsilon' \leq \zeta(r) + \zeta(r') +  \varepsilon + \varepsilon' < \eta + \nu.
\end{align*}
Since $\nu > 0$ was arbitrary, (iii) follows.  

To construct such a bijection $f''_{x''}$, we rely on the fact that for any rooted network $g$ and $a,b \in \mathbb{R}_{+}$, $\mathcal{I}(g^{a})^{b} = \mathcal{I}(g^{b})^{a} =  \mathcal{I}(g^{\min(a,b)})$. This implies that for any $x'' \le r''$, $\mathcal{I}(g^r)^{x''} = \mathcal{I}(g^{r''})^{x''}$, $\mathcal{I}(g'^r)^{x''} = \mathcal{I}(g'^{r'})^{x''} = \mathcal{I}(g'^{r''})^{x''}$, and $\mathcal{I}(g''^{r'})^{x''} = \mathcal{I}(g''^{r''})^{x''}$. For any $x'' > r''$, $\mathcal{I}(g^{r''})^{x''} = \mathcal{I}(g^{r''})^{r''}$ and $\mathcal{I}(g''^{r''})^{x''} = \mathcal{I}(g''^{r''})^{r''}$. 

These equalities ensure that the root-preserving bijection $f''_{x''} = f'_{x''}\circ f_{x''}$ is well-defined for any $x'' \in \mathbb{Z}_{+}$ because they imply that the domain of $f'_{x''}$ and the range of $f_{x''}$ are both equal to $\mathcal{I}(g'^{r''})^{x''}$. Furthermore by construction $D_{jk}(g^{r''}) = D_{f_{x''}''(j)f_{x''}''(k)}({g''}^{r''})$ for every $j, k \in \mathcal{I}(g^{r''})^{x''}$, and by the triangle inequality for $|\cdot|$ we obtain $|c_{jp}(g^{r''}) - c_{f_{x''}''(j)p}({g''}^{r''})| \leq \varepsilon + \varepsilon'$ for every $j \in \mathcal{I}(g^{r''})^{x''}$ and $p \in [P]$. 
\end{proof}

\begin{proposition}$(\mathcal{G},d)$ is a separable metric space.
\end{proposition}

\begin{proof}To be separable, $\mathcal{G}$ must have a countable dense subset. Let $\tilde{\mathcal{G}}$ be the subset of rooted networks with finite vertex sets and rational-valued covariates. $\tilde{\mathcal{G}}$ is countable because it is the countable union of countable sets. The claim follows by showing that $\tilde{\mathcal{G}}$ is dense in $\mathcal{G}$.

To demonstrate that $\tilde{\mathcal{G}}$ is dense in $\mathcal{G}$, we fix an arbitrary rooted network $g \in \mathcal{G}$ and $\eta > 0$, and show that there exists a rooted network $g' \in \tilde{\mathcal{G}}$ such that $d(g,g') \leq \eta$. Choose any $r \in \mathbb{Z}_+$, $\varepsilon \in \mathbb{R}_{++}$ such that $\zeta(r) + \varepsilon \leq \eta$. Define $g'$ to be the rooted network  on the same set of agents and edge weights as $g^{r}$, but with rational-valued covariates chosen to be uniformly within $\varepsilon$ of their analogues in $g$. Since $g$ is locally finite, $g^{r}$ has a finite vertex set, and so $g' \in \tilde{\mathcal{G}}$. Furthermore $g'$ and $g^{r}$ are $\varepsilon$-isomorphic by construction. It follows that $d(g,g') \leq \zeta(r) + \varepsilon \leq \eta$. \end{proof}

\begin{proposition}$(\mathcal{G},d)$ is a complete metric space.
\end{proposition}

\begin{proof}To be complete, every Cauchy sequence in $\mathcal{G}$ must have a limit that is also in $\mathcal{G}$. Let $g_{n}$ be a Cauchy sequence in $\mathcal{G}$. That is, for every $\eta > 0$ there exists an $m \in \mathbb{N}$ such that $d(g_{m'},g_{m''}) \leq \eta$ for every $m',m'' \in \mathbb{N}$ with $m',m'' \geq m$. 

Fix an arbitrary $r \in \mathbb{Z}_+$ and let $g_{n}^{r}$ be the sequence of rooted networks formed by taking $g_{n}$ and replacing each rooted network in the sequence with its truncation at radius $r$. Since $\mathcal{G}$ is locally finite, the vertex set of each element of $g_{n}^{r}$ is finite. Let $N_{n}^{r} := |\mathcal{I}(g_{n}^{r})|$ denote the sequence of vertex set sizes. Since $g_{n}$ is a Cauchy sequence, $g_{n}^{r}$ is a Cauchy sequence, and so the entries of $N_{n}^{r}$ must also be uniformly bounded. Let $N^{r} := \limsup_{n \to \infty}N_{n}^{r} < \infty$.  

Let $g_{s(n)}^{r}$ be an infinite subsequence of $g_{n}^{r}$ such that $N_{s(n)}^{r} = V$ for some $V \in \mathbb{N}$ and every $n \in \mathbb{N}$. Such a subseqence and choice of $V$ must exist because $N_{n}^{r}$ takes only finitely many values. Since $g_{n}^{r}$ is a Cauchy sequence, $g_{s(n)}^{r}$ is a Cauchy sequence, and so for every $\eta > 0$ there exists an  $m \in \mathbb{N}$ such that for every  $m',m'' \in \mathbb{N}$ with $m',m'' \geq m$ there exists a bijection $f: \mathcal{I}(g_{s(m')}^{r}) \leftrightarrow \mathcal{I}(g_{s(m'')}^{r})$ with $D_{jk}(g_{s(m')}^{r}) = D_{f(j)f(k)}(g_{s(m'')}^{r})$ for every $j,k \in \mathcal{I}(g_{s(m')}^{r})$ and $|c_{jp}(g_{s(m')}^{r}) - c_{f(j)p}(g_{s(m'')}^{r})| \leq \eta$ for every $j \in \mathcal{I}(g_{s(m')}^{r})$ and $p \in [P]$. 



Let the rooted networks in the sequence $\tilde{g}_{s(n)}^{r}$ be $0$-isomorphic to those in $g_{s(n)}^{r}$ but with agents ordered in a canonical way such that for any $j \in [V]$ and $p \in[P]$ the corresponding sequences of covariates $c_{jp}(\tilde{g}_{s(n)}^{r})$ are Cauchy sequences in $\mathbb{R}$ and the corresponding sequence of edges $D_{jk}(\tilde{g}_{s(n)}^{r})$ are constant. The resulting sequence of covariate matrices are Cauchy sequences with respect to the max-norm in $\mathbb{R}^{V\times P}$, and so converge to a unique limit $C$ because finite-dimensional Euclidean space is complete. Let $E = D_{jk}(\tilde{g}_{s(n)}^{r})$ be the constant (in $n$) edge set.


Let $g_{\infty}^{r}$ be the rooted network with vertex set $[V]$, edge set $E$ and covariate set $C$. Then $g_{\infty}^r$ is a limit for the sequence $g_{s(n)}^{r}$ by construction. It is also a limit for the sequence $g_{n}^{r}$ because $g_{s(n)}^{r}$ is a subsequence of $g_{n}^{r}$ and $g_{n}^{r}$ is a Cauchy sequence. Since $r \in \mathbb{Z}_{+}$ was arbitrary, $g_{n}^{r}$ converges to $g_{\infty}^{r} \in \mathcal{G}$ for any $r \in \mathbb{Z}_{+}$. Define  $g$ to be the rooted network such that $g^r = g_\infty^r$ for every $r \in \mathbb{Z}_+$. Then $g$ is the limit of $g_{n}$ and the claim follows. 
\end{proof}

%

\subsection{Theorem \ref{thm:mse_bound} and Corollary \ref{cor:mse_bound}}
{\bf Proof of Theorem \ref{thm:mse_bound}}
\begin{proof}
Our proof follows the structure of that in \cite{doring2017rate}, Theorem 6.

\noindent Let $\bar{h}(g) = E[\hat{h}(g)|D_1(g), \ldots D_C(g)]$.
Decomposing the mean-squared error gives
\begin{equation}\label{eq:MSE_decomp}
E[(\hat{h}(g) - h(g))^2] = E[(\hat{h}(g) - \bar{h}(g))^2] + E[(\bar{h}(g) - h(g))^2]~,
\end{equation}
and the claim follows by bounding $E[(\hat{h}(g) - \bar{h}(g))^2] \leq \sigma^2/k$ and $E[(\bar{h}(g) - h(g))^2] \leq E[\varphi_g(U_{(k,T)})^2]$. 

Before presenting each bound, we derive a series of useful identities. First, we show
\begin{equation}\label{eq:r_def}
E[\bar{Y}_c(g)|D_c(g)] = f(D_c(g))~,
\end{equation}
where
\[f(d):= E\left[\frac{1}{|\mathcal{I}_c(g)|}\sum_{i \in \mathcal{I}_c(g)}h(G_{ic})\Bigg|D_c(g) = d\right]~.\]
To see this, note that by the law of iterated expecations and the definitions of $D_c(g)$ and $\mathcal{I}_c(g)$:
\[E[\bar{Y}_c(g)|D_c(g)] = E\left[\frac{1}{|\mathcal{I}_c(g)|}\sum_{i \in \mathcal{I}_c(g)}E[Y_{ic}|G_{1c}, \ldots, G_{m_cc}]\Bigg|D_c(g)\right]~,\]
and by Assumption \ref{ass:sampling},
\[E[Y_{ic}|G_{1c}, \ldots, G_{m_cc}] = h(G_{ic})~.\]
Next, it follows by \eqref{eq:r_def} and Lemma \ref{lem:order_properties} that 
\begin{equation}\label{eq:bar_Y}
E[\bar{Y}^*_j(g)|D_1(g), \ldots  D_C(g)] = f(D^*_j(g))~,
\end{equation}
so that by the definition of $\bar{h}(g)$, we have
\begin{equation}\label{eq:bar_h}
\bar{h}(g) = \frac{1}{k}\sum_{j = 1}^k f(D^*_j(g))~.
\end{equation}

Next, we derive the following identity:
\begin{equation}\label{eq:phi_bound}
E\left[\frac{1}{|\mathcal{I}_c(g)|}\sum_{i \in \mathcal{I}_c(g)}|h(g) - h(G_{ic})| \Bigg|D_c(g)\right] \le \phi_g(D_c(g))~.
\end{equation}
To see this, note that by Assumption \ref{ass:m_smooth}, $|h(g) - h(G_{ic})| \le \phi_g(d(g, G_{ic}))$ and by the definition of $\mathcal{I}_c(g)$, $d(g,G_{ic}) = D_c(g)$ for $i \in \mathcal{I}_c(g)$. Hence
\[E\left[\frac{1}{|\mathcal{I}_c(g)|}\sum_{i \in \mathcal{I}_c(g)}|h(g) - h(G_{ic})| \Bigg|D_c(g)\right] \le E\left[\frac{1}{|\mathcal{I}_c(g)|}\sum_{i \in \mathcal{I}_c(g)}\phi_g(d(g,G_{ic})) \Bigg|D_c(g)\right] = \phi(D_c(g))~.\]
Applying \eqref{eq:phi_bound} with the definition of $f(\cdot)$, we obtain by the triangle inequality that 
\[|h(g) - f(d)| \le \phi_g(d)~,\]
from which it follows that 
\begin{equation}\label{eq:phi_bound2}
|h(g) - f(D^*_j(g))| \le \phi_g(D^*_j(g))~.
\end{equation} 

Now we derive the first bound in \eqref{eq:MSE_decomp}. To that end, 
\begin{align*}
E[(\hat{h}(g) - \bar{h}(g))^2|D_1(g), \ldots, D_C(g)]
&= E\left[\left(\frac{1}{k}\sum_{j = 1}^k\left(\bar{Y}^{*}_{j}(g) - f(D^{*}_{j}(g))\right)\right)^2 \Big| D_1(g), \ldots, D_C(g)\right] \\
&=  \frac{1}{k^{2}}\sum_{j=1}^{k}E\left[\left(\bar{Y}^{*}_{j}(g) - f(D^{*}_{j}(g))\right)^2 \Big| D_1(g), \ldots, D_C(g)\right] \\
&\le  \frac{1}{k^2}\sum_{j=1}^kE\left[\left(\bar{Y}^*_j(g)\right)^2|D_1(g), \ldots, D_C(g)\right] \le \frac{\bar{\sigma}^2}{k}~,
\end{align*}
where the first equality follows by \eqref{eq:bar_h}, the second follows from the independence of the pairs $\{(\bar{Y}^*_j(g), D^*_j(g))\}_{1 \le j \le k}$ conditional on $D_1(g), \ldots, D_C(g)$ (by Lemma \ref{lem:order_properties}), the first inequality follows from the definition of a conditional variance (after noting \eqref{eq:bar_Y}), and  final inequality follows from the following derivation and Assumption \ref{ass:sampling}:
\begin{align*}
E\left[\bar{Y}^2_c(g)|D_c(g)\right] &= E\left[\left(\frac{1}{|\mathcal{I}_c(g)|}\sum_{i \in \mathcal{I}_c(g)} Y_{ic}\right)^2 \Bigg|D_c(g)\right] \\
&\le E\left[\frac{1}{|\mathcal{I}_c(g)|}\sum_{i \in \mathcal{I}_c(g)} Y_{ic}^2 \Bigg|D_c(g)\right]  \\
&=  E\left[\frac{1}{|\mathcal{I}_c(g)|}\sum_{i \in \mathcal{I}_c(g)} E[Y_{ic}^2|G_{1c},\ldots,G_{m_cc}] \Bigg|D_c(g)\right]  \\
&\le  E\left[\frac{1}{|\mathcal{I}_c(g)|}\sum_{i \in \mathcal{I}_c(g)} \bar{\sigma}^2 \Bigg|D_c(g)\right]  = \bar{\sigma}^2 ~,
\end{align*}
where the first equality follows by defininition, the first inequality is Jensen's, the second equality follows by the law of iterated expectations, and the final inequality by Assumption  \ref{ass:sampling} and the fact that $E[h(g,U_{ic})^2] \le \bar{\sigma}^2$ for all $g \in \mathcal{G}$ under Assumption \ref{ass:var}.

Next we derive the second bound in \eqref{eq:MSE_decomp}. Note that:
\begin{align*}
(\bar{h}(g) - h(g))^2 &\le \left(\frac{1}{k}\sum_{j=1}^k\left|h(g) - f(D^{*}_{j}(g))\right|\right)^2\\
&\le \left(\frac{1}{k}\sum_{j=1}^k \phi_g(D_j^*(g))\right)^2 \\
&\le \phi_g(D_k^*(g))^2 \\
&\le \varphi_g(\psi_g(D_k^*(g)))~,
\end{align*}
where the first inequality follows by \eqref{eq:bar_h} and the triangle inequality, the second by \eqref{eq:phi_bound2}, the third by the definition of $\phi_g(\cdot)$ and $D^*_j(g)$, and the final inequality by the definition of the upper generalized inverse. 
Under Assumption \ref{ass:no_ties}, the probability integral transform implies
\[\psi_g(D_k^*(g)) \,{\buildrel d \over =}\, U_{(k,C)}~,\]
where $U_{(k,C)}$ is the $k$th order statistic from a sequence of $C$ independent and identically distributed standard uniform random variables, and so is distributed $Beta(k,C+1-k)$. As a result
\[\varphi_g\left(\psi_g(D^*_k(g))\right)^2 \,{\buildrel d \over =}\, \varphi_g(U_{(k,C)})^2\]
and so
\[E[(\bar{h}(g) - h(g))^2] \le E[\varphi_g(U_{(k,C)})^2]~.\]
\end{proof}

\noindent {\bf Proof of Corollary \ref{cor:mse_bound}}
\begin{proof}
This follows immediately from the inequality $(a - b)^2 \le 2(a^2 + b^2)$ and Theorem \ref{thm:mse_bound}.
\end{proof}

\subsection{Theorem \ref{thm:normal} and Corollary \ref{cor:normal}}
\noindent To prove Theorems \ref{thm:normal} and \ref{thm:variance}, we require the following preliminary lemma:
\begin{lemma}\label{lem:order_LLN}
Maintain the Assumptions of Theorem \ref{thm:normal}, let 
\[m_t(d) := E[\bar{Y}^t_c(g)|D_c(g)=d]~,\]
for $t \in \{1,2\}$. Then
\[\frac{1}{k}\sum_{j=1}^k |m_t(D^*_j(g)) - m_t(0)| \xrightarrow{p} 0~,\]
as $C \rightarrow \infty$, $k/C \rightarrow 0$. 
\end{lemma}
\begin{proof}
Since $m_t(d)$ is continuous in a neighborhood of zero by assumption, let $\omega(h):[0,\delta] \rightarrow \mathbb{R}_+$ define a modulus of continuity at zero for $m_t(d)$ for some $\delta > 0$. That is, $\omega(\cdot)$ is non-decreasing, continuous at zero, and satisfies
\[|m_t(d) - m_t(0)| \le \omega(|d|)~,\]
for $|d| \le \delta$. Moreover, since $m_t(d)$ is bounded by assumption,
\[|m_t(d) - m_t(0)| \le M'~,\]
for some constant $M' > 0$, for all $d$. By the definition of $D^*_j(g)$ for $j = 1, \ldots, k$ and the properties of the modulus of continuity:
\begin{align*}
\frac{1}{k}\sum_{j=1}^k \left| m_t(D^*_j(g)) - m(0) \right| 
&\le \frac{1}{k} \sum_{j=1}^k \omega(D^*_j(g)) I\{ D_j^*(g) \le \delta \} 
+ M \frac{1}{k} \sum_{j=1}^k I\{ D_j^*(g) > \delta \} \\
&\le \omega(D^*_k(g)) + M I\{ D^*_k(g) > \delta \} ~.
\end{align*}
The result then follows by the continuous mapping theorem once we argue that $D^*_k(g) \xrightarrow{p} 0$ as $C \rightarrow \infty$, $k/C \rightarrow 0$. To that end, note that for any $\epsilon > 0$, the events
\[\left\{D^*_k(g) > \epsilon\right\} = \left\{\frac{1}{C}\sum_{c = 1}^CI\{D_c(g) < \epsilon\} < \frac{k}{C}\right\}~.\]
By the law of large numbers, 
\[\frac{1}{C}\sum_{c=1}^CI\{D_c(g) < \epsilon\} \xrightarrow{p} P(D_c(g) < \epsilon) > 0 ~,\]
where the last inequality follows by Assumption \ref{ass:pos_mass}. Thus, as $k/C \rightarrow 0$, 
\[P(D_k^*(g) > \epsilon) \rightarrow 0~,\]
as desired.
\end{proof}

\noindent {\bf Proof of Theorem \ref{thm:normal}}
\begin{proof}
Borrowing the notation from the proof of Theorem \ref{thm:mse_bound}, consider the following decomposition:
\[\sqrt{k}(\hat{h}(g) - h(g)) = \sqrt{k}(\hat{h}(g) - \bar{h}(g)) + \sqrt{k}(\bar{h}(g) - h(g))~.\]
First, we argue that under the stated assumptions the second term converges in probability to zero. To see this, note that it suffices to show that 
\[E[(\sqrt{k}(\bar{h}(g) - h(g)))^2] \rightarrow 0~.\]
By the characterization of the squared bias obtained in the proof of Theorem \ref{thm:mse_bound}, it follows immediately that
\[E[(\sqrt{k}(\bar{h}(g) - h(g)))^2] \le kE[\varphi_g(U_{(k,C)})^2] \rightarrow 0~,\]
where the convergence to zero follows by assumption. It thus remains to characterize the limiting distribution of the first term. We proceed by studying its limiting distribution conditional on $D_1(g), \ldots, D_C(g)$. Re-writing:
\[\sqrt{k}(\hat{h}(g) - \bar{h}(g)) = \frac{1}{\sqrt{k}}\sum_{j=1}^k\left(Y^*_j(g) - E[Y^*_j(g)|D_1(g), \ldots, D_C(g)]\right)~.\] 
by Lemma \ref{lem:order_properties}, conditional on $D_1(g), D_2(g), \ldots D_C(g)$, 
\[\left((D_1^*(g),Y^*_1(g)), \ldots, (D_k^*(g),Y_k^*(g))\right)\] are independent,
\[\text{var}\left(\bar{Y}_j^*(g) - E[\bar{Y}^*_j(g)|D_1(g), \ldots, D_C(g)]|D_1(g), \ldots, D_C(g)\right) = \sigma^2_g(D^*_j(g))~,\]
where $\sigma^2_g(d) = \text{var}(\bar{Y}_c(g)|D_c(g) = d)$, and 
\[E[(\bar{Y}_j^*(g) - E[\bar{Y}^*_j(g)|D_1(g), \ldots, D_C(g)])^3|D_1(g), \ldots D_C(g)] \le M'~,\]
for some $M' < \infty$, where the last inequality follows by our assumption of uniformly bounded conditional fourth (and thus third) moments. 
Let
\[F_k(t) = P\left(\frac{\frac{1}{\sqrt{k}}\sum_{j=1}^k\left(Y^*_j(g) - E[Y^*_j(g)|D_1(g), \ldots, D_C(g)]\right)}{\sqrt{\frac{1}{k}\sum_{j=1}^k\sigma^2_g(D^*_j(g))}} \le t\Bigg|D_1(g), \ldots D_C(g)\right)~.\] 
Then by the Berry-Eseen theorem applied conditionally on $D_1(g), D_2(g), \ldots D_C(g)$, along with the above inequality
\[\sup_{t \in \mathbb{R}} |F_k(t) - \Phi(t)| \le \frac{\gamma M'/\sqrt{k}}{\left(\frac{1}{k}\sum_{j=1}^k\sigma^2_g(D^*_j(g))\right)^{3/2}}~,\]
for some constant $\gamma > 0$. 
Next, we show $\frac{1}{k}\sum_{j=1}^k\sigma^2_g(D^*_j(g))  \xrightarrow{p} \sigma^2_g(0)$. To see this, note that
\[\sigma^2_g(d) = m_2(d) - m_1(d)^2~,\]
where $m_1$ and $m_2$ are defined as in the statement of Lemma \ref{lem:order_LLN}. By Lemma \ref{lem:order_LLN}, 
\[\frac{1}{k}\sum_{j=1}^km_2(D^*_j(g)) \xrightarrow{p} m_2(0)~.\]
By the triangle inequality, the identity $a^2 - b^2 = (a+b)(a-b)$ and the assumption of uniformly bounded fourth (and thus first) moments,
\[\left|\frac{1}{k}\sum_{j=1}^km_1(D^*_j(g))^2 - m_1(0)^2\right| \le C \frac{1}{k}\sum_{j=1}^k\left|m_1(D^*_j(g)) - m_1(0)\right| \xrightarrow{p} 0~,\]
for some constant $C< \infty$, where the convergence in probability follows from Lemma \ref{lem:order_LLN}.
Putting both of these together, we obtain \[\frac{1}{k}\sum_{j=1}^k\sigma^2_g(D^*_j(g)) \xrightarrow{p} \sigma^2_g(0)~,\]
so that for every $t \in \mathbb{R}$,
 \[|F_k(t) - \Phi(t)| \xrightarrow{p} 0~.\]
It then follows by the dominated convergence theorem and Slutsky's theorem that
\[\sqrt{k}(\hat{h}(g) - \bar{h}(g)) \xrightarrow{d} N(0, \sigma^2_g(0))~,\]
as desired.
\end{proof}

\noindent {\bf Proof of Corollary \ref{cor:normal}}
\begin{proof}
This follows immediately from the continuous mapping theorem and Theorem \ref{thm:normal}.
\end{proof}

\subsection{Theorem \ref{thm:variance}}
\begin{proof}
Note that by definition,
\[\hat{\sigma}^2_g = \frac{1}{k}\sum_{j=1}^k\bar{Y}_j^*(g)^2 - \left(\frac{1}{k}\sum_{j=1}^k\bar{Y}_j^*(g)\right)^2~,\]
we show that 
\[ \frac{1}{k}\sum_{j=1}^k\bar{Y}_j^*(g)^2 \xrightarrow{p} m_2(0)~,\]
where $m_2(\cdot)$ is defined in the statement of Lemma \ref{lem:order_LLN}. The argument that the second term converges to $m_1(0)^2$ follows similarly.
By Lemmas \ref{lem:order_properties} and \ref{lem:order_LLN},
\[E\left[\frac{1}{k}\sum_{j=1}^k\bar{Y}_j^*(g)^2 \middle| D_1(g), \ldots, D_C(g)\right] = \frac{1}{k}\sum_{j=1}^km_2(D^*_j(g)) \xrightarrow{p} m(0)~.\]
By Chebyshev's inequality conditional on $D_1(g), \ldots, D_C(g)$, for any $\epsilon > 0$, 
\begin{align*}
P\left( \left| \frac{1}{k} \sum_{j=1}^k \bar{Y}_j^*(g)^2 
- \frac{1}{k}\sum_{j=1}^km_2(D^*_j(g)) \right| 
> \epsilon \, \middle| \, D_1(g), \ldots, D_C(g) \right) 
&\leq \frac{\sum_{j=1}^k E\left[ \bar{Y}_j^*(g)^4 
\, \middle| \, D_1(g), \ldots, D_C(g)\right]}{k^2 \epsilon^2} \\
&\leq \frac{M}{k \epsilon^2} \rightarrow 0~,
\end{align*}
where the first inequality follows from Lemma \ref{lem:order_properties} and the second inequality from Lemma \ref{lem:order_properties} and Assumption \ref{ass:normal_regular}.
The result then follows from the dominated convergence theorem.
\end{proof}

\subsection{Theorem \ref{thm:perm_cest}}
We verify the assumptions of Theorem 4.2 in \cite{canay2018approximate}. Specifically, we verify their Assumption 4.5 in the case where $h(g,U)$ is continuous and their Assumption 4.6 in the case where $h(g,U)$ is discrete. Their assumptions 4.5 and 4.6 parts (i), (ii) and (iii) follow from our Lemma \ref{lem:verify_perm} below along with our Assumption  \ref{ass:dist_cont} that $h(g,U)$ is either continuous or discrete. Assumption 4.5 part (iv) follows from our choice of test statistic in Step 4 of Algorithm 4.1. 

Our proof of Lemma \ref{lem:verify_perm} relies on the following Lemma \ref{lem:order_OK} which is a modification of Lemma \ref{lem:order_properties} (and Proposition 8.1 in \cite{biau2015lectures}).
\begin{lemma}\label{lem:order_OK}
Let Assumption \ref{ass:sampling} hold and for any measurable $f:\mathbb{R}\times\mathcal{G} \rightarrow \mathbb{R}$ define \[r(g) = E[f(h(g,U_{ic}),g)]~.\] Then for any $g \in \mathcal{G}$, the entries of 
\[\left(W^{*}_{1}(g), \ldots, W^{*}_{C}(g)\right)\]
are independent conditional on $G_1(g), \ldots, G_C(g)$ and for every $1 \leq j \leq C$ 
\[E[f(W^{*}_{j}(g))|G_1(g), \ldots G_C(g)] = r(G^{*}_{j}(g))~.\]
\end{lemma}
\begin{proof}
Proposition 8.1 of \cite{biau2015lectures} directly implies the first claim that the elements of $\left(W^{*}_{1}(g), \ldots, W^{*}_{C}(g)\right)$ are independent conditional on $G_1(g), \ldots, G_C(g)$. 
For the second claim,  note that by a similar argument to the second statement in Proposition 8.1 of \cite{biau2015lectures},
\[E[f(W^*_{j}(g)) | G_1(g), \ldots, G_C(g)] = \tilde{r}_g(G^*_{j}(g))~,\]
where $\tilde{r}_g(\tilde{g}) = E[f(W_c(g)) | G_c(g) \simeq_0 \tilde{g}]$. Next we show that $\tilde{r}_g(\tilde{g}) = r(\tilde{g})$ (up to $\nu$-null sets, where $\nu$ is the pushforward measure induced by $G_c(g)$). Once we have shown that, then it will follow that $\tilde{r}_g(G^*_{j}(g)) = r(G^*_{j}(g))$ (since the pushforward measure induced by $G^*_j(g)$ is dominated by $\nu$ by construction) which demonstrates the claim.

Let $\tau$ be the random index such that $W_{\tau,c} = W_c(g)$, then 
\begin{align*}
E[f(W_c(g))|G_{1c}, \ldots, G_{m_cc}]
&= \sum_{i=1}^{m_c}{\bf 1}\{\tau = i\}E[f(W_{i,c})|G_{1c}, ..., G_{m_cc}] \\
&= \sum_{i=1}^{m_c}{\bf 1}\{\tau = i\}E[f(h(G_{ic},U_{ic}), G_{ic})|G_{1c}, \ldots, G_{m_cc}] \\
&=  \sum_{i=1}^{m_c}{\bf 1}\{\tau = i\}E[f(h(G_{ic}, U), G_{ic})|G_{ic}] \\
&= \sum_{i=1}^{m_c}{\bf 1}\{\tau = i\}r(G_{ic}) \\
&= r(G_c(g))~,
\end{align*}
where the first equality follows from the fact that $\tau$ is a function of $G_{1c}, ... G_{m_cc}$, the second follows from the definition of $W_{ic}$, and the third and fourth equalities follow from Assumption \ref{ass:sampling}(iii). By the law of iterated expectations and the fact that $G_c(g)$ is a measurable function of $G_{1c}, \ldots G_{m_cc}$ it follows that
\[\tilde{r}_g(G_c(g)) = E[E[f(W_c(g))|G_{1c}, ..., G_{m_cc}]|G_c(g)] = r(G_c(g))~,\]
and so $\tilde{r}_g(\tilde{g}) = r(\tilde{g})$ (up to $\nu$-null sets).
\end{proof}



\begin{lemma}\label{lem:verify_perm}
Under Assumptions \ref{ass:sampling}, \ref{ass:pos_mass} and \ref{ass:dist_cont} and the null hypothesis of policy irrelevance (\ref{eq:null})
\[S_C \xrightarrow{d} S\]
where $S = (S_1, ..., S_{2q})$ is a random vector with independent and identically distributed entries equal in distribution to $h(g,U)$. For any permutation $\pi \in \mathbf{H}$
\[S^{\pi} \,{\buildrel d \over =}\, S\]
so that we satisfy  (i) and (ii) of  Assumptions 4.5 and 4.6 of \cite{canay2018approximate}. 
\end{lemma}

\begin{proof}
Lemma \ref{lem:order_OK} implies that the entries of
\[(W^{*}(g),W^{*}(g')) := \left(W_{1}^{*}(g), \ldots, W_{q}^{*}(g), W_{1}^{*}(g'), \ldots, W_{q}^{*}(g')\right)\]
are independent conditional on $\{G_c(g)\}_{c \in \mathcal{D}_{1}}$ and $\{G_c(g')\}_{c \in \mathcal{D}_{2}}$, and that the conditional distribution functions of $Y_{j}^{*}(g)$ and $Y_{j}^{*}(g')$ are given by $h_{y}(G_{j}^{*}(g))$ and $h_{y}(G_{j}^{*}(g'))$ respectively. It follows by the law of iterated expectations that
\begin{align*}
P\left(Y_{1}^{*}(g)\le y_1, ..., Y_{q}^{*}(g)\le y_q, Y_{1}^{*}(g')\le y_{q+1}, \ldots, Y_{q}^{*}(g')\le y_{2q} \right) \\
= E\left[\displaystyle\prod_{j=1}^q h_{y_{j}}(G_{j}^{*}(g)) \displaystyle\prod_{j=1}^q h_{y_{j+q}}(G_{j}^{*}(g'))\right]~.
\end{align*}

We first show that Assumption \ref{ass:pos_mass} implies that $G_{j}^{*}(g) \xrightarrow{p} g$ and $G_{j}^{*}(g') \xrightarrow{p} g'$ for every $j$ as $C \rightarrow \infty$. To see this, fix $\epsilon > 0$ and write
\[\left\{d\left(G_{j}^{*}(g), g\right) > \epsilon\right\} = \left\{\frac{1}{C}\sum_{c = 1}^C\mathbbm{1}\{G_c^{*}(g) \in B_{g, \epsilon}\} < \frac{j}{C}\right\}~,\]
where $B_{g, \epsilon} = \{\tilde{g} \in \mathcal{G}: d(g, \tilde{g}) \le \epsilon\}$. By the law of large numbers and Assumption \ref{ass:pos_mass}, 
\[\frac{1}{C}\sum_{c = 1}^C \mathbbm{1}\{G_c^{*}(g) \in B_{g, \epsilon}\} \xrightarrow{p} P\left(G_c^{*}(g) \in B_{g, \epsilon}\right) > 0~,\]
and since $j/C \leq q/C \rightarrow 0$, it follows that $P\left(d\left(G_{j}^{*}(g), g\right) > \epsilon\right) \rightarrow 0$
as $C \rightarrow \infty$. 

Assumption \ref{ass:dist_cont} and the continuous mapping theorem imply that 
\[ h_{y}(G_{j}^{*}(g)) \xrightarrow{p}  h_{y}(g) \text{ and } h_{y}(G_{j}^{*}(g')) \xrightarrow{p}  h_{y}(g')~,\]
for every $j$ and $y \in \mathbb{R}$, and so it follows from the dominated convergence theorem that
\[E\left[\displaystyle\prod_{j=1}^q h_{y_{j}}(G_{j}^{*}(g)) \displaystyle\prod_{j=1}^q h_{y_{j+q}}(G_{j}^{*}(g'))\right] 
\rightarrow \prod_{j=1}^q h_{y_{j}}(g) \displaystyle\prod_{j=1}^q  h_{y_{j+q}}(g')~.\]
The claim follows from the fact that $h_{y}(g) = h_{y}(g')$ under the null hypothesis (\ref{eq:null}). 
\end{proof}

\section{Interpreting the bias}\label{sec:bias_disc}
The bound on estimation error described in Theorem \ref{thm:mse_bound} of Section \ref{sec:mse_bound} contains a variance term and a bias term. The variance term is standard. The bias term is new and so in this section we characterize how its features depend on the components $\phi_{g}$ and $\psi_{g}$. Intuitively, the first controls the smoothness of the regression function $h(g)$ and the second controls the quality of the nearest-neighbors that make up $\hat{h}(g)$ in terms of proximity to $g$. Supporting simulation evidence can be found in Appendix Section \ref{sec:simulations} below. 

The continuity parameter $\phi_{g}$ is often relatively easy to characterize using economic theory because many models of network interference give an explicit bound. In particular, for our three examples in Section \ref{sec:motivation_examples}, $\phi_{g}(x)$ quickly converges to $0$ with $x$ for any $g\in\mathcal{G}$. In the neighborhood spillovers model of Example \ref{ex:spillovers} with binary treatments and uniformly bounded expected outcomes, $\phi_{g}(x) \leq M\mathbbm{1}\{x > \frac{1}{1+r}\}$ where $M = \sup_{g \in \mathcal{G}}2h(g)$, because if $d(g,\tilde{g}) \leq \frac{1}{1+r}$ then $h(g) = h(\tilde{g})$. Similarly, in the social capital formation model of Example \ref{ex:capital} with uniformly bounded $2$-neighborhoods, $\phi_{g}(x) \leq M\mathbbm{1}\{x > \frac{1}{3}\}$ where $M$ bounds the number of agents within path distance $2$ of the root agent. In both of these examples, the function $\phi_{g}(x)$ is flat in a neighborhood of $0$ and so there is no asymptotic bias ($C \to \infty$). 

In the linear-in-means peer effects model of Example \ref{ex:social_interactions} with binary treatments and  assuming $\sup_{i \in \mathcal{I}}|T_i^*(1)| < \infty$, $\phi_{g}(x) \leq \frac{2M(\delta \rho)^{(1-x)/x}}{1 - \delta\rho}$ where $\rho$ bounds the spectral radius of $A^{*}(D)$, $|\delta\rho| < 1$ by assumption, and $M = \sup_{i \in V(\tilde{g}), \tilde{g} \in \mathcal{G}}\left|T_{i}\beta + T_{i}^{*}(1)\gamma\right|$. This is because if $d(g,\tilde{g}) \leq \frac{1}{1+r}$ then $h^{s}(g) = h^{s}(\tilde{g})$ for every $s \leq r$ and the remainder term in the policy function $\left|\sum_{s = r+1}^{\infty}[\delta^{s}A^{*}(D)^{s}(\mathbf{T}\beta + \mathbf{T}^{*}(1)\gamma)]_{i}\right| \leq M\sum_{s = r+1}^{\infty}\delta^{s}\rho^{s} \leq \frac{M(\delta \rho)^{r}}{1 - \delta\rho}$ for any $g \in \mathcal{G}$. In contrast to the first two examples, in this example the function $\phi_{g}(x)$ is not necessarily flat in a neighborhood of $0$. However, it is close to flat in the sense that its slope can be made arbitrarily uniformly close to $0$ in an open neighborhood of $0$.

In contrast to these explicit bounds on $\phi_{g}$, we are not aware of any similarly convenient way to analytically characterize the regularity parameter $\psi_{g}$, even for relatively simple network formation models. If the network is sparse or has a predictable structure (agents interact in small groups or on a regular lattice), then the rooted network variable may essentially act like a discrete random variable, and so $\psi_{g}(\ell)$ may be uniformly bounded away from $0$ for any fixed $\ell > 0$. For example, the model of \cite{jackson2012social} implies that certain social networks should in equilibrium be described by a union of completely connected subgraphs. If agent degree is also bounded, then this model has  $\psi_{g}(\ell)$ uniformly bounded away from $0$ for any fixed $\ell > 0$ and so there is no asymptotic bias ($C \to \infty$).

Irregular network formation models are more common empirically, however. For example, a large literature considers models of network formation in which connections between agents are conditionally independent across agent-pairs. Examples include the Erd\"os-Renyi model, Watts-Strogatz model, and random geometric graph model \cite[see broadly][Section 4.1]{jackson2008social}. For such models, the associated $\psi_{g}$ function can change dramatically with the model parameters. 



\section{Simulation evidence}\label{sec:simulations}
We provide simulation evidence characterizing some of the finite sample properties of our test procedure and estimator. Section \ref{sec:sim_design} describes the simulation design, Section \ref{sec:sim_test} gives the results for the first application testing policy irrelevance, and Section \ref{sec:sim_MSE} gives the results for the second application estimating policy effects. 

\subsection{Simulation design}\label{sec:sim_design}
We simulate data from $C$ communities, where $C$ is specified below. Each community contains $20$ agents. Links between agents are drawn from either an Erd\"os-Renyi, Watts-Strogatz, or random geometric graph model, as detailed below. 

Outcomes are generated for each agent $i \in [20]$ in each community $c \in [C]$ according to the model
\begin{align*}
Y_{ic} = \theta_1f(G_{ic}^1) + \theta_2f(G_{ic}^2) + U_{ic}
\end{align*}
where $\theta = (\theta_1, \theta_2) \in \mathbb{R}^2$ is specified below, 
\begin{align*}
f(g) &:= deg(g) + 2clust(g)~, \\
deg(g) &:= \frac{1}{|\mathcal{I}(g)|}\sum_{i \in \mathcal{I}(g)}\sum_{j \in \mathcal{I}(g)}D_{ij}(g)
\end{align*}
measures the average degree of the network $(V(g), E(g))$, 
\begin{align*}
clust(g) := \frac{1}{|\mathcal{I}(g)|}\sum_{i \in \mathcal{I}(g)}\frac{\sum_{j \in \mathcal{I}(g)}\sum_{k\in \mathcal{I}(g)}D_{ij}D_{ik}D_{jk}}{\sum_{j \in \mathcal{I}(g)}\sum_{k\in \mathcal{I}(g)}D_{ik}D_{jk}}
\end{align*}
measures the average clustering of the network $(\mathcal{I}(g), D(g))$, and $U_{ic} \sim U[-5, 5]$ is independent of $G_{ic}$. Our focus on the degree and clustering statistics is meant to mimic the first two examples of Section \ref{sec:motivation_examples} and the empirical example of Section \ref{sec:empirical_illustration}.


We draw the rooted networks from one of three models: the Erd\"os-Renyi model, the Watts-Strogatz model, and the random geometric grah model. See Section 4.1 of \cite{jackson2008social} for a textbook introduction. All three models generate unweighted networks. The first model is the Erd\"os-Renyi model, which is parametrized by $p \in [0,1]$. Under the Erd\"os-Renyi model, the upper diagonal entries of the adjacency matrix are independent and identically distributed Bernoulli random variables with mean $p$. The second model is the Watts-Strogatz model, which is parametrized by $(M,p) \in \mathbbm{N}\times [0,1]$. Under the Watts-Strogatz model, the network connections are generated in two stages. In the first stage, the nodes are arranged in a circle and each node is connected to their $2M$ closest neighbors. In the second stage, the connected pairs of nodes are listed sequentially, and with probability $p$ (drawn independently for each pair), the link between the pair is destroyed and another pair of agents (drawn uniformly at random from the set of unlinked nodes) is connected instead. The third model is the random geometric graph model, which is parametrized by $r\in [0,1]$. Under the random geometric graph model, nodes draw positions independent and uniformly distributed on $[0,1]^2$. Nodes within a Euclidean distance of $r$ are connected. 



\subsection{Testing policy irrelevance}\label{sec:sim_test}

In this section we study the finite-sample behavior of the approximate randomization test outlined in Algorithm \ref{algo:test}. Based on our simulation results, we recommend that practitioners adopt $q = 2\lfloor \log(C) \rfloor$ as a preliminary rule of thumb. We leave a detailed study of data-driven choices of $q$ to future work.

We first evaluate the size control properties of the test when the null hypothesis is true. The rooted networks we study are depicted in Figure \ref{fig:rooted_networks_testing}. We consider the following pairwise comparisons of these rooted networks and corresponding network formation model, with outcome model given by $\theta = (0,2)$:
\begin{itemize}
	\item $(g_1, g_2)$ under Erd\"os-Renyi model with parameter $0.1$.
	\item $(g_3, g_8)$ under Watts-Strogatz model with parameter $(2, 0.2)$.
	\item $(g_6, g_7)$ under Watts-Strogatz model with parameter $(2, 0.8)$.
	\item $(g_1, g_2)$ under Watts-Strogatz model with parameter $(2, 0.275)$.
	\item $(g_4, g_5)$ under random geometric model with parameter $0.2$.
	\item $(g_1, g_2)$ under random geometric model with parameter $0.2$.
\end{itemize}
We note that under our outcome model with $\theta = (0,2)$, these pairs of rooted networks have the same distributions of outcomes, and thus the null hypothesis holds. Table \ref{tab:rej_prob_null} reports the rejection probabilities of our test for $C \in \{20, 50, 100, 200\}$, and $q$ ranging from $4$ to $20$, with $\alpha = 0.05$. 
  The results show that the test rejects the null hypothesis with probability approximately equal to $\alpha$ when $q$ is small and/or $C$ is large. Size distortion may occur when $q$ is large and $C$ is small. These simulation results also demonstrate that, although we can control size for most DGPs with sufficiently small values of $q$ and large values of $C$, for some of our DGPs size control is particularly challenging unless $q$ is chosen to be very small; these are DGPs for which the rooted networks under consideration appear extremely infrequently. 

\begin{figure}
    \centering
     \begin{tikzpicture}[->,>= stealth,shorten >=1pt,auto,node distance=1cm,
        thick,main node/.style={circle,fill=blue!20,draw,minimum size=.3cm,inner sep=0pt]}]

      \node[main node] (2) {};
    \node[draw, fill=blue!20, shape=diamond, aspect=0.7, minimum height=0.3cm, inner sep=0pt] (6) [ above left of =2] {$1$};
    \node[main node] (5) [above right of=2] {};
    \node[main node] (1) [ below left of =2] {};
    \node[main node] (3) [ below right of =2] {};

    \path[-]
    (6) edge node {} (2)
        	edge node {} (1)
    (2) edge node {} (5)
    	edge node {} (1)
	edge node {} (3);
\end{tikzpicture}
\hspace{5mm}
     \begin{tikzpicture}[->,>= stealth,shorten >=1pt,auto,node distance=1cm,
        thick,main node/.style={circle,fill=blue!20,draw,minimum size=.3cm,inner sep=0pt]}]

      \node[main node] (2) {};
    \node[main node] (6) [ above left of =2] {};
    \node[draw, fill=blue!20, shape=diamond, aspect=0.7, minimum height=0.3cm, inner sep=0pt] (5) [above right of=2] {$2$};
    \node[main node] (1) [ below left of =2] {};
    \node[main node] (3) [ below right of =2] {};

    \path[-]
    (6) edge node {} (2)
        	edge node {} (1)
    (2) edge node {} (5)
    	edge node {} (1)
	edge node {} (3);
\end{tikzpicture}
\hspace{5mm}
        \begin{tikzpicture}[->,>= stealth,shorten >=1pt,auto,node distance=1cm,
        thick,main node/.style={circle,fill=blue!20,draw,minimum size=.3cm,inner sep=0pt]}]

      \node[draw, fill=blue!20, shape=diamond, aspect=0.7, minimum height=0.3cm, inner sep=0pt] (2) {$3$};
    \node[main node] (6) [ left of =2] {};
    \node[main node] (5) [above left of=6] {};
    \node[main node] (1) [below left of =6] {};

    \path[-]
    (6) edge node {} (2)   
    (1) edge node {} (6)
    (5) edge node {} (6);
\end{tikzpicture}
\hspace{5mm}
     \begin{tikzpicture}[->,>= stealth,shorten >=1pt,auto,node distance=1cm,
        thick,main node/.style={circle,fill=blue!20,draw,minimum size=.3cm,inner sep=0pt]}]

      \node[main node] (2) {};
    \node[draw, fill=blue!20, shape=diamond, aspect=0.7, minimum height=0.3cm, inner sep=0pt] (6) [ above left of =2] {$4$};
    \node[main node] (5) [right of=2] {};
    \node[main node] (1) [ below left of =2] {};

    \path[-]
    (6) edge node {} (2)
        	edge node {} (1)
    (2) edge node {} (5)
    	edge node {} (1);
\end{tikzpicture}
\hspace{5mm}
\begin{tikzpicture}[->,>=stealth,shorten >=1pt,auto,node distance=1cm,
        thick,main node/.style={circle,fill=blue!20,draw,minimum size=.3cm,inner sep=0pt}]

      \node[main node] (2) {};
      \node[main node] (6) [above left of=2] {};
      \node[main node] (1) [below left of=2] {};
      \node[draw, fill=blue!20, shape=diamond, aspect=0.7, minimum height=0.3cm, inner sep=0pt] (3) [right of=2] {$5$};

      \path[-]
        (6) edge node {} (2)
        (6) edge node {} (1)
        (2) edge node {} (5)
        (2) edge node {} (1)
        (2) edge node {} (3);
    \end{tikzpicture}
    \hspace{5mm}
    \begin{tikzpicture}[->,>=stealth,shorten >=1pt,auto,node distance=1cm,
        thick,main node/.style={circle,fill=blue!20,draw,minimum size=.3cm,inner sep=0pt}]

      \node[draw, fill=blue!20, shape=diamond, aspect=0.7, minimum height=0.3cm, inner sep=0pt] (2) {$6$};
      \node[main node] (6) [above left of=2] {};
      \node[main node] (1) [below left of=2] {};
      \node[main node] (3) [right of=2] {};
      \node[main node] (4) [right of=3] {};

      \path[-]
        (6) edge node {} (2)
        (2) edge node {} (5)
        (2) edge node {} (1)
        (2) edge node {} (3)
        (3) edge node {} (4);
    \end{tikzpicture}
    \hspace{5mm}
    \begin{tikzpicture}[->,>=stealth,shorten >=1pt,auto,node distance=1cm,
        thick,main node/.style={circle,fill=blue!20,draw,minimum size=.3cm,inner sep=0pt}]

    \node[main node] (2) {};
    \node[main node] (6) [above left of=2] {};
    \node[main node] (1) [below left of=2] {};
    \node[draw, fill=blue!20, shape=diamond, aspect=0.7, minimum height=0.3cm, inner sep=0pt] (3) [right of=2] {$7$};
    \node[main node] (4) [right of=3] {};
    
    \path[-]
    (6) edge node {} (2)
    (2) edge node {} (5)
    (2) edge node {} (1)
    (2) edge node {} (3)
    (3) edge node {} (4);
    \end{tikzpicture}
    \hspace{4.8mm}
    \begin{tikzpicture}[->,>=stealth,shorten >=1pt,auto,node distance=1cm,
        thick,main node/.style={circle,fill=blue!20,draw,minimum size=.3cm,inner sep=0pt}]
    \raisebox{6mm}{
      \node[main node] (2) {};
      \node[draw, fill=blue!20, shape=diamond, aspect=0.7, minimum height=0.3cm, inner sep=0pt] (6) [left of=2] {$8$};
      \node[main node] (5) [left of=6] {};
      \node[main node] (1) [left of=5] {};

      \path[-]
        (6) edge node {} (2)
        (5) edge node {} (6)
        (1) edge node {} (5);
    }
    \end{tikzpicture}
\caption{Eight rooted networks truncated at radius $2,$ labeled $g_1$ to $g_8$.}
\label{fig:rooted_networks_testing}
\end{figure}


\begin{table}[htbp]
    \centering
    \caption{Rejection probabilities: $\alpha = 0.05$ ($2,000$ Monte Carlo iterations)}
    \renewcommand{\arraystretch}{1.0}
    \setlength{\tabcolsep}{4pt} 
    \begin{tabular}{@{}c*{17}{c}@{}}
        \toprule
        & \multicolumn{17}{c}{$H_0 : g_1 =_d g_2$ (Erd\"os-Renyi(0.1))} \\
        & \multicolumn{17}{c}{$q$}  \\
        \cmidrule(lr){2-18} 
        $C$ & 4 & 5 & 6 & 7 & 8 & 9 & 10 & 11 & 12 & 13 & 14 & 15 & 16 & 17 & 18 & 19 & 20  \\
        \midrule
        20  &  6.7 & 7.6  & 8.5  & 9.6  & 9.3  & 8.7  & 8.5  &  &  &  &  &  &  &  &  &  &  \\
        50  & 4.2 & 4.7  & 5.2  & 5.2  & 6.0  & 7.4  & 8.8  & 9.8  & 11.3  & 12.6  & 14.1  & 15.2  & 16.5  & 17.1  & 17.3  & 17.6  & 17.5 \\
        100 & 4.2 & 4.0  & 4.0  & 4.6  & 4.9  & 5.5  & 6.0  & 5.1  & 5.3  & 5.9  & 6.5  & 6.4  & 6.0  & 7.3  & 8.0  & 9.2  & 10.1 \\
        200 & 5.6 & 5.4  & 5.0  & 5.1  & 4.6  & 4.9  & 4.6  & 5.0  & 4.3  & 4.7  & 4.9  & 5.3  & 5.4  & 4.9  & 5.1  & 5.0  & 5.0 \\
        \midrule
        & \multicolumn{17}{c}{$H_0 : g_3 =_d g_8$ (Watts-Strogatz(2, 0.2))} \\
        \midrule
        20  &  4.8  & 5.2  & 4.8  & 4.4  & 4.2  & 4.2  & 4.9  &  &  &  &  &  &  &  &  &  &  \\
        50  &  5.1  & 5.6  & 5.1  & 5.1  & 5.1  & 5.8  & 4.9  & 5.0  & 5.0  & 5.1  & 5.2  & 5.1  & 5.7  & 5.7  & 5.7  & 5.7  & 5.4 \\
        100 & 5.2  & 5.5  & 5.3  & 4.5  & 4.9  & 5.1  & 5.5  & 5.5  & 5.7  & 5.1  & 5.7  & 4.9  & 5.0  & 4.7  & 5.1  & 5.2  & 5.1 \\
        200 & 5.4  & 4.2  & 5.1  & 4.3  & 4.5  & 4.3  & 4.1  & 4.0  & 4.3  & 4.4  & 4.2  & 4.1  & 4.2  & 4.5  & 4.4  & 4.2  & 4.0 \\
        \midrule
         & \multicolumn{17}{c}{$H_0 : g_6 =_d g_7$  (Watts-Strogatz(2, 0.8))} \\
        \midrule
        20  & 5.1  & 4.6  & 5.0  & 5.4  & 5.4  & 5.9  & 5.9  &  &  &  &  &  &  &  &  &  &  \\
        50  & 5.1  & 4.8  & 4.8  & 5.0  & 4.7  & 5.3  & 4.5  & 5.2  & 5.2  & 5.4  & 5.5  & 5.8  & 5.7  & 5.0  & 5.0  & 5.0  & 5.1 \\
        100 & 5.5  & 6.2  & 5.6  & 5.7  & 5.9  & 6.2  & 5.2  & 5.1  & 5.3  & 5.3  & 5.3  & 5.1  & 4.8  & 4.7  & 4.5  & 5.0  & 5.2 \\
        200 & 5.2  & 5.5  & 4.6  & 4.8  & 5.1  & 5.1  & 5.1  & 4.8  & 5.1  & 5.3  & 5.3  & 5.1  & 5.5  & 5.8  & 5.9  & 5.2  & 4.5 \\
        \midrule
        & \multicolumn{17}{c}{$H_0 : g_1 =_d g_2$  (Watts-Strogatz(2, 0.275))} \\
        \midrule
        20  & 12.6  & 14.8  & 13.9  & 13.8  & 13.5  & 12.4  & 12.3  &  &  &  &  &  &  &  &  &  &  \\
        50  & 12.6  & 16.5  & 19.9  & 23.9  & 27.1  & 29.3  & 31.1  & 31.6  & 32.3  & 30.2  & 30.0  & 28.6  & 27.9  & 27.9  & 26.9  & 27.0  & 27.1 \\
        100 & 8.7  & 11.8  & 14.5  & 20.2  & 23.3  & 28.5  & 32.5  & 37.8  & 42.8  & 47.0  & 50.6  & 54.7  & 56.5  & 58.7  & 60.2  & 60.0  & 59.8 \\
        200 & 5.9  & 6.9  & 8.8  & 10.7  & 13.9  & 18.4  & 21.6  & 24.7  & 28.4  & 31.7  & 36.2  & 41.3  & 45.4  & 49.5  & 52.7  & 56.0  & 59.9 \\
        \midrule
        & \multicolumn{17}{c}{$H_0 : g_4 =_d g_5$  (Random Geometric(0.2))} \\
        \midrule
        20  & 5.8  & 6.3  & 7.2  & 10.1  & 13.7  & 18.9  & 25.1  &  &  &  &  &  &  &  &  &  &  \\
        50  & 4.1  & 4.5  & 5.2  & 5.5  & 5.1  & 5.3  & 5.1  & 6.0  & 5.9  & 5.7  & 6.3  & 7.0  & 8.5  & 11.9  & 14.6  & 18.2  & 23.1 \\
        100 & 4.2  & 4.5  & 4.3  & 4.9  & 4.7  & 4.7  & 4.9  & 5.1  & 4.8  & 4.9  & 4.5  & 4.5  & 4.4  & 4.7  & 4.2  & 4.3  & 4.9 \\
        200 & 5.0  & 4.8  & 4.6  & 5.0  & 5.2  & 4.7  & 5.2  & 5.0  & 4.2  & 4.4  & 4.8  & 4.5  & 4.4  & 5.2  & 5.5  & 5.0  & 4.5  \\
        \midrule
         & \multicolumn{17}{c}{$H_0 : g_1 =_d g_2$ (Random Geometric(0.2))} \\
        \midrule
        20  & 24.6  & 33.3  & 41.6  & 51.9  & 58.2  & 63.9  & 69.6  &  &  &  &  &  &  &  &  &  &  \\
        50  & 15.0  & 21.8  & 26.9  & 36.7  & 45.3  & 52.2  & 58.9  & 63.9  & 69.0  & 73.6  & 79.2  & 82.8  & 85.6  & 87.9  & 90.8  & 92.3  & 89.0 \\
        100 &7.7  & 10.3  & 15.4  & 22.1  & 28.0  & 34.0  & 40.9  & 48.0  & 54.9  & 59.9  & 65.6  & 71.7  & 75.8  & 79.5  & 82.6  & 85.7  & 89.0 \\
        200 & 5.0  & 5.7  & 6.4  & 8.2  & 9.8  & 12.9  & 17.6  & 21.9  & 26.4  & 33.5  & 40.1  & 45.2  & 51.9  & 57.7  & 63.6  & 67.7  & 72.4 \\
        \bottomrule
    \end{tabular}
    \label{tab:rej_prob_null}
\end{table}

Next, to evaluate the power properties of the test procedure, we consider two rooted networks $g_{3}$ and $g_{4}$ that are associated with two different conditional distributions of outcomes under the model in Section \ref{sec:sim_design}. To vary the extent to which the distributions of the outcomes under these two rooted networks differ, we modify our data generating process for the outcomes by varying the amount of unobserved heterogeneity $U_{ic}$. In particular, we consider the following distributions of $U_{ic}$: $U[-1,1]$, $U[-2,2]$, $U[-3,3]$, $U[-4,4]$. Table \ref{tab:rej_prob_alt} reports the results of the rejection probabilities of our test when testing $g_3 =_d g_4$  at level $\alpha = 0.05$, for $q$ ranging from $4$ to $20$. The results show that the test correctly rejects the null hypothesis with probability greater than $\alpha$. The probability of rejection generally increases with $q$ and $C$.

\begin{table}[htbp]
    \centering
    \caption{Rejection probabilities: $\alpha = 0.05$ ($2,000$ Monte Carlo iterations)}
    \renewcommand{\arraystretch}{1.0}
    \setlength{\tabcolsep}{1pt} 
    \begin{tabular}{@{}c*{17}{c}@{}}
        \toprule
        & \multicolumn{17}{c}{$H_0 : g_3 =_d g_4$ (Erd\"os-Renyi(0.1), $U_{ic} \sim U\in[-5,5]$)} \\
        & \multicolumn{17}{c}{$q$}  \\
        \cmidrule(lr){2-18} 
        $C$ & 4 & 5 & 6 & 7 & 8 & 9 & 10 & 11 & 12 & 13 &14 & 15 & 16 & 17 & 18 & 19 & 20  \\
        \midrule
        20 &  16.7  & 20.2  & 22.1  & 22.4  & 21.3  & 20.4  & 19.4  &  &  &  &  &  &  &  &  &  &  \\
        50 &  18.0  & 22.7  & 28.5  & 31.5  & 34.6  & 39.7  & 43.1  & 45.7  & 48.7  & 50.1  & 54.2  & 55.7  & 57.2  & 57.1  & 56.6  & 55.2  & 53.8 \\
        100 &  18.7  & 25.2  & 28.4  & 34.2  & 37.5  & 40.5  & 43.2  & 47.9  & 50.7  & 54.9  & 58.0  & 60.9  & 64.3  & 66.8  & 68.6  & 71.3  & 73.9 \\
        200 &  21.9  & 29.2  & 32.3  & 37.7  & 41.9  & 45.9  & 49.0  & 53.2  & 56.1  & 58.9  & 61.8  & 63.4  & 66.7  & 68.8  & 70.8  & 74.0  & 75.9 \\
         \midrule
        & \multicolumn{17}{c}{$H_0 : g_3 =_d g_4$ (Erd\"os-Renyi(0.1), $U_{ic} \sim U\in[-4,4]$)} \\
        \midrule
        20  & 22.8  & 28.4  & 30.9  & 31.3  & 30.5  & 28.6  & 27.6  &  &  &  &  &  &  &  &  &  &  \\
        50  & 25.4  & 32.7  & 40.5  & 45.5  & 50.4  & 55.3  & 58.9  & 62.3  & 66.8  & 70.6  & 73.3  & 74.4  & 74.1  & 75.2  & 73.6  & 73.5  & 72.7 \\
        100 & 28.7  & 35.2  & 41.3  & 47.7  & 52.3  & 55.7  & 61.1  & 66.0  & 70.2  & 73.8  & 77.7  & 79.4  & 82.0  & 83.6  & 86.1  & 88.1  & 90.2 \\
        200 & 30.4  & 41.5  & 46.9  & 54.7  & 58.1  & 63.0  & 68.9  & 71.7  & 74.6  & 77.6  & 80.0  & 83.1  & 84.2  & 86.0  & 88.1  & 89.5  & 90.7 \\
        \midrule
         & \multicolumn{17}{c}{$H_0 : g_3 =_d g_4$ (Erd\"os-Renyi(0.1), $U_{ic} \sim U\in[-3,3]$)} \\
        \midrule
        20  & 34.2  & 44.3  & 47.3  & 46.9  & 46.1  & 43.1  & 41.4  &  &  &  &  &  &  &  &  &  &  \\
        50  & 40.9  & 51.9  & 60.0  & 67.5  & 73.2  & 77.7  & 82.1  & 86.4  & 89.1  & 90.8  & 91.8  & 92.5  & 92.5  & 92.0  & 91.8  & 91.2  & 90.1 \\
        100 & 46.2  & 56.6  & 64.5  & 70.7  & 75.9  & 81.1  & 85.5  & 88.9  & 91.1  & 93.1  & 94.6  & 95.7  & 96.7  & 97.6  & 98.1  & 98.7  & 98.7 \\
        200 & 50.6  & 65.2  & 72.7  & 78.5  & 83.5  & 86.9  & 91.2  & 92.0  & 93.6  & 95.4  & 96.3  & 97.1  & 97.8  & 98.4  & 98.7  & 99.1  & 99.3 \\
        \midrule
        & \multicolumn{17}{c}{$H_0 : g_3 =_d g_4$ (Erd\"os-Renyi(0.1), $U_{ic} \sim U\in[-2,2]$)} \\
        \midrule
        20  & 61.6  & 71.8  & 75.2  & 72.3  & 70.1  & 67.4  & 65.4  &  &  &  &  &  &  &  &  &  &  \\
        50  & 70.5  & 83.5  & 89.2  & 93.7  & 96.2  & 97.8  & 98.6  & 99.3  & 99.5  & 99.6  & 99.7  & 99.6  & 99.5  & 99.4  & 99.6  & 99.3  & 99.4 \\
        100 & 81.6  & 90.1  & 93.9  & 96.5  & 98.3  & 98.9  & 99.3  & 99.8  & 99.9  & 99.8  & 99.9  & 99.9  & 100.0  & 100.0  & 100.0  & 100.0  & 100.0 \\
        200 & 90.6  & 96.5  & 98.3  & 98.9  & 99.5  & 99.6  & 99.7  & 99.9  & 99.9  & 99.9  & 100.0  & 99.9  & 100.0  & 100.0  & 100.0  & 100.0  & 100.0 \\
        \midrule
        & \multicolumn{17}{c}{$H_0 : g_3 =_d g_4$ (Erd\"os-Renyi(0.1), $U_{ic} \sim U\in[-1,1]$)} \\
        \midrule
        20  & 96.7  & 94.1  & 93.9  & 90.5  & 90.4  & 89.3  & 87.4  &  &  &  &  &  &  &  &  &  &  \\
        50  & 99.4  & 99.9  & 100.0  & 100.0  & 100.0  & 100.0  & 100.0  & 100.0  & 100.0  & 100.0  & 100.0  & 100.0  & 100.0  & 100.0  & 100.0  & 100.0  & 100.0 \\
        100 & 99.8  & 100.0  & 100.0  & 100.0  & 100.0  & 100.0  & 100.0  & 100.0  & 100.0  & 100.0  & 100.0  & 100.0  & 100.0  & 100.0  & 100.0  & 100.0  & 100.0 \\
        200 & 99.9  & 100.0  & 100.0  & 100.0  & 100.0  & 100.0  & 100.0  & 100.0  & 100.0  & 100.0  & 100.0  & 100.0  & 100.0  & 100.0  & 100.0  & 100.0  & 100.0 \\
        \bottomrule
    \end{tabular}
    \label{tab:rej_prob_alt}
\end{table}

\subsection{Estimating policy effects}\label{sec:sim_MSE}
We study the mean-squared error of the $k$-nearest-neighbor estimator for the policy function $h$ given in Section \ref{sec:estimator}. for four rooted networks and $\theta \in \{ (1,0),(1,1/2) \}$. Under $\theta = (1,0)$ the distribution of outcomes depends on the features of the network within radius $1$ of the root. Under $\theta = (1,1/2)$ the distribution of outcomes  also depends on the features of the network within radius $2$ of the root. 

The choice of rooted networks we consider is represented by $g_3$, $g_4$, $g_9$, and $g_0$ in Figure \ref{fig:rooted_networks}. Networks $g_{3}$ and $g_{9}$ depict two wheels with the rooted agent on the periphery. These networks have moderate average degree and no average clustering: $(3/2,0)$ and $(5/3,0)$ respectively. Network $g_{4}$ depicts a closed triangle connected to a single agent. This network has moderate average degree and high average clustering $(2,7/12)$.  Finally, network $g_{0}$ depicts a closed triangle connected to a wheel with the rooted agent both on the periphery of the wheel and part of the triangle. This network has moderate average degree and average clustering $(2,1/3)$.
The results of the simulation are shown in Table 5. As suggested by Theorem 4.2, mean-squared error is generally decreasing with $C$ for a fixed choice of $k$. In addition, mean-square error is generally smaller for the $\theta = (1, 0)$ experiment than it is for the $\theta = (1, 0.5)$ experiment for a fixed choice of $C$ and $k$. The effect is more pronounced for the networks $g_{4}$ and $g_{0}$, for which we typically observe fewer good matches in the data compared to $g_3$ and $g_9$ (we quantify this observation by estimating $\psi_{g}$ for each of the four networks in Section \ref{sec:psi_measure}). This is also consistent with Theorem 4.2.


Fixing $C$ and comparing across $k$, we expect a bias-variance trade-off. For networks $g_3$ and $g_9$, there is no meaningful bias in the estimated policy function because $f(\tilde{g}^{1})$ and $f(\tilde{g}^{2})$ are similar and $\psi_{\tilde{g}}(1) \approx 1$ for $\tilde{g} = g_{3}, g_9$. As a result, it is optimal to use the nearest-neighbor from every community in $[C]$ (i.e. choose $k = C$). In contrast for  $g_4$ and $g_{0}$ the rooted networks of the nearest neighbors in each community may be very different from the relevant policies, and so setting $k = C$ can lead to an inflated mean-squared error. 

We conclude that unless the researcher has additional information about the structure of network interference or the density of the policies of interest, $k$ should not be large relative to the sample size. This is also consistent with our findings for the test of policy irrelevance above. 

\begin{figure}
    \centering
        \begin{tikzpicture}[->,>= stealth,shorten >=1pt,auto,node distance=1cm,
        thick,main node/.style={circle,fill=blue!20,draw,minimum size=.3cm,inner sep=0pt]}]

      \node[draw, fill=blue!20, shape=diamond, aspect=0.7, minimum height=0.3cm, inner sep=0pt] (2) {$3$};
    \node[main node] (6) [ left of =2] {};
    \node[main node] (5) [above left of=6] {};
    \node[main node] (1) [below left of =6] {};

    \path[-]
    (6) edge node {} (2)   
    (1) edge node {} (6)
    (5) edge node {} (6);
\end{tikzpicture}
\hspace{5mm}
     \begin{tikzpicture}[->,>= stealth,shorten >=1pt,auto,node distance=1cm,
        thick,main node/.style={circle,fill=blue!20,draw,minimum size=.3cm,inner sep=0pt]}]

      \node[main node] (2) {};
    \node[draw, fill=blue!20, shape=diamond, aspect=0.7, minimum height=0.4cm, inner sep=0pt] (6) [ above left of =2] {$4$};
    \node[main node] (5) [right of=2] {};
    \node[main node] (1) [ below left of =2] {};

    \path[-]
    (6) edge node {} (2)
        	edge node {} (1)
    (2) edge node {} (5)
    	edge node {} (1);
\end{tikzpicture}
\hspace{5mm}
        \begin{tikzpicture}[->,>= stealth,shorten >=1pt,auto,node distance=1cm,
        thick,main node/.style={circle,fill=blue!20,draw,minimum size=.3cm,inner sep=0pt]}]

      \node[draw, fill=blue!20, shape=diamond, aspect=0.7, minimum height=0.4cm, inner sep=0pt] (2) {$9$};
    \node[main node] (6) [ left of =2] {};
    \node[main node] (5) [above left of=6] {};
    \node[main node] (4) [above of=6] {};
    \node[main node] (3) [below of=6] {};
    \node[main node] (1) [below left of =6] {};

    \path[-]
    (6) edge node {} (2)   
    (1) edge node {} (6)
    (5) edge node {} (6)
    (4) edge node {} (6)
    (3) edge node{} (6);
\end{tikzpicture}
\hspace{5mm}
     \begin{tikzpicture}[->,>= stealth,shorten >=1pt,auto,node distance=1cm,
        thick,main node/.style={circle,fill=blue!20,draw,minimum size=.3cm,inner sep=0pt]}]

      \node[draw, fill=blue!20, shape=diamond, aspect=0.7, minimum height=0.4cm, inner sep=0pt] (2) {$0$};
    \node[main node] (6) [ above left of =2] {};
    \node[main node] (5) [right of=2] {};
    \node[main node] (1) [ below left of =2] {};
    \node[main node] (3) [ below right of =5] {};
    \node[main node] (4) [ above right of =5] {};
    \node[main node] (8) [ right of =5] {};

    \path[-]
    (6) edge node {} (2)
        	edge node {} (1)
    (2) edge node {} (5) 
    (1) edge node {} (2)
    (4) edge node {} (5)
    (8) edge node {} (5)
    (3) edge node {} (5);
\end{tikzpicture}

\caption{Four rooted networks truncated at radius $2,$ labeled $g_3$, $g_4$, $g_9$, and $g_{0}$.}\label{fig:rooted_networks}
\end{figure}

\begin{table}[htbp]
\scriptsize
  \centering
  \caption{Estimated MSEs \\ ($1,000$ Monte Carlo iterations)}
    \begin{tabular}{rcccc|cccc|cccccc}
    \toprule
          & \multicolumn{1}{r}{} &       & $C = 20$ & \multicolumn{1}{r}{} &       & $C = 50$ &       & \multicolumn{1}{r}{} &       &       & $C = 100$ &       &       &  \\
          & \multicolumn{1}{r}{} &       & $k$     & \multicolumn{1}{c}{} &       & $k$     &       & \multicolumn{1}{r}{} &       &       & $k$     &       &       &  \\
\cmidrule{3-15}    \multicolumn{1}{c}{$\alpha$} & \multicolumn{1}{c}{$g_{\iota}$} & 5     & 10    & \multicolumn{1}{c}{20}    & 5     & 10    & 20    & \multicolumn{1}{c}{50}    & 5     & 10    & 20    & 50    & 75    & 100 \\
    \midrule
          & $g_3$    & 1.71  & 0.84  & 0.41  & 1.69  & 0.86  & 0.43  & 0.18  & 1.63  & 0.81  & 0.41  & 0.17  & 0.11  & 0.08 \\
    \multicolumn{1}{c}{$(1, 0)$} &$g_4$    & 1.77  & 1.53  & 3.27  & 1.65  & 0.89  & 0.55  & 3.02  & 1.64  & 0.84  & 0.42  & 0.64  & 1.92  & 2.97 \\
              &  $g_{9}$     & 1.7   & 0.84  & 0.42  & 1.75  & 0.84  & 0.44  & 0.17  & 1.66  & 0.85  & 0.42  & 0.17  & 0.12  & 0.09 \\
          & $g_{0}$   & 1.77  & 0.98  & 1.41  & 1.6   & 0.81  & 0.42  & 1.17  & 1.74  & 0.84  & 0.44  & 0.2   & 0.59  & 1.09 \\
    \midrule
          & $g_3$    & 1.71  & 0.84  & 0.41  & 1.69  & 0.86  & 0.43  & 0.18  & 1.63  & 0.81  & 0.41  & 0.17  & 0.11  & 0.09 \\
    \multicolumn{1}{c}{$(1, 0.5)$} &  $g_4$    & 1.85  & 2.15  & 5.37  & 1.66  & 0.9   & 0.72  & 5.09  & 1.65  & 0.86  & 0.44  & 1.1   & 3.32  & 5.04 \\
             &  $g_{9}$    & 1.7   & 0.84  & 0.42  & 1.75  & 0.84  & 0.44  & 0.17  & 1.66  & 0.86  & 0.42  & 0.17  & 0.12  & 0.09 \\
            & $g_{0}$  & 1.79  & 1.05  & 2.07  & 1.6   & 0.81  & 0.42  & 1.81  & 1.75  & 0.84  & 0.44  & 0.21  & 0.87  & 1.73 \\

    \bottomrule
    \end{tabular}%
  \label{tab:MSE_results}%
\end{table}%

\subsection{Measuring network regularity}\label{sec:psi_measure}
In Section \ref{sec:est_asf} we identified the function $\psi_{g}(\ell)$ as a key determinant of the estimation bias in Theorem \ref{thm:mse_bound}. This function measures the probability that the nearest-neighbor of $g$ from a randomly drawn network is within distance $\ell$ of $g$ as measured by $d$. 

Figure \ref{fig:psi_nulls} displays estimates of the $\psi_{g}$ function for the eight rooted networks considered in the simulation design of Section \ref{sec:sim_test}. The figures were constructed by generating $3000$ random graphs under different network formation models with $20$ nodes each and recording the distances of the nearest neighbor to $g$ in each graph. 

The results indicate that $g_3$ and $g_8$ appear frequently under Watts-Strogatz(2,0.2) and $g_6$ and $g_7$ appear frequently under Watts-Strogatz(2,0.8), while the other networks are relatively infrequently under their corresponding network formation models, which may explain why testing $g_3 =_d g_8$ under Watts-Strogatz(2,0.2) and $g_6 =_d g_7$ under Watts-Strogatz(2,0.8) displays better size control compared to the other cases (See Table \ref{tab:rej_prob_null}).

Figure \ref{fig:psi} displays estimates of the $\psi_{g}$ function for the four rooted networks considered in the simulation design of Section \ref{sec:sim_MSE}. The figures were constructed by generating $3000$ Erd\"os-Renyi(0.1) random graphs with $20$ nodes each and recording the distances of the nearest neighbor to $g$ in each graph. Intuitively, networks $g_{4}$ and  $g_{0}$ are rare because triadic closure is uncommon under the random graph model. The network $g_{9}$ is also relatively rare because the coincidence of five agents linked to a common agent is uncommon for such a sparse random network. 

We remark that strategic interaction between agents should in principle further discipline the regularity of the network, particularly if only a small number of configurations are consistent with equilibrium linking behavior. Characterizing $\psi_{g}$ for such strategic network formation models is an important area for future work.

 
%

\begin{figure}
    \centering
    \begin{subfigure}[b]{0.49\textwidth}
    	\centering
 	\includegraphics[width=1\textwidth]{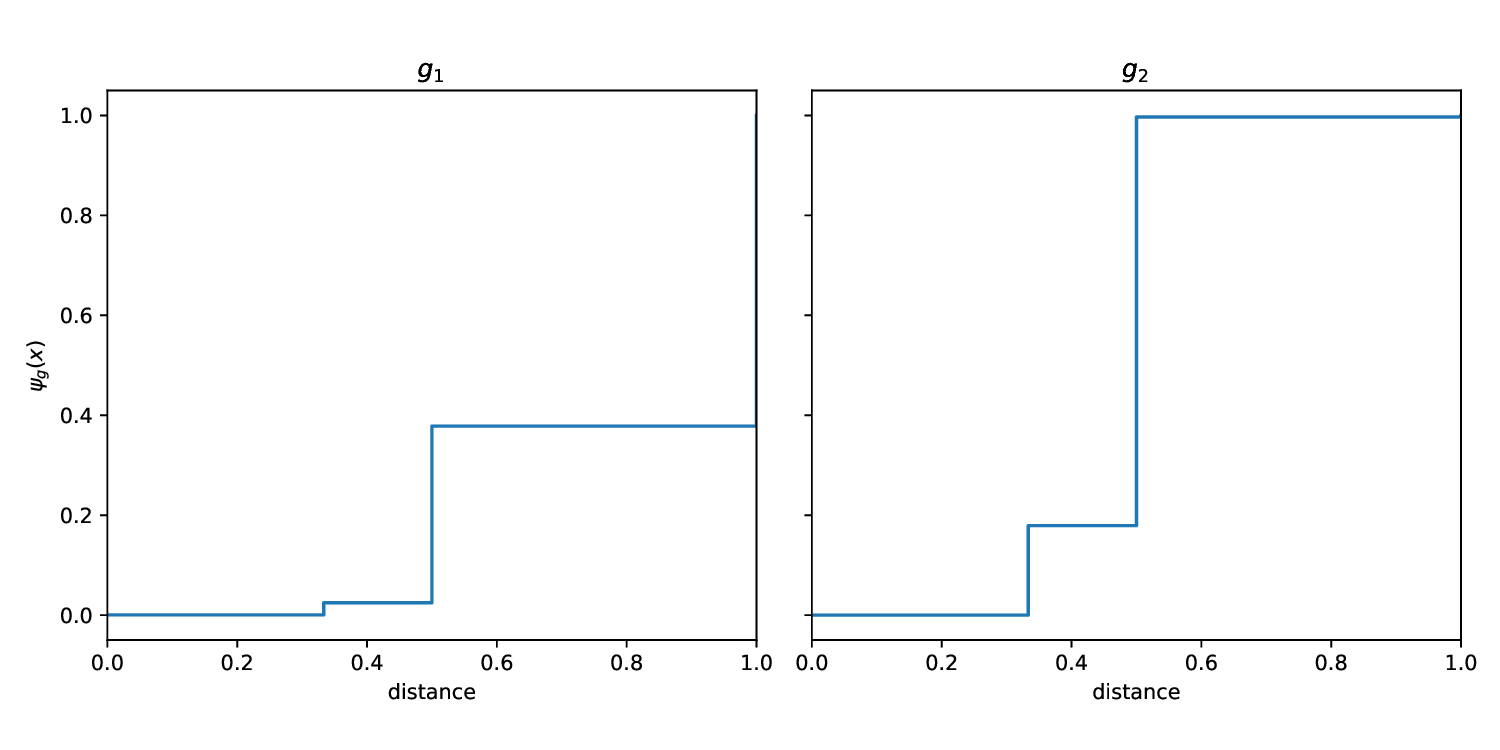}
    	\caption{$g_1$ and $g_2$ under Erd\"os-Renyi(0.1)}
    \end{subfigure}%
    \hfill
    \begin{subfigure}[b]{0.49\textwidth}
    	\centering
 	\includegraphics[width=1\textwidth]{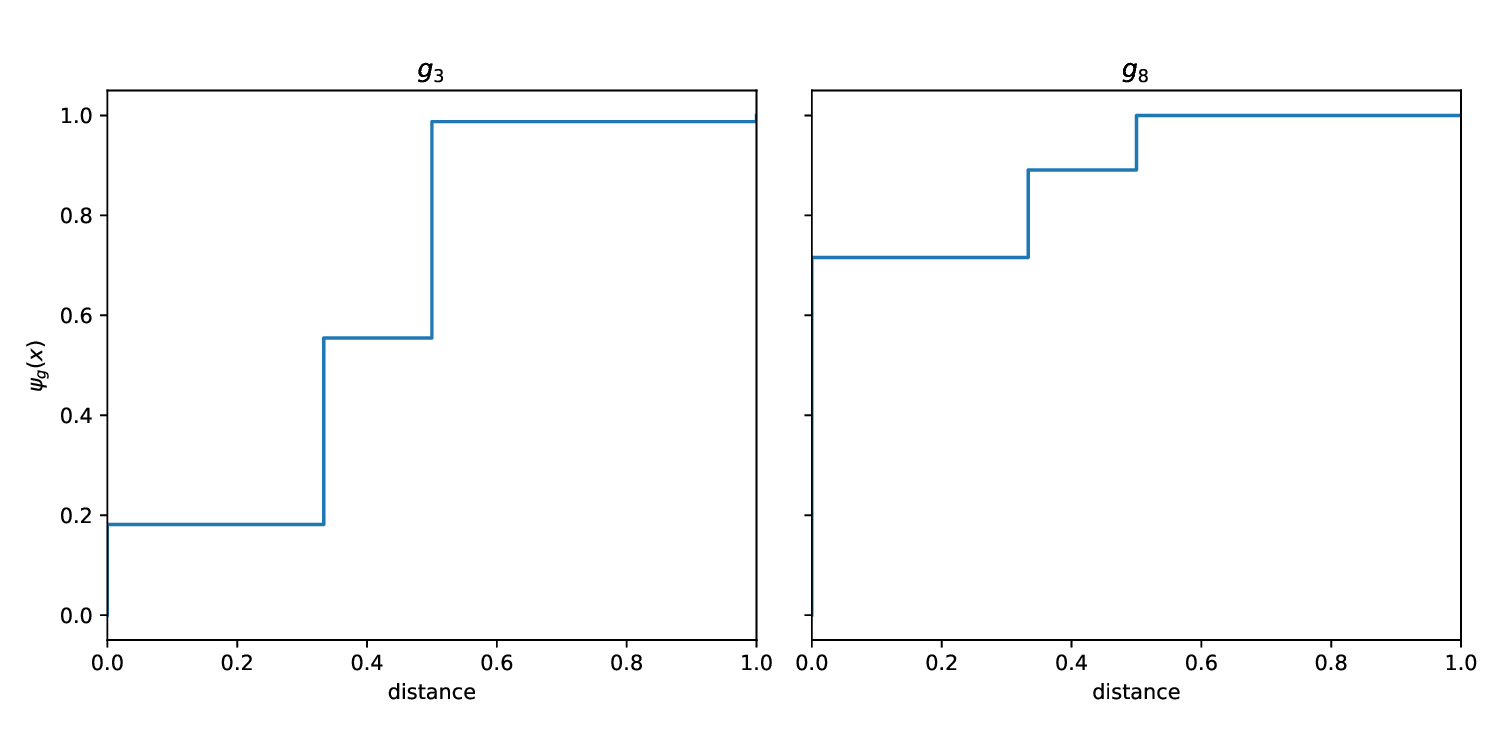}
    	\caption{$g_3$ and $g_8$ under Watts-Strogatz(2,0.2)}
    \end{subfigure}
    
    \begin{subfigure}[b]{0.49\textwidth}
    	\centering
 	\includegraphics[width=1\textwidth]{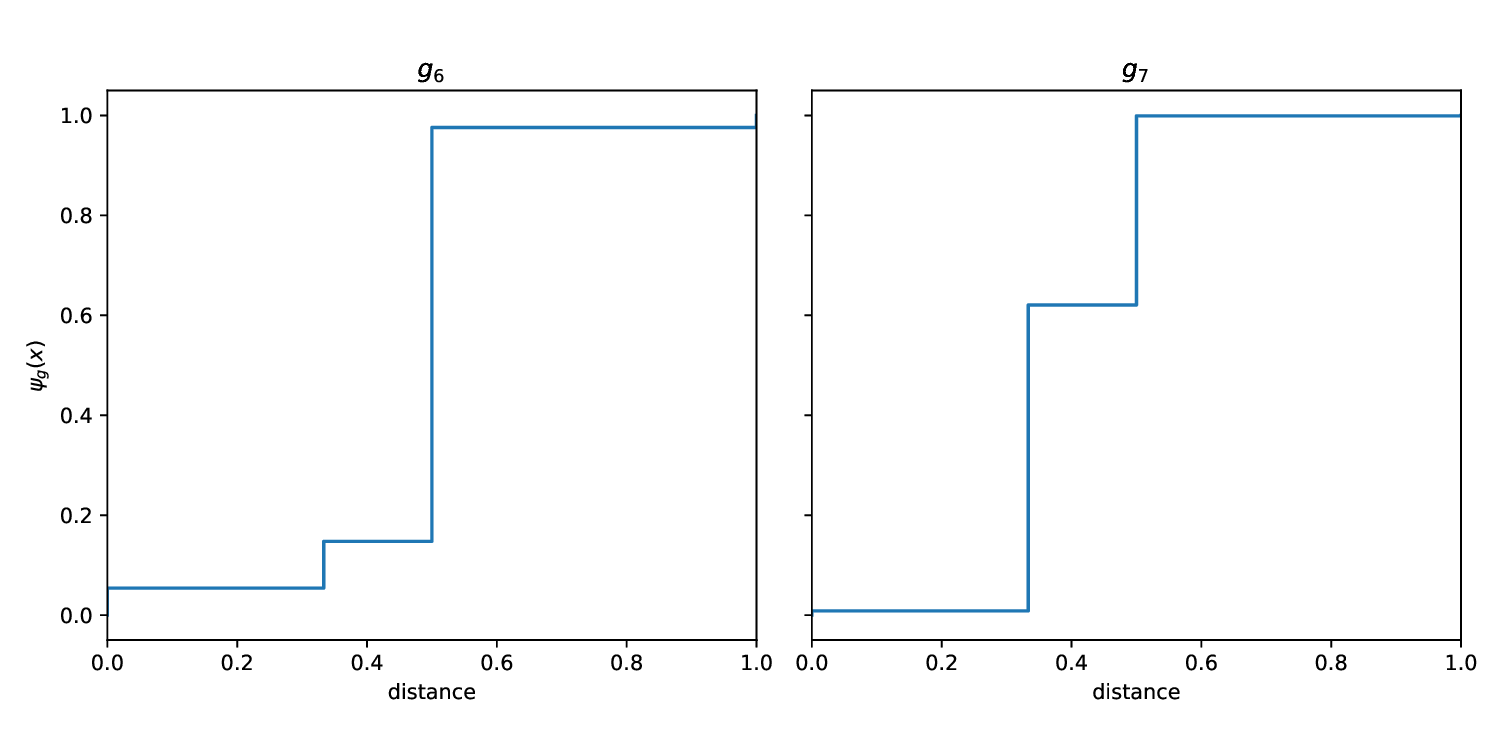}
    	\caption{$g_6$ and $g_7$ under Watts-Strogatz(2,0.8)}
    \end{subfigure}%
    \hfill
    \begin{subfigure}[b]{0.49\textwidth}
    	\centering
 	\includegraphics[width=1\textwidth]{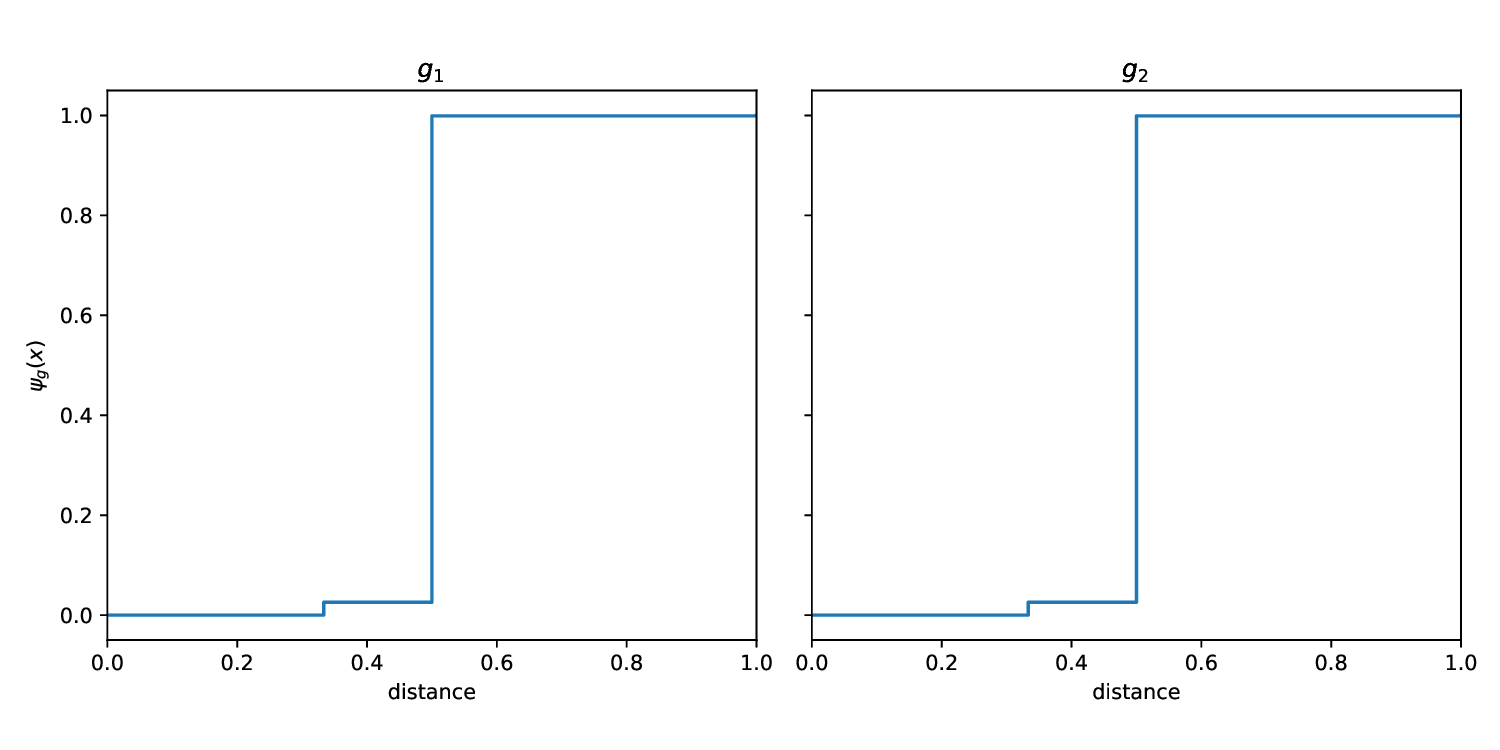}
    	\caption{$g_1$ and $g_2$ under Watts-Strogatz(2,0.275)}
    \end{subfigure}
    
     \begin{subfigure}[b]{0.49\textwidth}
    	\centering
 	\includegraphics[width=1\textwidth]{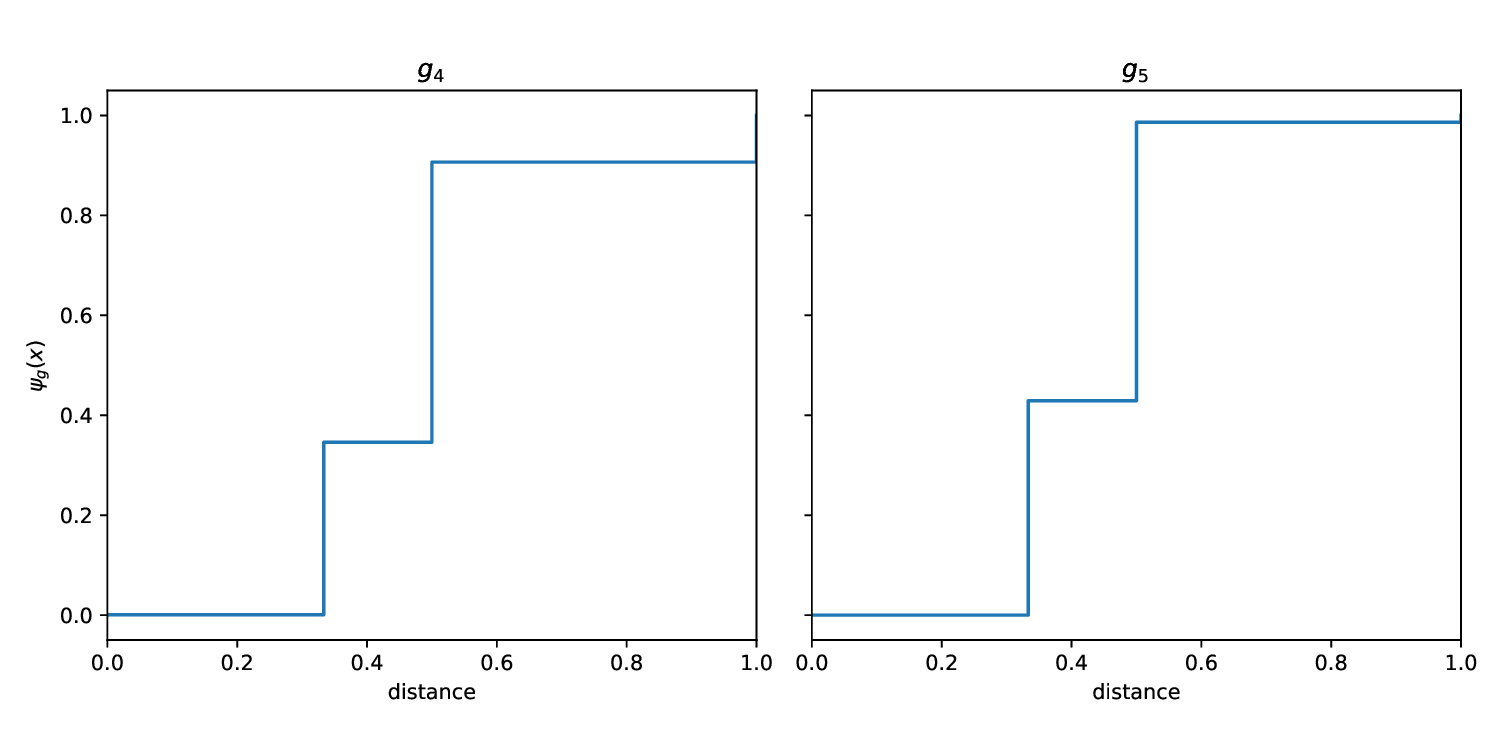}
    	\caption{$g_4$ and $g_5$ under random geometric (0.2)}
    \end{subfigure}%
    \hfill
    \begin{subfigure}[b]{0.49\textwidth}
    	\centering
 	\includegraphics[width=1\textwidth]{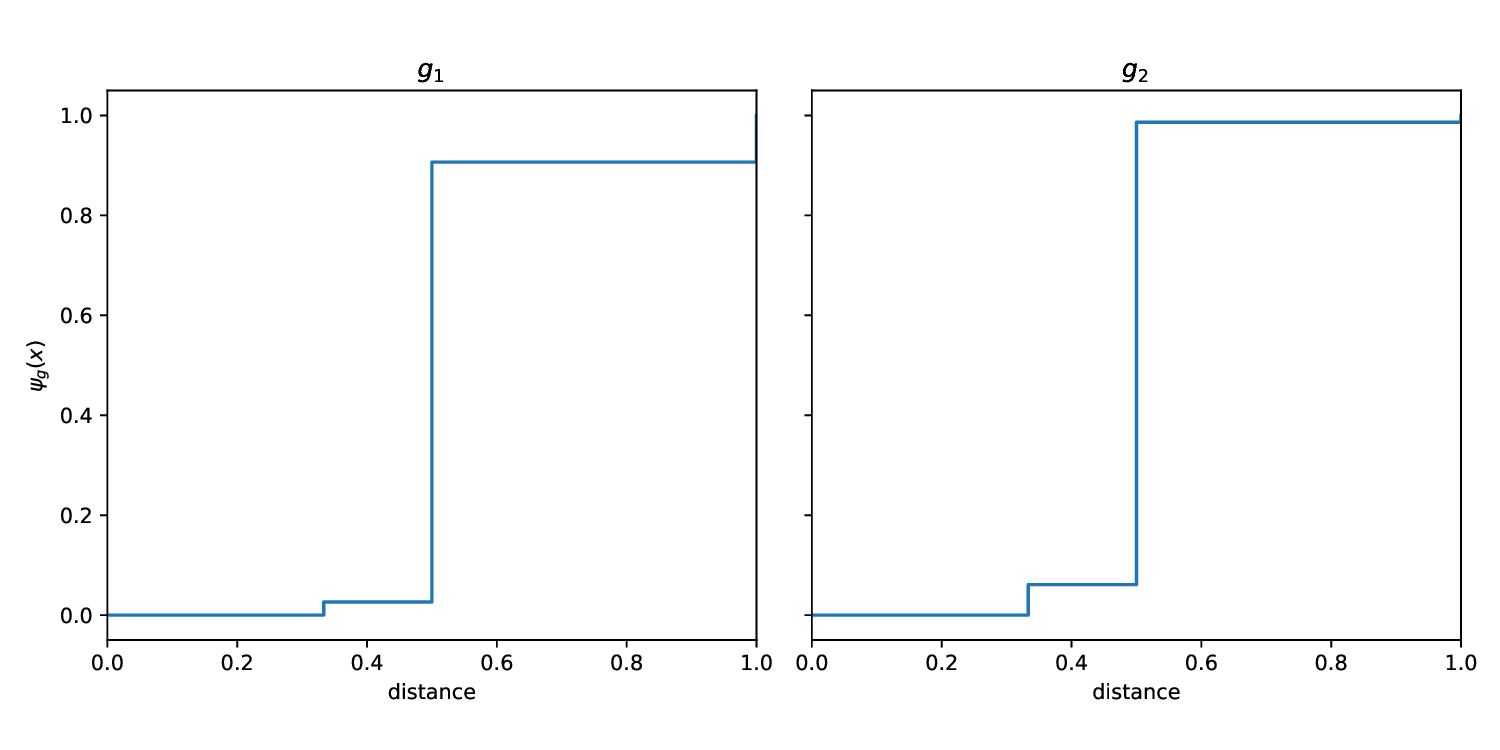}
    	\caption{$g_1$ and $g_2$ under random geometric (0.2)}
    \end{subfigure}
    
    \caption{Estimated $\psi_g(\cdot)$ for $g_1$ to $g_8$ under different network formation models}
    \label{fig:psi_nulls}
\end{figure}

\begin{figure}
	\centering
	\includegraphics[width=1\textwidth]{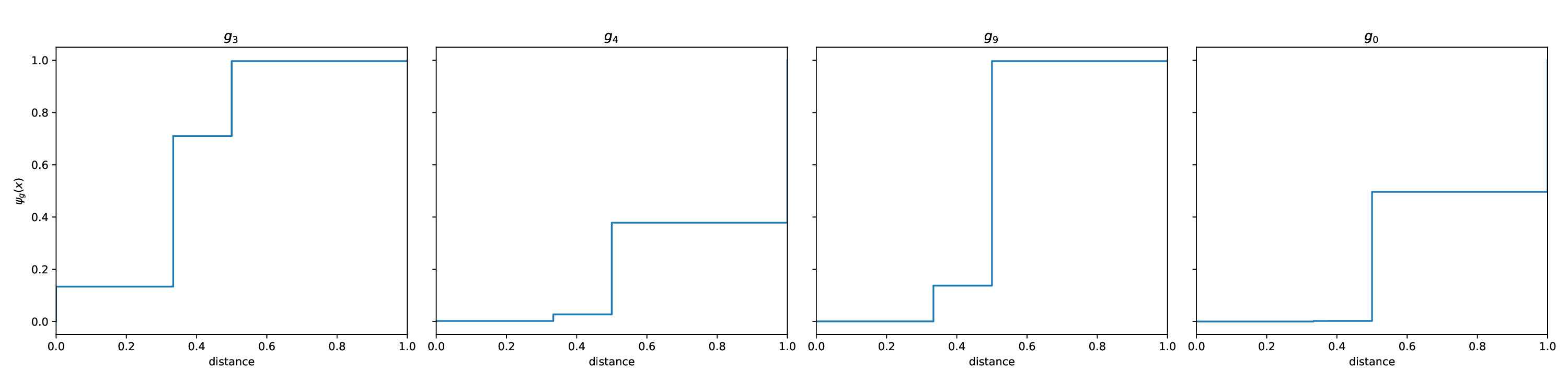}
	\caption{Estimated $\psi_{g}(\cdot)$ for $g_3$, $g_4$, $g_9$, $g_0$ under Erd\"os-Renyi(0.1).}\label{fig:psi}
\end{figure}

\section{Additional details for Section \ref{sec:empirical_illustration}}\label{sec:robust_app}
In this section we provide some additional supporting evidence for the empirical illustration. Figures \ref{fig:closest_knife_knife_vs_fork_jackson_testing}--\ref{fig:closest_spoon_fork_vs_spoon_jackson_testing} report the set of nearest neighbors which we use for our tests of policy irrelevance.  The figures demonstrate that, for fork and knife, there are villages which contain rooted networks which match at a radius of 2. The same is not true for spoon; this is further confirmed by the density plots provided in Figure \ref{fig:edf_distance_app}, which demonstrate that there are \emph{no} rooted networks in the data which match with spoon at a radius of two. In order to break these arbitrary ties when reporting the results of the approximate randomization test for our empirical application, we chose the rooted network with the fewest number of nodes. To check the robustness of the $p$-values to this arbitrary tie-breaking rule, in Tables \ref{tab:robust_p_knife_fork} and \ref{tab:robust_p_fork_spoon} we report the $p$-value obtained from ten randomly drawn rooted networks among those with the same distance from the target rooted network. The results show that our conclusions are relatively robust to this selection rule.

\begin{figure}[h]
     \centering
     \includegraphics[width=\textwidth]{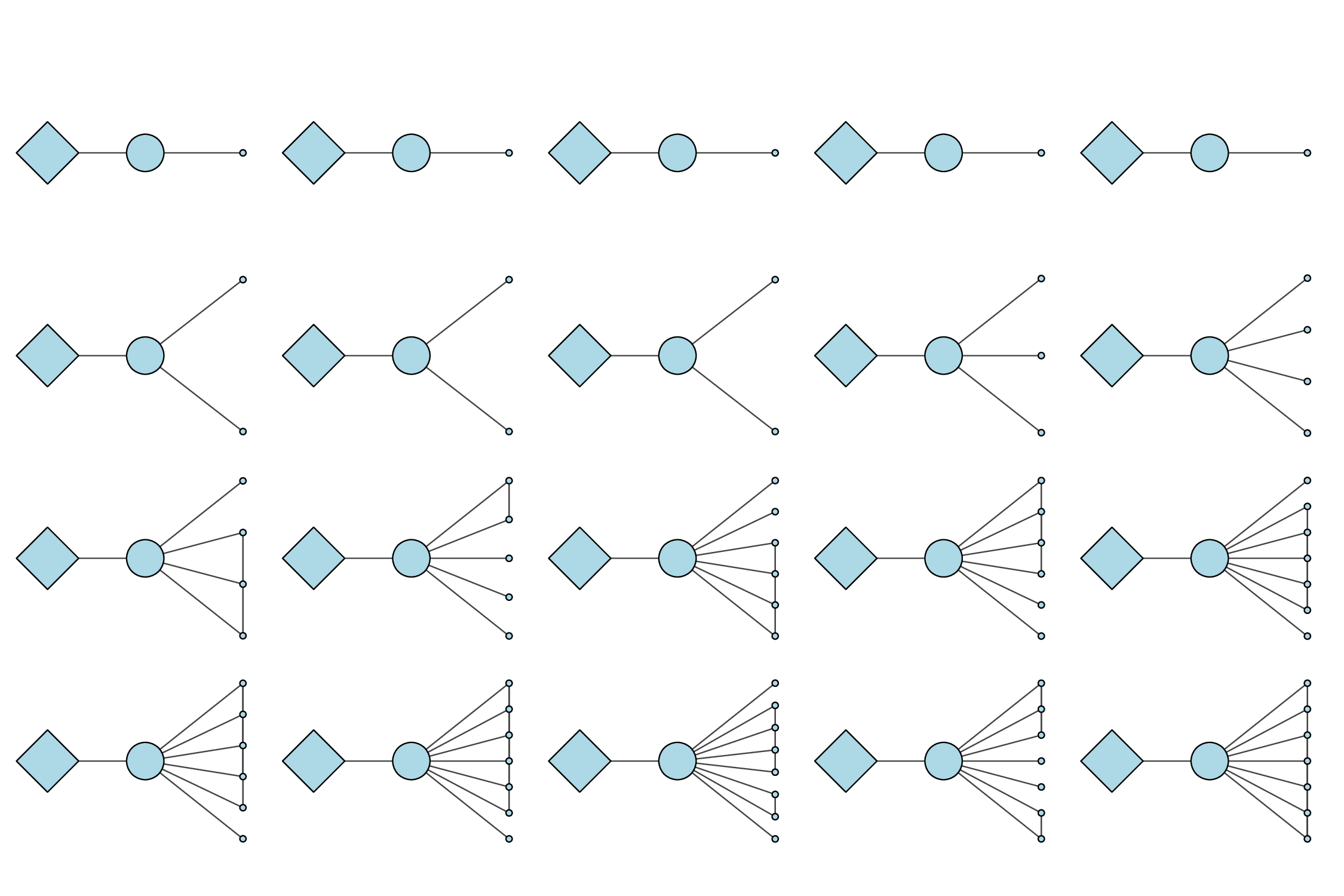}
    \caption{Neighbors matched with knife (radius=2) for testing $Y_\alpha =_d Y_\beta$}
    \label{fig:closest_knife_knife_vs_fork_jackson_testing}
\end{figure}

\begin{figure}[h]
     \centering
     \includegraphics[width=\textwidth]{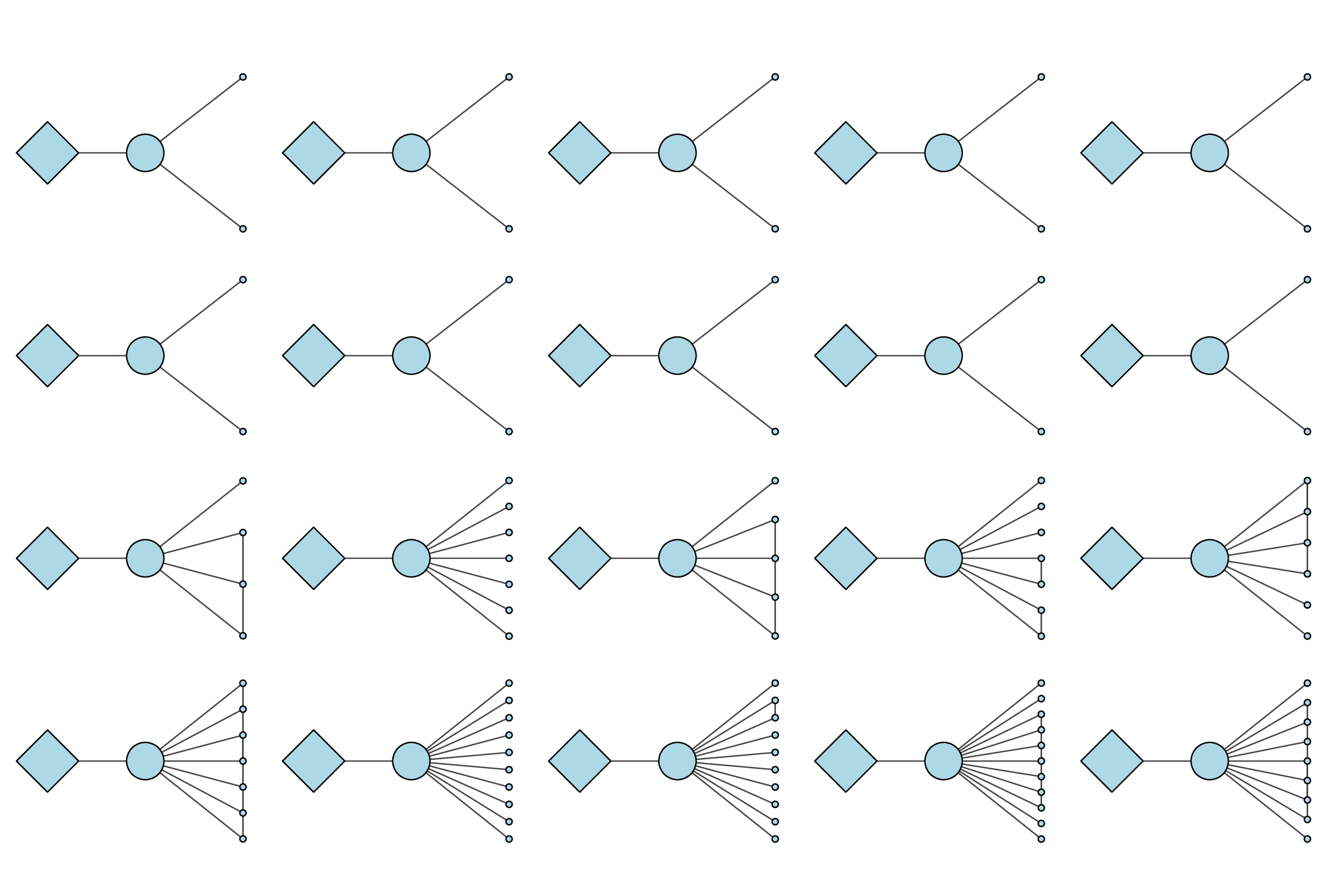}
    \caption{Neighbors matched with fork (radius=2) for testing $Y_\alpha =_d Y_\beta$}
    \label{fig:closest_fork_knife_vs_fork_jackson_testing}
\end{figure}

\begin{figure}[h]
     \centering
     \includegraphics[width=\textwidth]{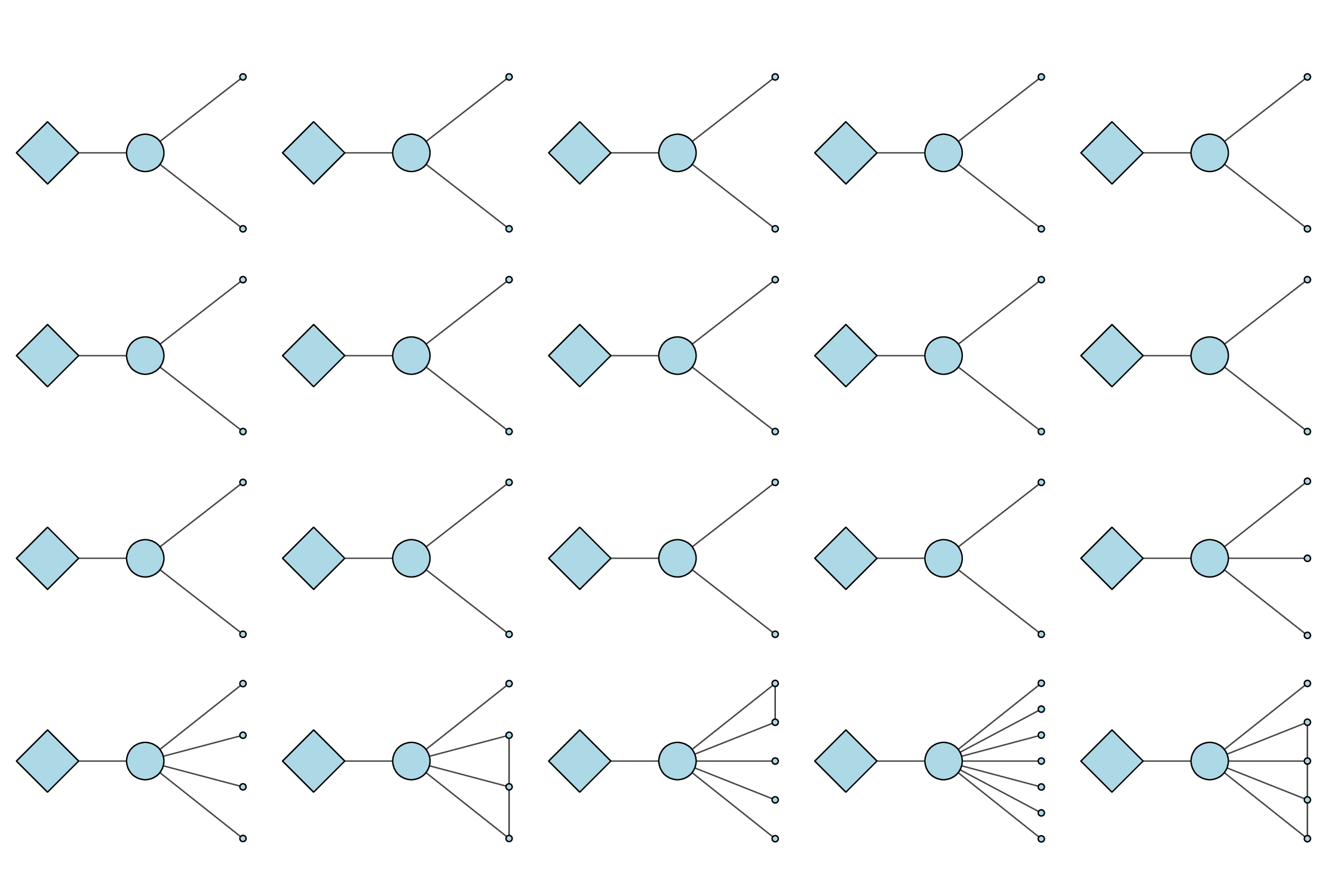}
    \caption{Neighbors matched with fork (radius=2) for testing $Y_\beta=_d Y_\gamma$}
    \label{fig:closest_fork_fork_vs_spoon_jackson_testing}
\end{figure}

\begin{figure}[h]
     \centering
     \includegraphics[width=\textwidth]{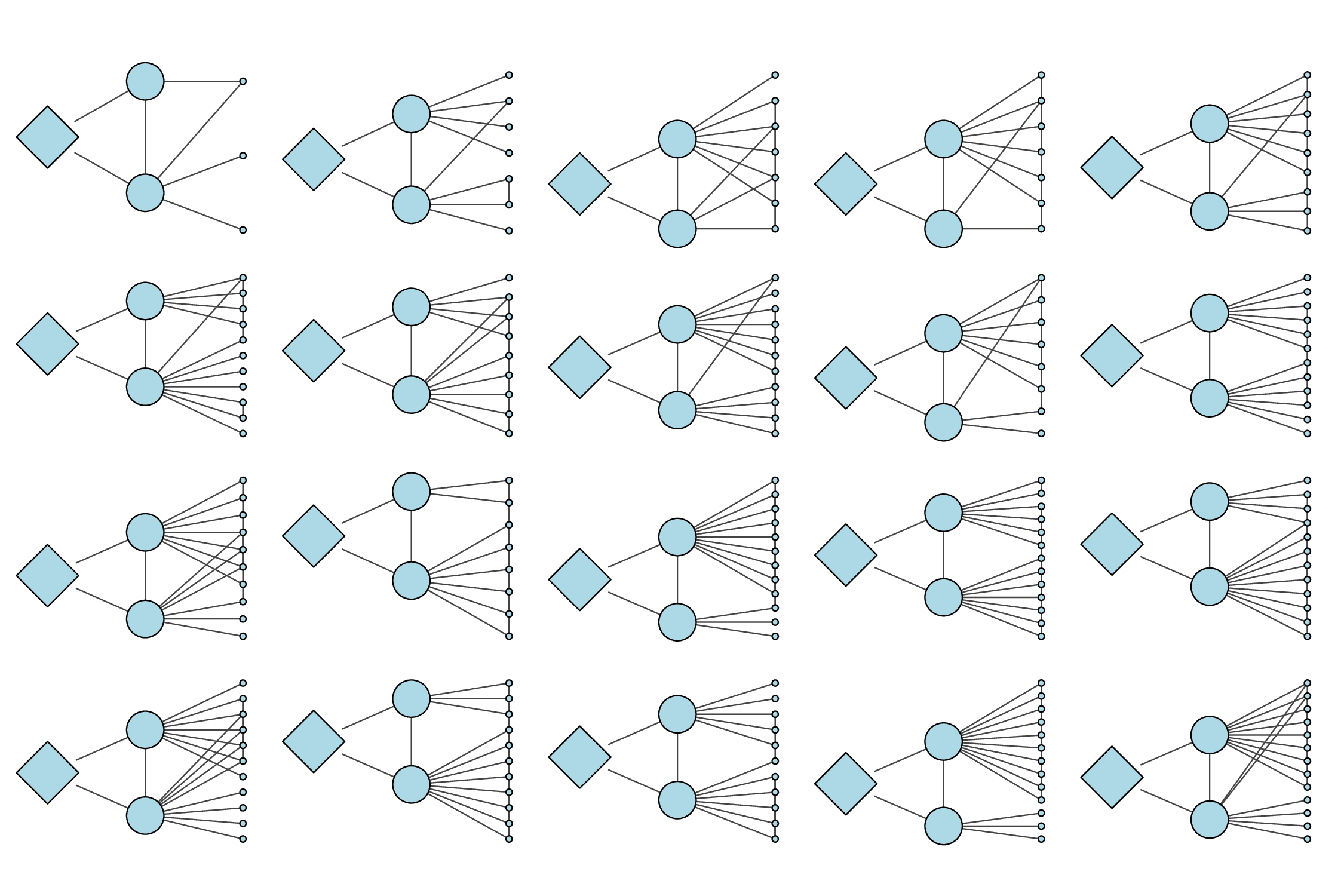}
    \caption{Neighbors matched with spoon (radius=2) for testing $Y_\beta=_d Y_\gamma$}
    \label{fig:closest_spoon_fork_vs_spoon_jackson_testing}
\end{figure}

\begin{figure}[htbp]
     \centering
     \includegraphics[width=\textwidth]{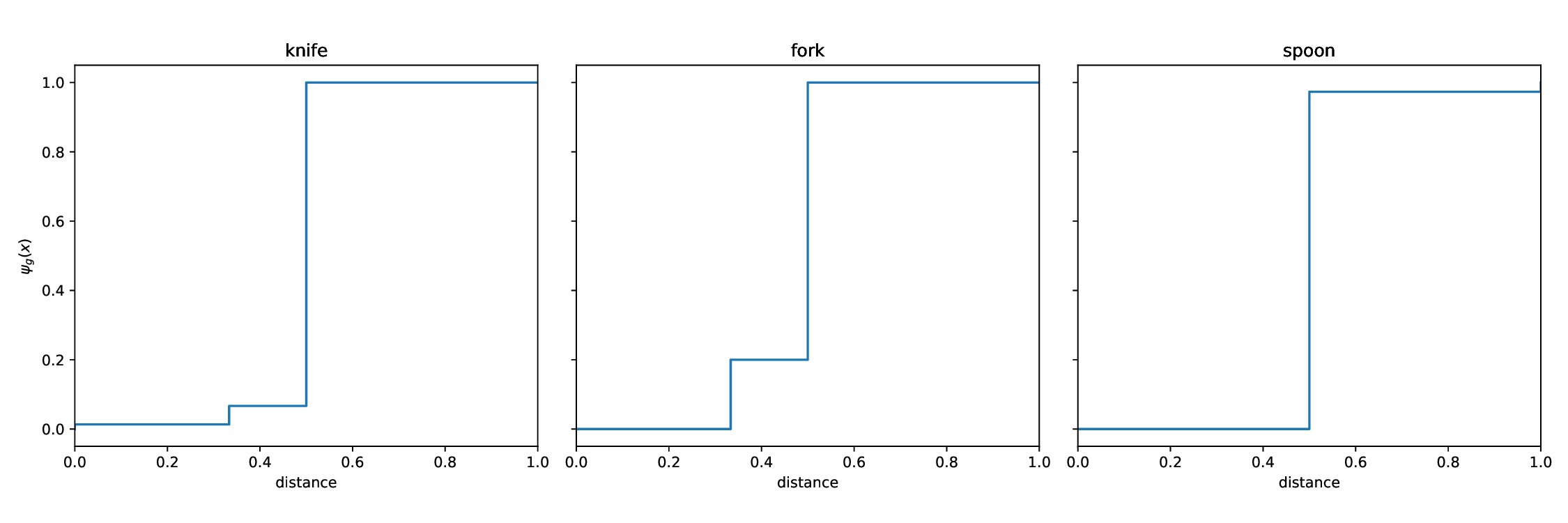}
    \caption{Estimated $\psi_g(\cdot)$ for knife, fork, and spoon}
    \label{fig:edf_distance_app}
\end{figure}

\begin{table}[htbp]
  \centering
  \caption{$p$-values for testing $Y_\alpha =_d Y_\beta$ among 10 different random picks of rooted networks with the same distance to the target rooted network}
    \begin{tabular}{ccccccccccc}
    	\toprule
    	$q=5$ & 1.00 & 0.29 & 0.89 & 1.00 & 1.00 & 1.00 & 1.00 & 1.00 & 1.00 & 1.00\\
	$q=10$ & 0.80 & 1.00 & 0.72 & 1.00 & 1.00 & 0.74 & 0.59 & 0.69 & 0.47 & 0.72 \\
	$q=20$ & 0.95 & 0.89 & 0.76 & 0.91 & 0.81 & 0.50 & 0.22 & 0.57 & 0.71 & 0.33 \\
	\bottomrule
    \end{tabular}%
  \label{tab:robust_p_knife_fork}%
\end{table}%

\begin{table}[htbp]
  \centering
  \caption{$p$-values for testing $Y_\beta =_d Y_\gamma$ among 10 different random picks of rooted networks with the same distance to the target rooted network}
    \begin{tabular}{ccccccccccc}
    	\toprule
    	$q=5$ & 0.09 & 0.36 & 1.00 & 0.28 & 0.09 & 0.22 & 0.27 & 0.03 & 0.05 & 1.00\\
	$q=10$ & 0.00 & 0.32 & 0.57 & 0.06 & 0.05 & 0.06 & 0.20 & 0.00 & 0.04 & 0.61 \\
	$q=20$ & 0.00 & 0.01 & 0.14 & 0.01 & 0.00 & 0.00 & 0.22 & 0.01 & 0.00 & 0.08 \\
	\bottomrule
    \end{tabular}%
  \label{tab:robust_p_fork_spoon}%
\end{table}%

\begin{table}[htbp]
    \centering
    \caption{Examples of networks with different distances to knife}
    \begin{tabular}{cccc}
    \toprule
    $d=0$ & $d=1/3$ & $d=1/2$ & $d=1$\\
    \midrule
    \begin{tikzpicture}[->,>= stealth,shorten >=1pt,auto,node distance=1.3cm,
            thick,main node/.style={circle,fill=blue!20,draw,minimum size=.3cm,inner sep=0pt]}]
        \raisebox{-2mm}{
        \node[main node] (2) {};
        \node[main node] (6) [ left of =2] {};
        \node[draw, fill=blue!20, shape=diamond, aspect=0.7, minimum height=0.4cm, inner sep=2.6pt] (5) [left of=6] {};
    
        \path[-]
        (6) edge node {} (2)   
        (5) edge node {} (6);
        }
    \end{tikzpicture}
    & 
    \begin{tikzpicture}[->,>= stealth,shorten >=1pt,auto,node distance=1.3cm,
            thick,main node/.style={circle,fill=blue!20,draw,minimum size=.3cm,inner sep=0pt]}]
        \raisebox{-2mm}{
        \node[main node] (2) {};
        \node[main node] (6) [ left of =2] {};
        \node[draw, fill=blue!20, shape=diamond, aspect=0.7, minimum height=0.4cm, inner sep=2.6pt] (5) [left of=6] {};
        \node[main node] (7) [right of=2] {};
    
        \path[-]
        (6) edge node {} (2)   
        (5) edge node {} (6)
        (2) edge node {} (7);
        }
    \end{tikzpicture}
    & 
    \begin{tikzpicture}[->,>= stealth,shorten >=1pt,auto,node distance=1.3cm,
            thick,main node/.style={circle,fill=blue!20,draw,minimum size=.3cm,inner sep=0pt]}]
        \raisebox{-2mm}{
        \node[main node] (6) {};
        \node[draw, fill=blue!20, shape=diamond, aspect=0.7, minimum height=0.4cm, inner sep=2.6pt] (5) [left of=6] {};
    
        \path[-] 
        (5) edge node {} (6);
        }
    \end{tikzpicture}
    & 
    \begin{tikzpicture}[->,>= stealth,shorten >=1pt,auto,node distance=1.3cm,
            thick,main node/.style={circle,fill=blue!20,draw,minimum size=.3cm,inner sep=0pt]}]
        \raisebox{-2mm}{
        \node[draw, fill=blue!20, shape=diamond, aspect=0.7, minimum height=0.4cm, inner sep=2.6pt] (6) {};
    
        \path[-];
        }
    \end{tikzpicture} \\
    & 
    \begin{tikzpicture}[->,>= stealth,shorten >=1pt,auto,node distance=1.3cm,
            thick,main node/.style={circle,fill=blue!20,draw,minimum size=.3cm,inner sep=0pt]}]
        \raisebox{-2mm}{
        \node[main node] (2) {};
        \node[main node] (6) [ left of =2] {};
        \node[draw, fill=blue!20, shape=diamond, aspect=0.7, minimum height=0.4cm, inner sep=2.6pt] (5) [left of=6] {};
        \node[main node] (7) [above right of=2] {};
        \node[main node] (8) [below right of=2] {};
    
        \path[-]
        (6) edge node {} (2)   
        (5) edge node {} (6)
        (2) edge node {} (7)
        (2) edge node {} (8);
        }
    \end{tikzpicture}
    & 
    \begin{tikzpicture}[->,>= stealth,shorten >=1pt,auto,node distance=1.3cm,
            thick,main node/.style={circle,fill=blue!20,draw,minimum size=.3cm,inner sep=0pt]}]
        \raisebox{-2mm}{
        \node[main node] (6) {};
        \node[draw, fill=blue!20, shape=diamond, aspect=0.7, minimum height=0.4cm, inner sep=2.6pt] (5) [left of=6] {};
        \node[main node] (7) [above right of=6]{};
        \node[main node] (8) [below right of=6]{};
    
        \path[-] 
        (5) edge node {} (6)
        (6) edge node {} (7)
        (6) edge node {} (8);
        }
    \end{tikzpicture}
     &  
     \begin{tikzpicture}[->,>= stealth,shorten >=1pt,auto,node distance=1.3cm,
            thick,main node/.style={circle,fill=blue!20,draw,minimum size=.3cm,inner sep=0pt]}]
        \raisebox{-2mm}{
        \node[draw, fill=blue!20, shape=diamond, aspect=0.7, minimum height=0.4cm, inner sep=2.6pt] (6) {};
        \node[main node] (7) [above right of=6] {};
        \node[main node] (8) [below right of=6] {};
    
        \path[-]
        (6) edge node {} (7)
        (6) edge node {} (8);
        }
    \end{tikzpicture} \\
     & 
     \begin{tikzpicture}[->,>= stealth,shorten >=1pt,auto,node distance=1.3cm,
            thick,main node/.style={circle,fill=blue!20,draw,minimum size=.3cm,inner sep=0pt]}]
        \raisebox{-2mm}{
        \node[main node] (2) {};
        \node[main node] (6) [ left of =2] {};
        \node[draw, fill=blue!20, shape=diamond, aspect=0.7, minimum height=0.4cm, inner sep=2.6pt] (5) [left of=6] {};
        \node[main node] (7) [above right of=2] {};
        \node[main node] (8) [below right of=2] {};
    
        \path[-]
        (6) edge node {} (2)   
        (5) edge node {} (6)
        (2) edge node {} (7)
        (2) edge node {} (8)
        (7) edge node {} (8);
        }
    \end{tikzpicture}
     & \begin{tikzpicture}[->,>= stealth,shorten >=1pt,auto,node distance=1.3cm,
            thick,main node/.style={circle,fill=blue!20,draw,minimum size=.3cm,inner sep=0pt]}]
        \raisebox{-2mm}{
        \node[main node] (6) {};
        \node[draw, fill=blue!20, shape=diamond, aspect=0.7, minimum height=0.4cm, inner sep=2.6pt] (5) [left of=6] {};
        \node[main node] (7) [above right of=6]{};
        \node[main node] (8) [below right of=6]{};
    
        \path[-] 
        (5) edge node {} (6)
        (6) edge node {} (7)
        (6) edge node {} (8)
        (7) edge node {} (8);
        }
    \end{tikzpicture} & 
    \begin{tikzpicture}[->,>= stealth,shorten >=1pt,auto,node distance=1.3cm,
            thick,main node/.style={circle,fill=blue!20,draw,minimum size=.3cm,inner sep=0pt]}]
        \raisebox{-2mm}{
        \node[draw, fill=blue!20, shape=diamond, aspect=0.7, minimum height=0.4cm, inner sep=2.6pt] (6) {};
        \node[main node] (7) [above right of=6] {};
        \node[main node] (8) [below right of=6] {};
        \node[main node] (9) [right of=8] {};
    
        \path[-]
        (6) edge node {} (7)
        (6) edge node {} (8)
        (8) edge node {} (9);
        }
    \end{tikzpicture} \\
     & 
     \begin{tikzpicture}[->,>= stealth,shorten >=1pt,auto,node distance=1.3cm,
            thick,main node/.style={circle,fill=blue!20,draw,minimum size=.3cm,inner sep=0pt]}]
        \raisebox{-2mm}{
        \node[main node] (2) {};
        \node[main node] (6) [ left of =2] {};
        \node[draw, fill=blue!20, shape=diamond, aspect=0.7, minimum height=0.4cm, inner sep=2.6pt] (5) [left of=6] {};
        \node[main node] (7) [above right of=2] {};
        \node[main node] (8) [below right of=2] {};
        \node[main node] (9) [right of=2] {};
    
        \path[-]
        (6) edge node {} (2)   
        (5) edge node {} (6)
        (2) edge node {} (7)
        (2) edge node {} (8)
        (2) edge node {} (9);
        }
    \end{tikzpicture}
     & \begin{tikzpicture}[->,>= stealth,shorten >=1pt,auto,node distance=1.3cm,
            thick,main node/.style={circle,fill=blue!20,draw,minimum size=.3cm,inner sep=0pt]}]
        \raisebox{-2mm}{
        \node[main node] (6) {};
        \node[draw, fill=blue!20, shape=diamond, aspect=0.7, minimum height=0.4cm, inner sep=2.6pt] (5) [left of=6] {};
        \node[main node] (7) [above right of=6]{};
        \node[main node] (8) [below right of=6]{};
        \node[main node] (9) [right of=6]{};
    
        \path[-] 
        (5) edge node {} (6)
        (6) edge node {} (7)
        (6) edge node {} (8)
        (6) edge node {} (9);
        }
    \end{tikzpicture} & 
    \begin{tikzpicture}[->,>= stealth,shorten >=1pt,auto,node distance=1.3cm,
            thick,main node/.style={circle,fill=blue!20,draw,minimum size=.3cm,inner sep=0pt]}]
        \raisebox{-2mm}{
        \node[draw, fill=blue!20, shape=diamond, aspect=0.7, minimum height=0.4cm, inner sep=2.6pt] (6) {};
        \node[main node] (7) [above right of=6] {};
        \node[main node] (8) [below right of=6] {};
        \node[main node] (10) [right of=6] {};
        \node[main node] (9) [right of=8] {};
    
        \path[-]
        (6) edge node {} (7)
        (6) edge node {} (8)
        (8) edge node {} (9)
        (6) edge node {} (10);
        }
    \end{tikzpicture} \\
    &&&\\
    \bottomrule
    \end{tabular}
    \label{tab:knife_distances}
\end{table}

\begin{table}[htbp]
    \centering
    \caption{Examples of networks with different distances to fork}
    \begin{tabular}{cccc}
    \toprule
    $d=0$ & $d=1/3$ & $d=1/2$ & $d=1$\\
    \midrule
    \begin{tikzpicture}[->,>= stealth,shorten >=1pt,auto,node distance=1.3cm,
            thick,main node/.style={circle,fill=blue!20,draw,minimum size=.3cm,inner sep=0pt]}]
        \raisebox{-2mm}{
        \node[main node] (2) {};
        \node[main node] (6) [above left of =2] {};
        \node[draw, fill=blue!20, shape=diamond, aspect=0.7, minimum height=0.4cm, inner sep=2.6pt] (5) [below left of=2] {};
         \node[main node] (7) [right of=2] {};
    
        \path[-]
        (6) edge node {} (2)   
        (5) edge node {} (2)
        (7) edge node {} (2);
        }
    \end{tikzpicture}
    & 
    \begin{tikzpicture}[->,>= stealth,shorten >=1pt,auto,node distance=1.3cm,
            thick,main node/.style={circle,fill=blue!20,draw,minimum size=.3cm,inner sep=0pt]}]
        \raisebox{-2mm}{
        \node[main node] (2) {};
        \node[main node] (6) [above left of =2] {};
        \node[draw, fill=blue!20, shape=diamond, aspect=0.7, minimum height=0.4cm, inner sep=2.6pt] (5) [below left of=2] {};
        \node[main node] (7) [right of=2] {};
        \node[main node] (8) [right of=6] {};
    
        \path[-]
        (6) edge node {} (2)   
        (5) edge node {} (2)
        (7) edge node {} (2)
        (6) edge node {} (8);
        }
    \end{tikzpicture}
    & 
    \begin{tikzpicture}[->,>= stealth,shorten >=1pt,auto,node distance=1.3cm,
            thick,main node/.style={circle,fill=blue!20,draw,minimum size=.3cm,inner sep=0pt]}]
        \raisebox{7.25mm}{
        \node[main node] (6) {};
        \node[draw, fill=blue!20, shape=diamond, aspect=0.7, minimum height=0.4cm, inner sep=2.6pt] (5) [left of=6] {};
        \node[main node] (9) [right of=6]{};
    
        \path[-] 
        (5) edge node {} (6)
        (6) edge node {} (9);
        }
    \end{tikzpicture}
    & 
    \begin{tikzpicture}[->,>= stealth,shorten >=1pt,auto,node distance=1.3cm,
            thick,main node/.style={circle,fill=blue!20,draw,minimum size=.3cm,inner sep=0pt]}]
        \raisebox{7.25mm}{
        \node[draw, fill=blue!20, shape=diamond, aspect=0.7, minimum height=0.4cm, inner sep=2.6pt] (6) {};
    
        \path[-];
        }
    \end{tikzpicture} \\
    & 
    \begin{tikzpicture}[->,>= stealth,shorten >=1pt,auto,node distance=1.3cm,
            thick,main node/.style={circle,fill=blue!20,draw,minimum size=.3cm,inner sep=0pt]}]
        \raisebox{-2mm}{
        \node[main node] (2) {};
        \node[main node] (6) [above left of =2] {};
        \node[draw, fill=blue!20, shape=diamond, aspect=0.7, minimum height=0.4cm, inner sep=2.6pt] (5) [below left of=2] {};
        \node[main node] (7) [right of=2] {};
        \node[main node] (8) [below of=7] {};
    
        \path[-]
        (6) edge node {} (2)   
        (5) edge node {} (2)
        (7) edge node {} (2)
        (7) edge node {} (8);
        }
    \end{tikzpicture}
    & 
    \begin{tikzpicture}[->,>= stealth,shorten >=1pt,auto,node distance=1.3cm,
            thick,main node/.style={circle,fill=blue!20,draw,minimum size=.3cm,inner sep=0pt]}]
        \raisebox{1.5mm}{
        \node[main node] (6) {};
        \node[draw, fill=blue!20, shape=diamond, aspect=0.7, minimum height=0.4cm, inner sep=2.6pt] (5) [left of=6] {};
        \node[main node] (7) [above right of=6]{};
        \node[main node] (8) [below right of=6]{};
        \node[main node] (9) [right of=6]{};
    
        \path[-] 
        (5) edge node {} (6)
        (6) edge node {} (7)
        (6) edge node {} (8)
        (6) edge node {} (9);
        }
    \end{tikzpicture}
     &  
     \begin{tikzpicture}[->,>= stealth,shorten >=1pt,auto,node distance=1.3cm,
            thick,main node/.style={circle,fill=blue!20,draw,minimum size=.3cm,inner sep=0pt]}]
        \raisebox{1.5mm}{
        \node[draw, fill=blue!20, shape=diamond, aspect=0.7, minimum height=0.4cm, inner sep=2.6pt] (6) {};
        \node[main node] (7) [above right of=6] {};
        \node[main node] (8) [below right of=6] {};
    
        \path[-]
        (6) edge node {} (7)
        (6) edge node {} (8);
        }
    \end{tikzpicture} \\
     & 
     \begin{tikzpicture}[->,>= stealth,shorten >=1pt,auto,node distance=1.3cm,
            thick,main node/.style={circle,fill=blue!20,draw,minimum size=.3cm,inner sep=0pt]}]
        \raisebox{-2mm}{
        \node[main node] (2) {};
        \node[main node] (6) [above left of =2] {};
        \node[draw, fill=blue!20, shape=diamond, aspect=0.7, minimum height=0.4cm, inner sep=2.6pt] (5) [below left of=2] {};
        \node[main node] (7) [right of=2] {};
        \node[main node] (8) [below of=7] {};
        \node[main node] (9) [right of=6] {};
    
        \path[-]
        (6) edge node {} (2)   
        (5) edge node {} (2)
        (7) edge node {} (2)
        (7) edge node {} (8)
        (6) edge node {} (9);
        }
    \end{tikzpicture}
     & 
     \begin{tikzpicture}[->,>= stealth,shorten >=1pt,auto,node distance=1.3cm,
            thick,main node/.style={circle,fill=blue!20,draw,minimum size=.3cm,inner sep=0pt]}]
        \raisebox{1.5mm}{
        \node[main node] (6) {};
        \node[draw, fill=blue!20, shape=diamond, aspect=0.7, minimum height=0.4cm, inner sep=2.6pt] (5) [left of=6] {};
        \node[main node] (7) [above right of=6]{};
        \node[main node] (8) [below right of=6]{};
        \node[main node] (9) [right of=6]{};
    
        \path[-] 
        (5) edge node {} (6)
        (6) edge node {} (7)
        (6) edge node {} (8)
        (6) edge node {} (9)
        (7) edge node {} (9);
        }
    \end{tikzpicture} & 
    \begin{tikzpicture}[->,>= stealth,shorten >=1pt,auto,node distance=1.3cm,
            thick,main node/.style={circle,fill=blue!20,draw,minimum size=.3cm,inner sep=0pt]}]
        \raisebox{1.5mm}{
        \node[draw, fill=blue!20, shape=diamond, aspect=0.7, minimum height=0.4cm, inner sep=2.6pt] (6) {};
        \node[main node] (7) [above right of=6] {};
        \node[main node] (8) [below right of=6] {};
        \node[main node] (9) [right of=8] {};
    
        \path[-]
        (6) edge node {} (7)
        (6) edge node {} (8)
        (8) edge node {} (9);
        }
    \end{tikzpicture} \\
     & 
     \begin{tikzpicture}[->,>= stealth,shorten >=1pt,auto,node distance=1.3cm,
            thick,main node/.style={circle,fill=blue!20,draw,minimum size=.3cm,inner sep=0pt]}]
        \raisebox{-2mm}{
        \node[main node] (2) {};
        \node[main node] (6) [above left of =2] {};
        \node[draw, fill=blue!20, shape=diamond, aspect=0.7, minimum height=0.4cm, inner sep=2.6pt] (5) [below left of=2] {};
        \node[main node] (7) [right of=2] {};
        \node[main node] (8) [below of=7] {};
        \node[main node] (9) [right of=6] {};
    
        \path[-]
        (6) edge node {} (2)   
        (5) edge node {} (2)
        (7) edge node {} (2)
        (7) edge node {} (8)
        (6) edge node {} (9)
        (7) edge node {} (9);
        }
    \end{tikzpicture}
     & 
     \begin{tikzpicture}[->,>= stealth,shorten >=1pt,auto,node distance=1.3cm,
            thick,main node/.style={circle,fill=blue!20,draw,minimum size=.3cm,inner sep=0pt]}]
        \raisebox{1.5mm}{
        \node[main node] (6) {};
        \node[draw, fill=blue!20, shape=diamond, aspect=0.7, minimum height=0.4cm, inner sep=2.6pt] (5) [left of=6] {};
        \node[main node] (7) [above right of=6]{};
        \node[main node] (8) [below right of=6]{};
        \node[main node] (9) [right of=6]{};
    
        \path[-] 
        (5) edge node {} (6)
        (6) edge node {} (7)
        (6) edge node {} (8)
        (6) edge node {} (9)
        (7) edge node {} (9)
        (8) edge node {} (9);
        }
    \end{tikzpicture} & 
    \begin{tikzpicture}[->,>= stealth,shorten >=1pt,auto,node distance=1.3cm,
            thick,main node/.style={circle,fill=blue!20,draw,minimum size=.3cm,inner sep=0pt]}]
        \raisebox{1.5mm}{
        \node[draw, fill=blue!20, shape=diamond, aspect=0.7, minimum height=0.4cm, inner sep=2.6pt] (6) {};
        \node[main node] (7) [above right of=6] {};
        \node[main node] (8) [below right of=6] {};
        \node[main node] (10) [right of=6] {};
        \node[main node] (9) [right of=8] {};
    
        \path[-]
        (6) edge node {} (7)
        (6) edge node {} (8)
        (8) edge node {} (9)
        (6) edge node {} (10);
        }
    \end{tikzpicture} \\
    &&&\\
    \bottomrule
    \end{tabular}
    \label{tab:fork_distances}
\end{table}

\begin{table}[htbp]
    \centering
    \caption{Examples of networks with different distances to spoon}
    \begin{tabular}{cccc}
    \toprule
    $d=0$ & $d=1/3$ & $d=1/2$ & $d=1$\\
    \midrule
    \begin{tikzpicture}[->,>= stealth,shorten >=1pt,auto,node distance=1.3cm,
            thick,main node/.style={circle,fill=blue!20,draw,minimum size=.3cm,inner sep=0pt]}]
        \raisebox{-2mm}{
        \node[main node] (2) {};
        \node[main node] (6) [above left of =2] {};
        \node[draw, fill=blue!20, shape=diamond, aspect=0.7, minimum height=0.4cm, inner sep=2.6pt] (5) [below left of=2] {};
         \node[main node] (7) [right of=2] {};
    
        \path[-]
        (6) edge node {} (2)   
        (5) edge node {} (2)
        (7) edge node {} (2)
        (5) edge node {} (6);
        }
    \end{tikzpicture}
    & 
    \begin{tikzpicture}[->,>= stealth,shorten >=1pt,auto,node distance=1.3cm,
            thick,main node/.style={circle,fill=blue!20,draw,minimum size=.3cm,inner sep=0pt]}]
        \raisebox{-2mm}{
        \node[main node] (2) {};
        \node[main node] (6) [above left of =2] {};
        \node[draw, fill=blue!20, shape=diamond, aspect=0.7, minimum height=0.4cm, inner sep=2.6pt] (5) [below left of=2] {};
        \node[main node] (7) [right of=2] {};
        \node[main node] (8) [above right of=7] {};
    
        \path[-]
        (6) edge node {} (2)   
        (5) edge node {} (2)
        (7) edge node {} (2)
        (5) edge node {} (6)
        (7) edge node {} (8);
        }
    \end{tikzpicture}
    & 
    \begin{tikzpicture}[->,>= stealth,shorten >=1pt,auto,node distance=1.3cm,
            thick,main node/.style={circle,fill=blue!20,draw,minimum size=.3cm,inner sep=0pt]}]
        \raisebox{-2mm}{
        \node[main node] (2) {};
        \node[main node] (6) [above left of =2] {};
        \node[draw, fill=blue!20, shape=diamond, aspect=0.7, minimum height=0.4cm, inner sep=2.6pt] (5) [below left of=2] {};
    
        \path[-]
        (6) edge node {} (2)   
        (5) edge node {} (2)
        (5) edge node {} (6);
        }
    \end{tikzpicture}    & 
    \begin{tikzpicture}[->,>= stealth,shorten >=1pt,auto,node distance=1.3cm,
            thick,main node/.style={circle,fill=blue!20,draw,minimum size=.3cm,inner sep=0pt]}]
        \raisebox{7.25mm}{
        \node[draw, fill=blue!20, shape=diamond, aspect=0.7, minimum height=0.4cm, inner sep=2.6pt] (6) {};
    
        \path[-];
        }
    \end{tikzpicture} \\
    & 
    \begin{tikzpicture}[->,>= stealth,shorten >=1pt,auto,node distance=1.3cm,
            thick,main node/.style={circle,fill=blue!20,draw,minimum size=.3cm,inner sep=0pt]}]
        \raisebox{-2mm}{
        \node[main node] (2) {};
        \node[main node] (6) [above left of =2] {};
        \node[draw, fill=blue!20, shape=diamond, aspect=0.7, minimum height=0.4cm, inner sep=2.6pt] (5) [below left of=2] {};
        \node[main node] (7) [right of=2] {};
        \node[main node] (8) [above right of=7] {};
         \node[main node] (9) [below right of=7] {};
    
        \path[-]
        (6) edge node {} (2)   
        (5) edge node {} (2)
        (7) edge node {} (2)
        (5) edge node {} (6)
        (7) edge node {} (8)
        (7) edge node {} (9);
        }
    \end{tikzpicture}
    & 
   \begin{tikzpicture}[->,>= stealth,shorten >=1pt,auto,node distance=1.3cm,
            thick,main node/.style={circle,fill=blue!20,draw,minimum size=.3cm,inner sep=0pt]}]
        \raisebox{-2mm}{
        \node[main node] (2) {};
        \node[main node] (6) [above left of =2] {};
        \node[draw, fill=blue!20, shape=diamond, aspect=0.7, minimum height=0.4cm, inner sep=2.6pt] (5) [below left of=2] {};
        \node[main node] (7) [above right of =2] {};
        \node[main node] (8) [below right of =2] {};
    
        \path[-]
        (6) edge node {} (2)   
        (5) edge node {} (2)
        (5) edge node {} (6)
        (7) edge node {} (2)
        (8) edge node {} (2);
        }
    \end{tikzpicture}
     &  
     \begin{tikzpicture}[->,>= stealth,shorten >=1pt,auto,node distance=1.3cm,
            thick,main node/.style={circle,fill=blue!20,draw,minimum size=.3cm,inner sep=0pt]}]
        \raisebox{-1.5mm}{
        \node[draw, fill=blue!20, shape=diamond, aspect=0.7, minimum height=0.4cm, inner sep=2.6pt] (6) {};
        \node[main node] (7) [above right of=6] {};
        \node[main node] (8) [below right of=6] {};
    
        \path[-]
        (6) edge node {} (7)
        (6) edge node {} (8);
        }
    \end{tikzpicture} \\
     & 
     \begin{tikzpicture}[->,>= stealth,shorten >=1pt,auto,node distance=1.3cm,
            thick,main node/.style={circle,fill=blue!20,draw,minimum size=.3cm,inner sep=0pt]}]
        \raisebox{-2mm}{
        \node[main node] (2) {};
        \node[main node] (6) [above left of =2] {};
        \node[draw, fill=blue!20, shape=diamond, aspect=0.7, minimum height=0.4cm, inner sep=2.6pt] (5) [below left of=2] {};
        \node[main node] (7) [right of=2] {};
        \node[main node] (8) [above right of=7] {};
         \node[main node] (9) [below right of=7] {};
    
        \path[-]
        (6) edge node {} (2)   
        (5) edge node {} (2)
        (7) edge node {} (2)
        (5) edge node {} (6)
        (7) edge node {} (8)
        (7) edge node {} (9)
        (8) edge node {} (9);
        }
    \end{tikzpicture}     & 
    \begin{tikzpicture}[->,>= stealth,shorten >=1pt,auto,node distance=1.3cm,
            thick,main node/.style={circle,fill=blue!20,draw,minimum size=.3cm,inner sep=0pt]}]
        \raisebox{-2mm}{
        \node[main node] (2) {};
        \node[main node] (6) [above left of =2] {};
        \node[draw, fill=blue!20, shape=diamond, aspect=0.7, minimum height=0.4cm, inner sep=2.6pt] (5) [below left of=2] {};
        \node[main node] (7) [above right of =2] {};
        \node[main node] (8) [below right of =2] {};
    
        \path[-]
        (6) edge node {} (2)   
        (5) edge node {} (2)
        (5) edge node {} (6)
        (7) edge node {} (2)
        (8) edge node {} (2)
        (8) edge node {} (7);
        }
    \end{tikzpicture} & 
    \begin{tikzpicture}[->,>= stealth,shorten >=1pt,auto,node distance=1.3cm,
            thick,main node/.style={circle,fill=blue!20,draw,minimum size=.3cm,inner sep=0pt]}]
        \raisebox{-1.5mm}{
        \node[draw, fill=blue!20, shape=diamond, aspect=0.7, minimum height=0.4cm, inner sep=2.6pt] (6) {};
        \node[main node] (7) [above right of=6] {};
        \node[main node] (8) [below right of=6] {};
        \node[main node] (9) [right of=8] {};
    
        \path[-]
        (6) edge node {} (7)
        (6) edge node {} (8)
        (8) edge node {} (9);
        }
    \end{tikzpicture} \\
     & 
    \begin{tikzpicture}[->,>= stealth,shorten >=1pt,auto,node distance=1.3cm,
            thick,main node/.style={circle,fill=blue!20,draw,minimum size=.3cm,inner sep=0pt]}]
        \raisebox{-2mm}{
        \node[main node] (2) {};
        \node[main node] (6) [above left of =2] {};
        \node[draw, fill=blue!20, shape=diamond, aspect=0.7, minimum height=0.4cm, inner sep=2.6pt] (5) [below left of=2] {};
        \node[main node] (7) [right of=2] {};
        \node[main node] (8) [above right of=7] {};
        \node[main node] (9) [below right of=7] {};
        \node[main node] (10) [right of=7] {};
    
        \path[-]
        (6) edge node {} (2)   
        (5) edge node {} (2)
        (7) edge node {} (2)
        (5) edge node {} (6)
        (7) edge node {} (8)
        (7) edge node {} (9)
        (7) edge node {} (10);
        }
    \end{tikzpicture} 
     & 
     \begin{tikzpicture}[->,>= stealth,shorten >=1pt,auto,node distance=1.3cm,
            thick,main node/.style={circle,fill=blue!20,draw,minimum size=.3cm,inner sep=0pt]}]
        \raisebox{-2mm}{
        \node[main node] (2) {};
        \node[main node] (6) [above left of =2] {};
        \node[draw, fill=blue!20, shape=diamond, aspect=0.7, minimum height=0.4cm, inner sep=2.6pt] (5) [below left of=2] {};
        \node[main node] (7) [above right of =2] {};
        \node[main node] (8) [below right of =2] {};
        \node[main node] (9) [right of =2] {};
    
        \path[-]
        (6) edge node {} (2)   
        (5) edge node {} (2)
        (5) edge node {} (6)
        (7) edge node {} (2)
        (8) edge node {} (2)
        (9) edge node {} (2);
        }
    \end{tikzpicture} & 
    \begin{tikzpicture}[->,>= stealth,shorten >=1pt,auto,node distance=1.3cm,
            thick,main node/.style={circle,fill=blue!20,draw,minimum size=.3cm,inner sep=0pt]}]
        \raisebox{-2mm}{
        \node[draw, fill=blue!20, shape=diamond, aspect=0.7, minimum height=0.3cm, inner sep=2.6pt] (6) {};
        \node[main node] (7) [above right of=6] {};
        \node[main node] (8) [below right of=6] {};
        \node[main node] (10) [right of=6] {};
        \node[main node] (9) [right of=8] {};
    
        \path[-]
        (6) edge node {} (7)
        (6) edge node {} (8)
        (8) edge node {} (9)
        (6) edge node {} (10);
        }
    \end{tikzpicture} \\
    &&&\\
    \bottomrule
    \end{tabular}
    \label{tab:spoon_distances}
\end{table}

\end{document}